\documentclass{LMCS}

\def\doi{8 (1:31) 2012}
\lmcsheading%
{\doi}
{1--29}
{}
{}
{Oct.~28, 2011}
{Mar.~28, 2012}
{}

\pdfoutput=1 

\usepackage{enumerate}
\usepackage{hyperref}
\usepackage{proof}
\usepackage{amssymb}
\usepackage{amsmath}
\usepackage{latexsym}
\usepackage[all]{xy}
\usepackage{graphicx,ps4pdf}
\PSforPDF{\usepackage{pstricks}}

\begin{document}


\title[A Reduction-Preserving Completion for Proving Confluence]{
A Reduction-Preserving Completion for Proving Confluence of 
Non-Terminating Term Rewriting Systems\rsuper*}


\author[]{Takahito Aoto}
\address{RIEC, Tohoku University\\
2-1-1 Katahira, Aoba-ku, Sendai, Miyagi, 980-8577, Japan}
\email{\{aoto,toyama\}@nue.riec.tohoku.ac.jp}


\author[]{Yoshihito Toyama}
\address{\vskip-6 pt}



\titlecomment{{\lsuper*}This is a revised and extended version of the paper:
Takahito Aoto and Yoshihito Toyama,
A Reduction-Preserving Completion for Proving Confluence of 
Non-Terminating Term Rewriting Systems,
in \textit{Proceedings of the 22nd 
International Conference on 
Rewriting Techniques and Applications},
LIPIcs, Vol.10, Schloss Dagstuhl -- Leibniz-Zentrum fuer Informatik,
pp.91-106, 2011.}



\begin{abstract}
\noindent 
We give a method to prove confluence of term rewriting systems that
contain non-terminating rewrite rules such as commutativity and
associativity.  Usually, confluence of term rewriting systems
containing such rules is proved by treating them as equational term
rewriting systems and considering $E$-critical pairs and/or termination
modulo $E$. In contrast, our method is based solely on usual critical
pairs and it also (partially) works even if 
the system is not terminating modulo $E$.
We first present confluence criteria for
term rewriting systems whose rewrite rules can be partitioned into
a terminating part and a possibly non-terminating part.  We then give a
reduction-preserving completion procedure so that the applicability of
the criteria is enhanced.  In contrast to the well-known Knuth-Bendix
completion procedure which preserves the equivalence relation of the
system, our completion procedure preserves the reduction relation of
the system, by which confluence of the original system
is inferred from that of the completed system.
\end{abstract}

\maketitle


\newcommand{\parto}{\mathrel{\rightarrow\!\!\!\!\!\!\!\!\!\;{+}\!\!\!\!{+}\,}}
\newcommand{\pargets}{\mathrel{\leftarrow\!\!\!\!\!\!\!\!\!\;{+}\!\!\!\!{+}\,}}
\newcommand{\parleftrightarrow}{\mathrel{\leftrightarrow\!\!\!\!\!\!\!\!\!\;{+}\!\!\!\!{+}\,}}
\newcommand{\gggg}{\gg_\textrm{\rm mul}}
\newcommand{\doubleheadsrightarrow}{\mathrel{\to\!\!\!\!\!\to}}

\newcommand{\cv}[2]{\curlyvee #1 \cup \curlyvee #2}

\section{Introduction}

Confluence is one of the most important properties
of term rewriting systems (TRSs for short)
and hence many efforts have been spent on developing
techniques to prove this property \cite{BaaderFandNipkowT:TR,ToyRTA}.
One of the classes of TRSs for which
many known confluence proving methods are not effective
is the class of TRSs 
containing associativity and commutativity rules (AC-rules).
Such TRSs are non-terminating by the existence of AC-rules
(more precisely, the commutativity rule is self-looping and 
associativity rules are looping under the presence of 
the commutativity rule)
and hence the Knuth-Bendix criterion 
(i.e.\ terminating TRSs are confluent
iff all critical pairs are joinable)
does not apply.
Furthermore, confluence criteria regardless of termination
based on critical pairs often do not apply either.

A well-known approach to deal with TRSs containing AC-rules
is to deal with them as equational term rewriting systems 
\cite{Hue80,JK86,PS81}.
In this approach,
non-terminating rules such as AC-rules
are treated exceptionally as an equational subsystem $\mathcal{E}$.
Then the 
confluence of equational term rewriting system $\langle \mathcal{R}, \mathcal{E} \rangle$
is obtained if 
$\mathcal{R}$ is terminating modulo $\mathcal{E}$ 
\cite{Hue80,JK86,PS81}
and either the $\mathcal{E}$-critical pairs of $\mathcal{R}$
satisfy certain conditions \cite{JK86,PS81}
or $\mathcal{R}$ is left-linear and 
the $\mathcal{E}/\mathcal{R}$-critical pairs 
satisfy a certain condition \cite{Hue80}.
This approach, however, only works if $\mathcal{R}$ is terminating modulo $\mathcal{E}$.
Furthermore, the computation of $\mathcal{E}$-critical pairs requires
a finite and complete $\mathcal{E}$-unification algorithm
which depends on $\mathcal{E}$.

In this paper, we give a method to prove confluence of TRSs 
that contain non-terminating rewrite rules such as AC-rules.
In contrast to the traditional approach described above,
our method is based solely on usual critical
pairs and it also (partially) works even if 
the system is not terminating modulo $\mathcal{E}$.
Thus the implementation of the method requires few special ingredients
and the method is easily integrated
into confluence provers and combined
with other confluence proving methods.

Let us explain the idea of our approach via concrete examples.

\begin{exa}
\label{exp:com-assoc}
Let $\mathcal{R}_1$ be the TRS consisting of 
the commutativity rule and an associativity rule.
\[
\mathcal{R}_1
= \left\{
\begin{array}{clcl}
(C) & \mathsf{+}(x,y)             &\to&  \mathsf{+}(y,x) \\
(A) & \mathsf{+}(\mathsf{+}(x,y),z) &\to&  \mathsf{+}(x,\mathsf{+}(y,z))\\
\end{array}
\right\}.
\]
This TRS is non-terminating and 
many known critical pair conditions for left-linear TRSs
do not apply.
However, ${\mathcal{R}_1}$ is 
confluent, 
i.e.\ 
$s \stackrel{*}{\to}_{\mathcal{R}_1} t_0$
and 
$s \stackrel{*}{\to}_{\mathcal{R}_1} t_1$
imply
$t_0 \stackrel{*}{\to}_{\mathcal{R}_1} u$
and 
$t_1 \stackrel{*}{\to}_{\mathcal{R}_1} u$
for some $u$.
Here $s \stackrel{*}{\to}_{\mathcal{R}_1} t$
denotes that $s$ rewrites to $t$ in arbitrary many rewrite steps.
One way to prove this is by observing that 
${\mathcal{R}_1}$ is reversible,
i.e.\ $s \stackrel{*}{\to}_{\mathcal{R}_1} t$
iff $t \stackrel{*}{\to}_{\mathcal{R}_1} s$.
This holds because for any rewrite rule $l \to r \in \mathcal{R}_1$
we have $r \stackrel{*}{\to}_{\mathcal{R}_1} l$:
for the $(C)$-rule, this holds obviously,  and
for the $(A)$-rule, this holds because
$\mathsf{+}(x,\mathsf{+}(y,z)) 
\stackrel{*}{\to}_{\mathcal{R}_1}
\mathsf{+}(\mathsf{+}(x,y),z)$ viz.\
\[
\begin{array}{lcl}
\mathsf{+}(x,\mathsf{+}(y,z)) 
&\to_C& \mathsf{+}(\mathsf{+}(y,z),x)\\
&\to_A& \mathsf{+}(y,\mathsf{+}(z,x))\\
&\to_C& \mathsf{+}(\mathsf{+}(z,x),y)\\
&\to_A& \mathsf{+}(z,\mathsf{+}(x,y))\\
&\to_C& \mathsf{+}(\mathsf{+}(x,y),z).\\
\end{array}
\]
Thus
for any $s_1 \to_{\mathcal{R}_1} s_2 \to_{\mathcal{R}_1} \cdots  \to_{\mathcal{R}_1} s_n$,
we have
$s_n \stackrel{*}{\to}_{\mathcal{R}_1} s_{n-1} 
\stackrel{*}{\to}_{\mathcal{R}_1} \cdots 
\stackrel{*}{\to}_{\mathcal{R}_1} s_1$.
Hence 
$s \stackrel{*}{\to}_{\mathcal{R}_1} t_0$
and 
$s \stackrel{*}{\to}_{\mathcal{R}_1} t_1$
imply 
$t_0 \stackrel{*}{\to}_{\mathcal{R}_1} s$
and 
$t_1 \stackrel{*}{\to}_{\mathcal{R}_1} s$.
\end{exa}

\begin{exa}
\label{exp:plus-com-assoc}
Next we consider the TRS $\mathcal{R}_2$,
which extends $\mathcal{R}_1$ slightly.
\[
\mathcal{R}_2
= \left\{
\begin{array}{clclclcl}
(\mathsf{add}_1) & \mathsf{+}(\mathsf{0},y)    &\to&  y &\quad\\
(\mathsf{add}_2) & \mathsf{+}(\mathsf{s}(x),y) &\to&  \mathsf{s}(\mathsf{+}(x,y)) \\
(C) & \mathsf{+}(x,y)             &\to&  \mathsf{+}(y,x) &\quad\\
(A) & \mathsf{+}(\mathsf{+}(x,y),z) &\to&  \mathsf{+}(x,\mathsf{+}(y,z))\\
\end{array}
\right\}.
\]
$\mathcal{R}_2$ is a TRS consisting of 
rules for addition of natural numbers 
denoted by $\mathsf{0},\mathsf{s}(\mathsf{0}),
\mathsf{s}(\mathsf{s}(\mathsf{0})),\ldots$
and AC-rules for plus.
The TRS $\mathcal{R}_2$ is again non-terminating and many
known critical pair conditions for left-linear TRSs
also do not apply.
However, $\mathcal{R}_2$ is confluent.
This can be explained like this.
Since $\mathsf{+}(y,\mathsf{0})  \to_{\mathcal{R}_2}
 \mathsf{+}(\mathsf{0},y)  \to_{\mathcal{R}_2}  y$
and $\mathsf{+}(y,\mathsf{s}(x)) 
\to_{\mathcal{R}_2} \mathsf{+}(\mathsf{s}(x),y) 
\to_{\mathcal{R}_2} \mathsf{s}(\mathsf{+}(x,y))$,
together with $(\mathsf{add}_1),(\mathsf{add}_2)$-rules,
all occurrences of the symbol $\mathsf{0}$ in a term can be eliminated and
all occurrences of the symbol $\mathsf{s}$ can be moved to the top
of the term.
Hence, for any term $t$, we have 
$t \stackrel{*}{\to}_{\mathcal{R}_2} 
\mathsf{s}^k(+(\cdots x \strut \cdots))$
where 
$k$ is the number of occurrences of the symbol $\mathsf{s}$ in the term $t$
and the part ``$+(\cdots x \strut \cdots)$'' denotes
the addition of all variables contained in the term $t$.
Thus for any $u_1,u_2$ such that
$t \stackrel{*}{\to}_{\mathcal{R}_2} u_1$
and 
$t \stackrel{*}{\to}_{\mathcal{R}_2} u_2$,
we have 
$u_1 \stackrel{*}{\to}_{\mathcal{R}_2} \mathsf{s}^k(+(\cdots x \strut \cdots))$
and 
$u_2 \stackrel{*}{\to}_{\mathcal{R}_2} \mathsf{s}^k(+(\cdots x \strut \cdots))$.
It remains to use the reversibility of AC-rules (i.e.\ $\mathcal{R}_1$)
to join two terms of the form $\mathsf{s}^k(+(\cdots x \strut \cdots))$
because they are equivalent modulo associativity and commutativity.
\end{exa}

A key point of this method is that,
in addition to rewrite rules of $\mathcal{R}_2$,
we considered auxiliary rewrite rules 
$\mathsf{add}_3:  \mathsf{+}(y,\mathsf{0}) \to  y$
and $\mathsf{add}_4: \mathsf{+}(y,\mathsf{s}(x)) 
\to \mathsf{s}(\mathsf{+}(x,y))$.
In our method, such rewrite rules
are added via a reduction-preserving completion procedure.
In contrast to the well-known Knuth-Bendix
completion procedure which preserves the equivalence relation of the
system, our completion procedure preserves the reduction relation of
the system, by which confluence of the original system
is inferred from that of the completed system.
We note that the Knuth-Bendix completion procedure 
for equational term rewriting systems
was initiated by \cite{LB77c,LB77a,LB77b,PS81}
and is generalized in \cite{BD89,JK86}.
Since the Knuth-Bendix completion procedure 
needs to preserve equivalence relation
but not necessarily reduction relation,
much flexibilities are allowed for 
the Knuth-Bendix completion procedure
compared to our reduction-preserving completion procedure.

The contribution of this paper is summarized as follows:
\begin{enumerate}[(1)]
\item new abstract criterion for the 
property \textit{Church-Rosser modulo} (Theorem \ref{thm:ARS}),
\item new confluence criteria (Theorems \ref{thm:linear} and \ref{thm:PCP}),
\item reduction-preserving completion 
for proving confluence   and
\item implementation and experiments for these techniques.
\end{enumerate}
This paper is a revised and extended version of \cite{Rcomp}.
Compared to \cite{Rcomp},
Theorems \ref{thm:ARS}, \ref{thm:linear} and \ref{thm:PCP}
are new---these extend the results in \cite{Rcomp}
which are adapted as 
Corollaries\footnotemark \ref{cor:ARS-I}, \ref{cor:linear} and \ref{cor:PCP}
respectively in the present paper.

\footnotetext{
Corollaries \ref{cor:ARS-I} and \ref{cor:linear}
are incorrectly claimed to be original
in \cite{Rcomp}.}

The rest of the paper is organized as follows.
We first present a criterion for Church-Rosser modulo
in an abstract setting (Section 2).  
Then based on this abstract criterion,
we present confluence criteria for
TRSs whose rewrite rules can be partitioned into
a terminating part and a possibly non-terminating part (Section 3).  
We then give a
reduction-preserving completion procedure so that the applicability of
the criteria is enhanced (Section 4).  
Finally we report on our implementation
and results of experiments (Section 5).

\newcommand{\vdashv}{\mathrel{\lower .22ex \hbox{$\vdash\!\dashv$}}}
\newcommand{\eqvdashv}{\mathrel{\ooalign{{$\vdashv$}\crcr{$\stackrel{=}{\phantom{\to}}$}}}}
\newcommand{\plsvdashv}{\mathrel{\ooalign{{$\vdashv$}\crcr{$\stackrel{+}{\phantom{\to}}$}}}}
\newcommand{\astvdashv}{\mathrel{\ooalign{{$\vdashv$}\crcr{$\stackrel{*}{\phantom{\to}}$}}}}
\newcommand{\parvdashv}{\mathrel{\vdashv}\!\!\!\!\!\!\!\!\!\;{+}\!\!\!\!{+}\,}

\section{Abstract criterion for Church-Rosser modulo}

In this section,
after providing some preliminaries (subsection 1),
we present a criterion for Church-Rosser modulo
an equivalence relation and
present some corollaries of the criterion that have been
appeared in the literature (subsection 2).
Then we compare our criterion and other
abstract criteria for Church-Rosser modulo
an equivalence relation (subsection 3).

\subsection{Preliminaries}

In this subsection, we fix
some notions and notations on relations
that will be used throughout the paper.

Let $\to$ be a relation on a set $A$.
The \textit{inverse} of $\to$ is denoted by $\gets$.
The \textit{reflexive closure} (the \textit{symmetric closure}, 
the \textit{transitive closure}, 
the \textit{reflexive and transitive closure},
the \textit{equivalence closure})
of $\to$ is denoted by $\stackrel{=}{\to}$
(${\leftrightarrow}$, 
$\stackrel{+}{\to}$, 
$\stackrel{*}{\to}$
$\stackrel{*}{\leftrightarrow}$, respectively).
We will also use $\to_0,\to_1,\ldots,\Rightarrow,\leadsto,\blacktriangleright,\ldots$ 
for binary relations, $\vdashv,\bowtie,\ldots$ for symmetric relations
and $\sim,\ldots$ for equivalence relations.
Closures for such relations are written in the similar way.

The union of two
relations $\to$ and $\Rightarrow$ is written as $\to \cup \Rightarrow$.
For any two relations $\to$ and $\Rightarrow$,
we write ${\to} \subseteq {\Rightarrow}$
if $a \to b$ implies $a \Rightarrow b$ for any $a,b$.
The composition of relations $\to$ and $\Rightarrow$ is written as $\to
\circ \Rightarrow$.  
For a (possibly infinite) number of indexed
relations $(\to_\alpha)_{\alpha \in I}$ where $I$ is a set of indexes,
$\bigcup_{\alpha \in I} \to_\alpha$ is written as $\to_I$.  We will
identify element and singleton set in this notation, i.e.\ ${\to}_{\{
  \alpha \}} = {\to}_\alpha$.
Closures for such relations are written in the similar way.

A relation $\to$ is \textit{well-founded} 
if
there exists no infinite descending chain $a_0 \to a_1 \to \cdots$.  
The relation $\to$ is said to be \textit{confluent}
if ${\stackrel{*}{\gets}} \circ {\stackrel{*}{\to}}
\subseteq {\stackrel{*}{\to}} \circ {\stackrel{*}{\gets}}$ holds.
The relation $\to$ is said to be \textit{Church-Rosser modulo
an equivalence relation $\sim$}
(\textit{CRM} in short)
if ${\stackrel{*}{\bowtie}}
\subseteq {\stackrel{*}{\to}}
 \circ {\sim} \circ {\stackrel{*}{\gets}}$ holds,
where ${\bowtie} =  {\leftrightarrow \cup \sim}$.

\subsection{Abstract criterion for Church-Rosser modulo}

In this subsection, we give a new criterion
for Church-Rosser modulo an equivalence relation.

In what follows, we consider 
a (strict) partial order $\succ$ on
the set $I$ of indexes.
Let $I$ be the set of indexes.
For a set $J \subseteq I$ of indexes,
the set $\{ \beta \in I \mid \exists \alpha \in J.~ \beta \prec \alpha \}$
is written as $\curlyvee J$.
If $J = \{ \alpha \}$, we write $\curlyvee \alpha$
instead of $\curlyvee \{ \alpha \}$.
We assume
that $\curlyvee$ associates stronger than
$\cup$ 
i.e.\
$\curlyvee I_1 \cup \curlyvee I_2 = (\curlyvee I_1) \cup (\curlyvee I_2)$.
The next lemma, which is the basis of our abstract criterion 
for Church-Rosser modulo, is 
obtained by induction on the set of indexes
w.r.t.\ the well-founded order $\succ$.

\begin{lem}
\label{lem:ARS}
Let $I$ be a set of indexes equipped
with a well-founded order $\succ$.
Let 
$\vdashv_{\alpha},\to_{\alpha}$ be relations on a set $A$
such that $\vdashv_{\alpha}$ is symmetric
for each $\alpha \in I$.
Let 
${\Rightarrow}_\alpha = {\vdashv_\alpha \cup \to}_\alpha$
for each $\alpha \in I$.
Suppose 
(i) ${\gets}_{\alpha} \circ {\to}_{\beta}
\subseteq 
{\stackrel{*}{\Leftrightarrow}}_{\cv{\alpha}{\beta}}$
and 
(ii) ${\vdashv}_{\alpha} \circ {\to}_{\beta}
\subseteq 
{\stackrel{*}{\Leftrightarrow}}_{\cv{\alpha}{\beta}}$.
Then ${\to}_I$ is Church-Rosser modulo ${\astvdashv}_I$.
\end{lem}

\proof
For each sequence 
$a_0 
\Leftrightarrow_{\alpha_0}
a_1
\Leftrightarrow_{\alpha_1}
\cdots
\Leftrightarrow_{\alpha_{n-1}}
a_n$,
let its weight be the multiset consisting of the 
indexes of the each steps 
i.e.\ 
$\{  \alpha_0, \alpha_1, \ldots, \alpha_{n-1}  \}$.
Let $\gg$ be the multiset extension of the well-founded order $\succ$.
We show by well-founded induction on the weight of the sequence
w.r.t.\ $\gg$ that 
for any sequence $a_0 \stackrel{*}{\Leftrightarrow}_I a_n$ 
there exists a sequence
$a_0 \stackrel{*}{\to}_I \circ \astvdashv_I \circ \stackrel{*}{\gets}_I a_n$.
\begin{enumerate}[(1)]
\item Suppose there exists $k$ such that
$a_{k-1} \gets_{\alpha} a_{k}  \to_{\beta} a_{k+1}$.
Then by assumption (i), 
there exists a sequence 
$a_{k-1} \stackrel{*}{\Leftrightarrow}_{\cv{\alpha}{\beta}}
a_{k+1}$.
Thus we have a  sequence
$a_0 \stackrel{*}{\Leftrightarrow}_I a_{k-1}
\stackrel{*}{\Leftrightarrow}_{\cv{\alpha}{\beta}}
a_{k+1}
\stackrel{*}{\Leftrightarrow}_I a_n$.
Since this new  sequence has a weight
less than the original  sequence
$a_0 \stackrel{*}{\Leftrightarrow}_I a_n$,
it follows that there exists a  sequence
$a_0 \stackrel{*}{\to}_{I} \circ \astvdashv_{I}
\circ \stackrel{*}{\gets}_{I} a_n$
by the induction hypothesis.


\item Suppose that there exists $k$ such that
$a_{k-1} \vdashv_{\alpha} a_{k}  \to_{\beta} a_{k+1}$.
Then by assumption (ii), 
there exists a  sequence 
$a_{k-1} \stackrel{*}{\Leftrightarrow}_{\cv{\alpha}{\beta}} a_{k+1}$.
Thus, it follows that there exists a  sequence
$a_0 \stackrel{*}{\to}_{I} \circ \astvdashv_{I}
\circ \stackrel{*}{\gets}_{I} a_n$
by the induction hypothesis
as in the previous case.


\item Suppose that there exists $k$ such that
$a_{k-1}  \gets_{\alpha} a_{k} \vdashv_{\beta} a_{k+1}$.
Then one can show that
there exists a  sequence
$a_0 \stackrel{*}{\to}_{I} \circ \astvdashv_{I}
\circ \stackrel{*}{\gets}_{I} a_n$
in the same way as the case (2).

\item It remains to treat the case
that 
($\alpha$) there exists no $k$ such that $a_{k-1} \gets_I a_k \to_I a_{k+1}$,
($\beta$) there exists no $k$ such that $a_{k-1} \vdashv_I a_k \to_I a_{k+1}$ and
($\gamma$) there exists no $k$ such that $a_{k-1} \gets_I a_k \vdashv_I a_{k+1}$.
We show by induction on the length of 
$a_0 \stackrel{*}{\Leftrightarrow} a_n$
that this  sequence has the form
$a_0 \stackrel{*}{\to}_{I} \circ  \astvdashv_{I} \circ
\stackrel{*}{\gets}_{I}  a_n$.
The case $n = 0$ is trivial.
Suppose
$a_0 \Leftrightarrow_{I} a_1
\stackrel{*}{\Leftrightarrow}_{I} a_n$.
By induction hypothesis 
we have $a_1
 \stackrel{*}{\to}_{I} a_l  \astvdashv_{I} a_m 
\stackrel{*}{\gets}_{I}  a_n$
for some $1 \le l,m \le n$.
We distinguish three cases:
\begin{enumerate}[(a)]
\item $a_0 \vdashv_I a_1$.
By ($\beta$), it follows that 
we have $a_0 \vdashv_I a_1 = a_l  \astvdashv_{I} a_m 
\stackrel{*}{\gets}_{I}  a_n$.
Hence the conclusion follows.

\item $a_0 \to_I a_1$.
Since we have $a_0 \to_I a_1
 \stackrel{*}{\to}_{I} a_l  \astvdashv_{I} a_m 
\stackrel{*}{\gets}_{I}  a_n$,
the conclusion follows.

\item $a_0 \gets_I a_1$.
Then by ($\alpha$), it follows that 
we have $a_0 \gets_I a_1 = a_l  \astvdashv_{I} a_m 
\stackrel{*}{\gets}_{I}  a_n$.
Furthermore, by ($\gamma$), it follows that 
$a_0 \gets_I a_1 = a_l = a_m 
\stackrel{*}{\gets}_{I}  a_n$.
Hence the conclusion follows.\qed\smallskip
\end{enumerate}
\end{enumerate}

\noindent The following abstract criterion for Church-Rosser modulo 
will be used as the basis of all of our confluence criteria
presented in this paper.

\begin{thm}[abstract criteria for CRM]
\label{thm:ARS}
Let $\vdashv,\to,\leadsto$ be relations on a set $A$
such that $\vdashv$ is symmetric,
${\leadsto} \subseteq {\vdashv}$, and
$\to \circ \stackrel{*}{\leadsto}$ is well-founded.
Let ${\Rightarrow} = {\leadsto \cup \to}$.
Suppose 
(i) ${\gets} \circ {\to}
\subseteq 
{\stackrel{*}{\Rightarrow}}
\circ {\eqvdashv} \circ 
{\stackrel{*}{\Leftarrow}}$
and 
(ii) ${\vdashv} \circ {\to}
\subseteq 
(
{\eqvdashv}
\circ {\stackrel{*}{\Leftarrow}})
\cup 
({\to} \circ {\stackrel{*}{\Rightarrow}}
\circ {\eqvdashv} \circ 
{\stackrel{*}{\Leftarrow}})$.
Then
$\to$ is Church-Rosser modulo $\astvdashv$.
\end{thm}

\newcommand{\LRleftrightarrow}{{\Lleftarrow}\!\!\!\!{\Rrightarrow}}
\newcommand{\boxrightarrow}{-\hspace{-.15em}\triangleright}
\newcommand{\astboxleftrightarrow}{\triangleleft\hspace{-.15em}{{\stackrel{*}-}\hspace{-.15em}\triangleright}}

\proof
Let ${\boxrightarrow} = {\vdashv} \cup {\rightarrow}$.
Suppose that the index of a step $a\: {\boxrightarrow}\: b$
be given by the multiset 
$lab(a \:{\boxrightarrow}\: b)$
defined like this:
$lab(a \vdashv b) = \{ a, b \}$
(i.e.\ $lab(a \leadsto b) = \{ a, b \}$),
$lab(a \rightarrow b) = \{ a \}$.
Let $\succ$ be the multiset extension 
of the transitive closure of 
the well-founded relation 
$\to \circ \stackrel{*}{\leadsto}$.
Then by our assumption it readily follows that 
(i) ${\gets}_{\alpha} \circ {\to}_{\beta}
\subseteq 
{\astboxleftrightarrow}_{\cv{\alpha}{\beta}}$
and 
(ii) ${\vdashv}_{\alpha} \circ {\to}_{\beta}
\subseteq 
{\astboxleftrightarrow}_{\cv{\alpha}{\beta}}$
are satisfied.
Thus, from Lemma \ref{lem:ARS},
${\to}$ is Church-Rosser modulo ${\astvdashv}$.\qed

Several corollaries of the theorem follow.

\begin{cor}[Corollary of Propositions 1 and 3 of \cite{JKR83}]
\label{cor:ARS-I}
Let $\vdashv,\to$ be relations on a set $A$
such that $\vdashv$ is symmetric and 
$\to$ is well-founded.
Suppose 
(i) ${\gets} \circ {\to}
\subseteq 
{\stackrel{*}{\to}} \circ
{\eqvdashv}
\circ {\stackrel{*}{\gets}}$
and 
(ii) ${\vdashv} \circ {\to}
\subseteq 
{\stackrel{*}{\to}} \circ
{\eqvdashv}
\circ {\stackrel{*}{\gets}}$.
Then
$\to$ is Church-Rosser modulo $\astvdashv$.
\end{cor}

\begin{proof}
Take ${\leadsto} := \emptyset$
in Theorem~\ref{thm:ARS}.
\end{proof}

\begin{cor}
\label{cor:ARS-II}
Let $\vdashv,\to,\leadsto$ be relations on a set $A$
such that $\vdashv$ is symmetric, ${\leadsto} \subseteq {\vdashv}$,
and $\to \circ \stackrel{*}{\leadsto}$ is well-founded.
Let ${\Rightarrow} =  {\to \cup \leadsto}$.
Suppose 
(i) ${\gets} \circ {\to}
\subseteq 
{\stackrel{*}{\Rightarrow}}
\circ {\eqvdashv} \circ 
{\stackrel{*}{\Leftarrow}}$
and 
(ii) ${\vdashv} \circ {\to}
\subseteq 
{\to} \circ
{\stackrel{*}{\Rightarrow}}
\circ {\eqvdashv} \circ 
{\stackrel{*}{\Leftarrow}}$.
Then
$\to$ is Church-Rosser modulo $\astvdashv$.
\end{cor}

\begin{proof}
Take the case
${\vdashv} \circ {\to}
\subseteq
{\to} \circ
{\stackrel{*}{\Rightarrow}} \circ
{\eqvdashv}
\circ {\stackrel{*}{\Leftarrow}}$ 
for the condition $(ii)$
in Theorem~\ref{thm:ARS}.
\end{proof}

In case ${\leadsto} := {\vdashv}$,
necessary and sufficient conditions for CRM
are obtained.

\begin{cor}
\label{cor:ARS-II'}
Let $\vdashv,\to$ be relations on a set $A$
such that $\vdashv$ is symmetric and 
$\to \circ \astvdashv$ is well-founded.
Let ${\Rightarrow} =  {\to \cup \vdashv}$.
Then
$\to$ is Church-Rosser modulo $\astvdashv$
if and only if 
(i) ${\gets} \circ {\to}
\subseteq 
{\stackrel{*}{\Rightarrow}}
\circ {\stackrel{*}{\Leftarrow}}$
and 
(ii) ${\vdashv} \circ {\to}
\subseteq 
{\to} \circ
{\stackrel{*}{\Rightarrow}}
\circ {\stackrel{*}{\Leftarrow}}$.
\end{cor}

\begin{proof}
($\Rightarrow$)
(i) is trivial.
To show (ii),
suppose $a \vdashv \circ \to b$.
Then 
$a 
\stackrel{*}{\to}
\circ \astvdashv 
\circ \stackrel{*}{\gets} b$ by our assumption.
Since 
${\stackrel{+}{\to}}
\circ {\astvdashv} 
\circ {\stackrel{*}{\gets}}
\subseteq 
{\to} \circ {\stackrel{*}{\Rightarrow}} \circ {\stackrel{*}{\Leftarrow}}$,
it remains to exclude the
case $a \astvdashv \circ \stackrel{*}{\gets} b$.
If $a \astvdashv \circ \stackrel{*}{\gets} b$
then 
$a \vdashv \circ \to b
\stackrel{*}{\to} \circ \astvdashv a$.
This contradicts our assumption that 
$\to \circ \astvdashv$
is well-founded.
($\Leftarrow$)
follows from Corollary~\ref{cor:ARS-II}.
\end{proof}

The conditions (i), (ii) can be
replaced with particular forms.

\begin{cor}
\label{cor:ARS-III}
Let $\vdashv,\to$ be relations on a set $A$
such that $\vdashv$ is symmetric and 
$\to \circ \astvdashv$ is well-founded.
Then
$\to$ is Church-Rosser modulo $\astvdashv$
if and only if 
(i) ${\gets} \circ {\to}
\subseteq 
{\stackrel{*}{\to}} \circ {\astvdashv} \circ {\stackrel{*}{\gets}}$
and 
(ii) ${\vdashv} \circ {\to}
\subseteq 
{\stackrel{+}{\to}} \circ
{\astvdashv} \circ {\stackrel{*}{\gets}}$.
\end{cor}

\begin{proof}
($\Rightarrow$)
(i) is trivial.
To show (ii),
suppose $a \vdashv \circ \to b$.
Then 
$a 
\stackrel{*}{\to}
\circ \astvdashv 
\circ \stackrel{*}{\gets} b$ by our assumption.
Thus it remains to exclude the
case $a \astvdashv \circ \stackrel{*}{\gets} b$.
If $a \astvdashv \circ \stackrel{*}{\gets} b$
then 
$a \vdashv \circ \to b
\stackrel{*}{\to} \circ \astvdashv a$.
This contradicts our assumption that 
$\to \circ \astvdashv$
is well-founded.
($\Leftarrow$)
Let ${\Rightarrow} = {\to \cup \vdashv}$.
Then 
${\gets} \circ {\to}
\subseteq 
{\stackrel{*}{\to}} \circ {\astvdashv} \circ {\stackrel{*}{\gets}}
\subseteq 
{\stackrel{*}{\Rightarrow}} \circ {\stackrel{*}{\Leftarrow}}$
and 
${\vdashv} \circ {\to}
\subseteq 
{\stackrel{+}{\to}} \circ
{\astvdashv} \circ {\stackrel{*}{\gets}}
\subseteq 
{\to} \circ
{\stackrel{*}{\Rightarrow}}
\circ 
{\stackrel{*}{\Leftarrow}}$.
Hence the claim follows
from Corollary~\ref{cor:ARS-II'}.
\end{proof}

\begin{rem}
\label{rem:Differnce with Huet}
Note that 
since 
Corollaries \ref{cor:ARS-II'} and \ref{cor:ARS-III} 
give the necessary and sufficient conditions for CRM,
the conditions (i), (ii) of Corollary \ref{cor:ARS-II'} 
imply
the conditions (i), (ii) of Corollary \ref{cor:ARS-III}.
\end{rem}

The condition 
(ii) ${\vdashv} \circ {\to}
\subseteq 
{\stackrel{+}{\to}} \circ
{\astvdashv} \circ {\stackrel{*}{\gets}}$
in this corollary
can be replaced with 
(ii) ${\vdashv} \circ {\to}
\subseteq 
{\stackrel{*}{\to}} \circ
{\astvdashv} \circ {\stackrel{*}{\gets}}$
as shown in the next corollary.

\begin{cor}[Lemma 2.8 of \cite{Hue80}]
\label{cor:Huet80}
Let $\vdashv,\to$ be relations on a set $A$
such that $\vdashv$ is symmetric and 
$\to \circ \astvdashv$ is well-founded.
Then
$\to$ is Church-Rosser modulo $\astvdashv$
if and only if 
(i) ${\gets} \circ {\to}
\subseteq 
{\stackrel{*}{\to}} \circ {\astvdashv} \circ {\stackrel{*}{\gets}}$
and 
(ii) ${\vdashv} \circ {\to}
\subseteq 
{\stackrel{*}{\to}} \circ
{\astvdashv} \circ {\stackrel{*}{\gets}}$.
\end{cor}

\begin{proof}
($\Rightarrow$) is trivial.
($\Leftarrow$)
We show our conditions
imply 
${\vdashv} \circ {\to}
\subseteq 
{\stackrel{+}{\to}} \circ
{\astvdashv} \circ {\stackrel{*}{\gets}}$.
Then the claim follows
from Corollary~\ref{cor:ARS-III}.
Suppose contrarily that
we have $a \vdashv {\circ} \to b$
but not 
$a \stackrel{+}{\to} \circ
\astvdashv  \circ \stackrel{*}{\gets} b$.
From $a \vdashv {\circ} \to b$,
we have
$a \stackrel{*}{\to} \circ
\astvdashv  \circ \stackrel{*}{\gets} b$ by our assumption.
Hence $a \astvdashv {\circ} \stackrel{*}{\gets} b$ holds.
Thus 
$a \vdashv {\circ} \to b
\stackrel{*}{\to}
{\circ}  \astvdashv a$.
This contradicts our assumption that 
$\to \circ \astvdashv$
is well-founded.
\end{proof}


\begin{rem}
The proof of Corollary \ref{cor:Huet80}
given in \cite{Hue80} is based on a combinatorial argument;
another proof of Corollary \ref{cor:Huet80}
based on an argument similar to the proof of Lemma \ref{lem:ARS}
has been given in \cite{AMrep1991}.
\end{rem}

We now give a proof of Theorem 5 of \cite{JK86}
based on Lemma~\ref{lem:ARS}.

\begin{lem}
\label{lem:For JK86}
Let $\vdashv,\to$ be relations on a set $A$
such that $\vdashv$ is symmetric and 
$\to \circ \astvdashv$ is well-founded.
Let $\blacktriangleright$ be a relation on $A$ satisfying
${\to} \subseteq {\blacktriangleright} \subseteq 
{\astvdashv} \circ {\to} \circ {\astvdashv}$.
Then 
$\blacktriangleright$ is Church-Rosser modulo $\astvdashv$
 if and only if
(i)$ {\gets} \circ {\blacktriangleright}
\subseteq 
{\stackrel{*}{\blacktriangleright}}
\circ {\astvdashv}
\circ {\stackrel{*}{\blacktriangleleft}}$
and 
(ii) ${\vdashv} \circ {\blacktriangleright}
\subseteq 
{\stackrel{+}{\blacktriangleright}}
\circ {\astvdashv}
\circ {\stackrel{*}{\blacktriangleleft}}$.
\end{lem}

\proof 
($\Rightarrow$) 
(i) is trivial.
(ii) 
Suppose 
$a \vdashv \circ \blacktriangleright b$.
Then 
$a \stackrel{*}{\blacktriangleright}
\circ \astvdashv
\circ \stackrel{*}{\blacktriangleleft} b$ by our assumption.
Thus it remains to exclude the case
$a \astvdashv
\circ \stackrel{*}{\blacktriangleleft} b$.
If $a \astvdashv
\circ \stackrel{*}{\blacktriangleleft} b$
then 
$a \vdashv \circ \blacktriangleright b
\stackrel{*}{\blacktriangleright} \circ
\astvdashv a$.
Hence
$a \:({\astvdashv} \circ {\to} \circ {\astvdashv})^+ \:a$.
This contradicts our assumption
that $\to \circ \astvdashv$ is well-founded.
($\Leftarrow$)
By assumption, ${\astvdashv} \circ {\to} \circ
{\astvdashv}$ is well-founded.  Let ${\rhd} = ({\astvdashv} \circ
{\to} \circ {\astvdashv})^+$. Clearly, $\rhd$ is well-founded.
Note that $a \astvdashv \circ \stackrel{+}{\blacktriangleright}
\circ \astvdashv b$ implies $a \rhd b$.
Let ${\Rightarrow} = {\blacktriangleright} \cup {\vdashv}$.  For each
step $a \Leftrightarrow b$, let its index be given by the multiset
$lab(a \vdashv b) = \{ a, b \}$, $lab(a \blacktriangleright b) = \{
a \}$ and $lab(a \blacktriangleleft  b) = \{ b \}$.  Let $\succ$ be
the multiset extension of $\rhd$.  Then, $\succ$ is a well-founded
relation on the set of indexes.
By Lemma~\ref{lem:ARS},
it suffices to show
(i$'$) 
${\blacktriangleleft}_{\alpha}  \circ {\blacktriangleright}_{\beta}
\subseteq 
{\stackrel{*}{\Leftrightarrow}}_{\cv{\alpha}{\beta}}$
and
(ii$'$)
${\vdashv}_{\alpha}  \circ {\blacktriangleright}_{\beta}
\subseteq 
{\stackrel{*}{\Leftrightarrow}}_{\cv{\alpha}{\beta}}$.

To prove (i$'$),
we claim that for any 
$b \blacktriangleleft a  = a_0 \vdashv a_1 \vdashv \cdots
\vdashv a_m \to c'  \astvdashv c$,
we have $b \stackrel{*}{\Leftrightarrow}_{\curlyvee \{ a \}} c$.
As $\alpha = \beta = \{ a \}$,
(i$'$) immediately follows from this claim.
Our proof proceeds by induction on $m$.
\begin{iteMize}{$\bullet$}
\item
Suppose $m = 0$.
Then we have $b \blacktriangleleft a  \to c'  \astvdashv c$.
Then, by our condition (i),
$b \stackrel{*}{\blacktriangleright}
\circ \astvdashv
\circ \stackrel{*}{\blacktriangleleft}
c' \astvdashv c$.
Since $b \blacktriangleleft a$ and $a \blacktriangleright c'$,
we have $b \stackrel{*}{\Leftrightarrow}_{\curlyvee \{ a \}} c$.

\item
Suppose $m > 0$.  Then, since we have $b \blacktriangleleft a \vdashv
a_1$, by our condition (ii), $b \stackrel{*}{\blacktriangleright}
\circ \astvdashv \circ \stackrel{*}{\blacktriangleleft} a_1'
\blacktriangleleft a_1$ for some $a_1'$.  Then we have $a_1'
\blacktriangleleft a_1 \vdashv \cdots \vdashv a_m \to c'
\astvdashv c$, and hence, by induction hypothesis, $a_1'
\stackrel{*}{\Leftrightarrow}_{\curlyvee \{ a_1 \}} c$.  Since $a
\vdashv a_1$, $a_1 \rhd d$ implies $a \rhd d$ for any $d$.  Hence
$a_1' \stackrel{*}{\Leftrightarrow}_{\curlyvee \{ a \}} c$.  Thus we
have $b \stackrel{*}{\blacktriangleright} \circ \astvdashv \circ
\stackrel{*}{\blacktriangleleft} a_1'
\stackrel{*}{\Leftrightarrow}_{\curlyvee \{ a \}} c$.  Hence, by $a
\blacktriangleright b$ and $a \vdashv \circ \blacktriangleright
a_1'$, we have $b \stackrel{*}{\Leftrightarrow}_{\curlyvee \{a \}} c$.
\end{iteMize}

To prove (ii$'$),
we claim that for any 
$b \vdashv a  \blacktriangleright c$,
we have $b \stackrel{*}{\Leftrightarrow}_{\curlyvee \{a, b \}} c$.
From the condition (ii) we have $b \blacktriangleright b' 
\stackrel{*}{\blacktriangleright} \circ  \astvdashv \circ
\stackrel{*}{\blacktriangleleft} c$.
 By $a  \vdashv b \blacktriangleright b'$ and 
$a  \blacktriangleright c$ we have 
 $b' \stackrel{*}{\Leftrightarrow}_{\curlyvee \{a \}} c$, and 
by $b \blacktriangleright b'$ we have 
$b \stackrel{*}{\Leftrightarrow}_{\curlyvee \{a, b \}} b'$. Hence, 
we obtain $b \stackrel{*}{\Leftrightarrow}_{\curlyvee \{a, b \}} c$.

Thus, from Lemma~\ref{lem:ARS}, we conclude $\blacktriangleright$ is
Church-Rosser modulo $\astvdashv$.\qed

Let $\blacktriangleright$ be a relation on a set $A$ satisfying
${\to} \subseteq {\blacktriangleright} \subseteq 
{\astvdashv} \circ {\to} \circ {\astvdashv}$.
Then
$\to$ is said to be \textit{$\blacktriangleright$-Church-Rosser 
modulo $\astvdashv$}
if 
${\stackrel{*}{\bowtie}}\subseteq 
{\stackrel{*}{\blacktriangleright}}
\circ {\astvdashv}
\circ {\stackrel{*}{\blacktriangleleft}}$,
where 
${\bowtie} = {\leftrightarrow \cup  \vdashv}$. 

\begin{cor}[Theorem 5 of \cite{JK86}]
\label{cor:Theorem of JK86 revised}
Let $\vdashv,\to$ be relations on a set $A$
such that $\vdashv$ is symmetric and 
$\to \circ \astvdashv$ is well-founded.
Let $\blacktriangleright$ be a relation on $A$ satisfying
${\to} \subseteq {\blacktriangleright} \subseteq 
{\astvdashv} \circ {\to} \circ {\astvdashv}$.
Then 
$\to$ is $\blacktriangleright$-Church-Rosser 
modulo $\astvdashv$
if and only if
(i) ${\gets} \circ {\blacktriangleright}
\subseteq 
{\stackrel{*}{\blacktriangleright}}
\circ {\astvdashv}
\circ {\stackrel{*}{\blacktriangleleft}}$
and 
(ii) ${\vdashv} \circ {\blacktriangleright}
\subseteq 
{\stackrel{+}{\blacktriangleright}}
\circ {\astvdashv}
\circ {\stackrel{*}{\blacktriangleleft}}$.
\end{cor}

\begin{proof}
Let ${\bowtie} = {\leftrightarrow \cup  \vdashv}$
and ${\Rightarrow} = {\blacktriangleright} \cup {\vdashv}$. From
${\to} \subseteq {\blacktriangleright} \subseteq {\astvdashv} \circ
{\to} \circ {\astvdashv}$, it follows that ${\bowtie}\subseteq
{\Leftrightarrow} \subseteq {\stackrel{*}{\bowtie}}$.  Hence, we have
${\stackrel{*}{\bowtie}} = {\stackrel{*}{\Leftrightarrow}}$. 
Thus, from Lemma~\ref{lem:For JK86}, the claim  follows.
\end{proof}

\begin{rem}
As ${\blacktriangleright} \circ {\astvdashv}$
is well-founded,
the condition 
(ii) ${\vdashv} \circ {\blacktriangleright}
\subseteq 
{\stackrel{+}{\blacktriangleright}}
\circ {\astvdashv}
\circ {\stackrel{*}{\blacktriangleleft}}$
in this corollary
can be replaced with
(ii) ${\vdashv} \circ {\blacktriangleright}
\subseteq 
{\stackrel{*}{\blacktriangleright}}
\circ {\astvdashv}
\circ {\stackrel{*}{\blacktriangleleft}}$,
similar to Corollaries
\ref{cor:ARS-III} and \ref{cor:Huet80}.
\end{rem}

\begin{rem}
A $\blacktriangleright$-normal form of an element $a$
is an element $b$ such that
$a \stackrel{*}{\blacktriangleright} b$
and $b \blacktriangleright c$ for no $c$.
The condition ``(i) and (ii)''
can be replaced with the condition
that, for any $a,b$ and their respective 
$\blacktriangleright$-normal forms
$\hat a, \hat b$,
(i$'$) 
$a \gets \circ \blacktriangleright b$
implies $\hat a  \astvdashv \hat b$
and 
(ii$'$) 
$a \vdashv \circ \blacktriangleright b$
implies
$\hat a  \astvdashv \hat b$,
which is explained as follows.
Clearly,
(i$'$) and (ii$'$) imply (i) and (ii).
To show the reverse direction, suppose (i) and (ii).
Let 
$\hat a, \hat b$ be $\blacktriangleright$-normal forms
$a,b$, respectively.
If $a \gets \circ \blacktriangleright b$
($a \vdashv \circ \blacktriangleright b$)
then 
$\hat a \stackrel{*}{\blacktriangleleft} a
\gets \circ \blacktriangleright b
\stackrel{*}{\blacktriangleright} \hat b$
($\hat a \stackrel{*}{\blacktriangleleft} a
\vdashv \circ \blacktriangleright b
\stackrel{*}{\blacktriangleright} \hat b$, respectively)
and thus 
$\hat a \stackrel{*}{\bowtie} \hat b$,
where ${\bowtie} = {\leftrightarrow \cup  \vdashv}$.
Thus $\hat a \stackrel{*}{\blacktriangleright}
\circ \astvdashv
\circ \stackrel{*}{\blacktriangleleft} \hat b$,
as $\to$ is $\blacktriangleright$-Church-Rosser 
modulo $\astvdashv$.
Since $\hat a, \hat b$ are $\blacktriangleright$-normal forms,
we conclude  $\hat a \astvdashv \hat b$.
\end{rem}

\subsection{Related works}

Several other abstract criteria for \textrm{CRM}
have been obtained in \cite{Hue80} and \cite{Ohl98}.
In this subsection, we compare our criterion with these.

The following 
necessary and sufficient criterion for \textrm{CRM}
is obtained in \cite{Hue80}.

\begin{prop}[Corollary of Lemmas 2.6 and 2.7 of \cite{Hue80}]
\label{prop:suffient and necessary condition of CRsim}
Let $\vdashv,\to$ be relations on a set $A$
such that $\vdashv$ is symmetric and 
$\to$ is well-founded.
Then 
$\to$ is Church-Rosser modulo $\astvdashv$
if and only if 
(i$\,'$) ${\gets} \circ {\to}
\subseteq 
{\stackrel{*}{\to}} \circ {\astvdashv} \circ {\stackrel{*}{\gets}}$
and 
(ii$\,'$) ${\astvdashv} \circ {\to}
\subseteq 
{\stackrel{*}{\to}} \circ
{\astvdashv} \circ {\stackrel{*}{\gets}}$.
\end{prop}
In this proposition, 
$\to$ is supposed to be well-founded,
similarly to Corollary \ref{cor:ARS-I}.
This proposition and 
Corollary \ref{cor:ARS-I}, however, differ
in the following points.
\begin{enumerate}[(1)]
\item 
Proposition \ref{prop:suffient and necessary condition of CRsim}
gives necessary and sufficient conditions
for \textrm{CRM},
while Corollary \ref{cor:ARS-I} gives only sufficient conditions.
\item 
The condition part of ($ii'$) of
Proposition \ref{prop:suffient and necessary condition of CRsim}
is not localized
(i.e.\ ${\astvdashv} \circ {\to}$ is assumed),
while that of ($ii$) of Corollary \ref{cor:ARS-I} 
is localized  (i.e.\  ${\vdashv} \circ {\to}$ is assumed).
\item 
In the conclusion parts of ($i'$) and ($ii'$) of
Proposition \ref{prop:suffient and necessary condition of CRsim}
an arbitrary number of ${\vdashv}$-steps are allowed,
while 
in those of ($i$) and ($ii$) of Corollary \ref{cor:ARS-I} 
the number of ${\vdashv}$-steps needs to be at most one.
\end{enumerate}


\noindent The decreasing diagram technique \cite{CRbyDD}
is a powerful technique to obtain many confluence criteria.
In \cite{Ohl98}, the technique is extended to
obtain a criterion for CRM.

\begin{prop}[Theorem 14 of \cite{Ohl98}]
\label{prop:Theorem of Ohl 98}
Let $I$ be a set of indexes equipped
with a well-founded order $\succ$.
Let $\to_{\alpha}$ be a relation on a set $A$
for each $\alpha \in I$
and $\vdashv$ a symmetric relation on $A$.
Suppose 
(i) ${\gets}_{\alpha} \circ {\to}_{\beta}
\subseteq 
{\stackrel{*}{\to}}_{\curlyvee \alpha}
\circ
{\stackrel{=}{\to}}_{\beta}
\circ
{\stackrel{*}{\to}}_{\cv{\alpha}{\beta}}
\circ
{\astvdashv}
\circ
{\stackrel{*}{\gets}}_{\cv{\alpha}{\beta}}
\circ
{\stackrel{=}{\gets}}_{\alpha}
\circ
{\stackrel{*}{\gets}}_{\curlyvee \beta}$
and 
(ii) ${\vdashv} \circ {\to}_{\beta}
\subseteq 
{\stackrel{=}{\to}}_{\beta}
\circ
{\stackrel{*}{\to}}_{\curlyvee \beta}
\circ
{\astvdashv}
\circ
{\stackrel{*}{\gets}}_{\curlyvee \beta}$.
Then ${\to}_I$ is Church-Rosser modulo ${\astvdashv}$.
\end{prop}

This proposition is given in terms of indexed relations
as in Lemma \ref{lem:ARS}.
This proposition and Lemma \ref{lem:ARS} differ
in the following points.
\begin{enumerate}[(1)]
\item
In Lemma \ref{lem:ARS} $\vdashv\,$-steps
are indexed, while in Proposition \ref{prop:Theorem of Ohl 98}
$\vdashv\,$-steps are not indexed.

\item
In the case ${\succ} = \emptyset$,
Lemma \ref{lem:ARS} is meaningless,
while 
Proposition \ref{prop:Theorem of Ohl 98}
still remains as a simple criterion for CRM.
\end{enumerate}

\begin{rem}
The labelings like $lab(a \to b) = \{ a \}$ and $lab(a \to b) = \{ a, b \}$
used in the proof of Theorem \ref{thm:ARS}
are called source and step labelings in \cite{converted},
which are used to obtain abstract confluence criteria
from the decreasing diagram criterion.
\end{rem}

\section{Confluence criteria}

In this  section,
we develop several confluence criteria
for TRSs that can be partitioned into
terminating TRS $\mathcal{S}$
and reversible TRS $\mathcal{P}$.
After the preliminaries (subsection 1),
we present our first criterion
that works for the case that $\mathcal{S}$ is linear
(subsection 2).
Next we claim another criterion
effective for left-linear $\mathcal{S}$---for this,
we first give a  criterion using the usual notion of 
the critical pairs (subsection 3)
and then extend the criterion using the
notion of parallel critical pairs (subsection 4).
Finally, we give some examples and describe relations
among the given criteria (subsection 5).

\subsection{Preliminaries}

In this subsection, we fix
some notions and notations on term rewriting systems
and after that we present some lemmas
that will be used in later subsections.

Let $\mathcal{F}$ be a set of fixed arity function symbols
and $\mathcal{V}$ be a set of variables.
The set $\mathrm{T}(\mathcal{F},\mathcal{V})$
of terms over $\mathcal{F}$ and $\mathcal{V}$
is defined like this:
(1) $\mathcal{V} \subseteq \mathrm{T}(\mathcal{F},\mathcal{V})$;
(2) if $f \in \mathcal{F}$ has arity $n$
and $t_1,\ldots,t_n \in \mathrm{T}(\mathcal{F},\mathcal{V})$
then $f(t_1,\ldots,t_n) \in \mathrm{T}(\mathcal{F},\mathcal{V})$.
The sets of function symbols and variables occurring in a term $t$
are denoted by $\mathcal{F}(t)$ and $\mathcal{V}(t)$, respectively.
A \textit{linear} term is a term in which any variable
occurs at most once.
\textit{Positions} are finite sequences of positive integers.
The \textit{empty sequence} is denoted by $\epsilon$. 
The set of positions in a term $t$ is denoted by 
$\mathrm{Pos}(t)$.
The \textit{concatenation} of positions $p,q$ is denoted by $p.q$.
We use $\le$ for prefix ordering on positions,
i.e.\ $p \le q$ iff $\exists o.~p.o = q$.
We write $p < q$ iff $p \le q$ and $p \neq q$.
For positions $p, q$ such that $p \le q$,
the position $o$ satisfying $p.o = q$ is denoted by $q/p$.
Positions $p_1,\ldots,p_n$ are \textit{parallel}
if $p_i \not\le p_j$ for any $i \neq j$.
We write $p \parallel q$ if two positions $p,q$ are parallel.
For sets $U,V$ of positions,
we write $U \parallel V$ if $p \parallel q$ holds
for any $p \in U$ and $q \in V$.
If $p$ is a position in a term $t$, 
then the symbol in $t$ at the position $p$ is written as $t(p)$,
the subterm of $t$ at the position $p$ is written as $t/p$,
and the term obtained by replacing the subterm $t/p$ by a term $s$
is written as $t[s]_p$.
For any $P \subseteq \mathrm{Pos}(t)$, we define
$\mathcal{V}_P(t) = \bigcup_{p \in P} \mathcal{V}(t/p)$.
For $X \subseteq \mathcal{F} \cup \mathcal{V}$,
we put $\mathrm{Pos}_X(t) = \{ p \in \mathrm{Pos}(t) \mid t(p) \in X \}$.
For parallel positions $p_1,\ldots,p_n$ in a term $t$,
the term obtained by replacing each subterm $t/p_i$ by a term $s_i$
is written as $t[s_1,\ldots,s_n]_{p_1,\ldots,p_n}$.
A \textit{context} is an expression $t[\,,\ldots,\,]_{p_1,\ldots,p_n}$ 
in which such subterms are dropped.

A map $\sigma$ from $\mathcal{V}$ to $\mathrm{T}(\mathcal{F},\mathcal{V})$ 
is a \textit{substitution} 
if the domain $\mathrm{dom}(\sigma)$ of $\sigma$ is finite
where $\mathrm{dom}(\sigma)  = \{ x \in \mathcal{V} \mid \sigma(x) \neq  x \}$.
As usual, we identify each substitution with its homomorphic extension.
For a substitution $\sigma$ and a term $t$, $\sigma(t)$ is also written as $t\sigma$.
A relation $R$ on $\mathrm{T}(\mathcal{F},\mathcal{V})$
is \textit{stable} if for any terms $s,t \in \mathrm{T}(\mathcal{F},\mathcal{V})$,
$s\:R\:t$ implies $s\theta\:R\:t\theta$ for any substitution $\theta$;
it is \textit{monotone} if 
$s\:R\:t$ implies 
$f(\ldots,s,\ldots) \:R\:f(\ldots,t,\ldots)$ for any $f \in \mathcal{F}$.
A relation $R$ on $\mathrm{T}(\mathcal{F},\mathcal{V})$
is a \textit{rewrite relation} if it is stable and monotone.

For a set $\mathcal{E}$ of equations,
we write $\mathcal{E}^{-1} = \{ r \approx l \mid l \approx r \in \mathcal{E} \}$.
Equations are identified modulo renaming (of variables),
for example, $+(x,y) = +(y,x)$ equals to $+(y,z) = +(z,y)$.
A set $\mathcal{E} = \{ s_1 \approx t_1, \ldots, s_n \approx t_n \}$ of equations
is \textit{unifiable} if there exists a substitution $\sigma$
such that $s_i\sigma = t_i\sigma$ for all $i$;
the substitution $\sigma$ is a \textit{unifier} of $\mathcal{E}$.
The most general unifier is abbreviated as \textit{mgu}.
If there exists an equation $l \approx r \in \mathcal{E}$ 
or $r \approx l \in \mathcal{E}$ 
and a position $p$ in a term $s$
and substitution $\theta$ such that 
$s/p = l\theta$ and $t = s[r\theta]_p$,
then we write $s \leftrightarrow_{\mathcal{E}} t$.

An equation $l \approx r$ is a \textit{rewrite rule} 
if it satisfies the conditions
(1) $l \notin \mathcal{V}$
and (2) $\mathcal{V}(r) \subseteq \mathcal{V}(l)$.
A rewrite rule $l \approx r$ is written as $l \to r$.
A rewrite rule $l \to r$ is \textit{linear} (\textit{left-linear})
if $l$ and $r$ are linear ($l$ is linear, respectively);
it is \textit{bidirectional} if
$r \approx l$ is a rewrite rule.
A \textit{term rewriting system} (\textit{TRS} for short) is
a finite set of rewrite rules.
A TRS is left-linear (linear, bidirectional)
if so are all its rewrite rules.
If a TRS $\mathcal{R}$ is bidirectional
then $\mathcal{R}^{-1} = \{ r \to l \mid l \to r \in \mathcal{R} \}$
is a TRS
and $\mathcal{R} \cup \mathcal{R}^{-1}$
is a bidirectional TRS.
Let $\mathcal{R}$ be a TRS.
If there exists a rewrite rule $l \to r \in \mathcal{R}$ 
and a position $p$ in a term $s$
and substitution $\theta$ such that 
$s/p = l\theta$ and $t = s[r\theta]_p$,
we write $s \to_{p,\mathcal{R}} t$.
If $p$ ($p$ and $\mathcal{R}$) is clear from the context,
$s \to_{p,\mathcal{R}} t$ is 
written as $s \to_{\mathcal{R}} t$ ($s \to  t$, respectively).
We call $s \to_{p,\mathcal{R}} t$ a \textit{rewrite step};
the subterm $s/p$ is the \textit{redex} of this rewrite step.
We say $t$ is obtained by \textit{contracting} the redex $s/p$.
The relation $\to_{\mathcal{R}}$ on $\mathrm{T}(\mathcal{F},\mathcal{V})$
is a rewrite relation
and called \textit{the} rewrite relation of $\mathcal{R}$.
A term $s$ is \textit{normal} 
if $s \to_\mathcal{R} t$ for no term $t$.
The set of normal terms is denoted by $\mathrm{NF}(\mathcal{R})$.
A \textit{normal form} (or \textit{$\mathcal{R}$-normal form})
of a term $s$ is a term $t \in \mathrm{NF}(\mathcal{R})$ 
such that $s \stackrel{*}{\rightarrow}_\mathcal{R} t$.  
Two terms $s$ and $t$ are said to be \textit{joinable}
if $s \stackrel{*}{\to}_{\mathcal{R}} \circ
\stackrel{*}{\gets}_{\mathcal{R}} t$.
A TRS $\mathcal{R}$ is \textit{terminating} 
if $\to_\mathcal{R}$ is well-founded;
$\mathcal{R}$ is confluent if $\to_\mathcal{R}$ is confluent.
A TRS $\mathcal{R}$ is \textit{terminating relative to} 
a TRS $\mathcal{P}$ if $\to_\mathcal{R} \circ
\stackrel{*}{\to}_\mathcal{P}$ is well-founded;
A TRS $\mathcal{R}$ is \textit{terminating modulo} 
a set $\mathcal{E}$ of equations 
if $\to_\mathcal{R} \circ \stackrel{*}{\leftrightarrow}_\mathcal{E}$ is well-founded.

Let $s,t$ be terms whose variables are disjoint.
The term $s$ \textit{overlaps} on $t$ (at a position $p$)
when there exists a non-variable subterm $u = t/p$ of $t$ 
such that $u$ and $s$ are unifiable.
Let $l_1 \to r_1$ and $l_2 \to r_2$ be rewrite rules.
W.l.o.g.\ let their variables be disjoint.
Suppose that $l_1$ overlaps on $l_2$ at a position $p$
and $\sigma$ is the mgu of $l_1$ and $l_2/p$.
Then the term $l_2[l_1]_p\sigma$ yields a \textit{critical pair}
$\langle l_2[r_1]_p\sigma, r_2\sigma \rangle$
obtained by the overlap of $l_1 \to r_1$ on $l_2 \to r_2$ at the position $p$.
In the case of self-overlap 
(i.e.\ when $l_1 \to r_1$ and $l_2 \to r_2$ 
are identical modulo renaming),
we do not consider the case $p = \epsilon$.
We call the critical pair \textit{outer} if $p = \epsilon$
and \textit{inner} if $p > \epsilon$.
The set of outer (inner) critical pairs obtained by the overlaps of 
a rewrite rule from $\mathcal{R}$ on a rewrite rule from $\mathcal{Q}$ 
is denoted by 
$\mathrm{CP}_\textit{out}(\mathcal{R},\mathcal{Q})$
($\mathrm{CP}_\textit{in}(\mathcal{R},\mathcal{Q})$, respectively).
We put
$\mathrm{CP}(\mathcal{R},\mathcal{Q}) = 
\mathrm{CP}_\textit{out}(\mathcal{R},\mathcal{Q})
\cup 
\mathrm{CP}_\textit{in}(\mathcal{R},\mathcal{Q})$.
For a set $C$ of pairs of terms,
we write $C^{-1} = \{ \langle v,u \rangle \mid 
\langle u,v \rangle \in C \}$.
We note that
$\mathrm{CP}_\textit{out}(\mathcal{R},\mathcal{Q})
= \mathrm{CP}_\textit{out}(\mathcal{Q},\mathcal{R})^{-1}$.

\begin{exa}
\label{exp:critical paris of plus-com-assoc}
Let $\mathcal{R}_2 = \{ (\mathsf{add}_1),(\mathsf{add}_2),(C),(A) 
\}$ be the TRS for addition of natural
numbers and AC-rules for plus given in Example \ref{exp:plus-com-assoc}.
Let $\mathcal{S} = \{ (\mathsf{add}_1),(\mathsf{add}_2) \}$
and $\mathcal{P} = \{ (C), (A) \}$.
Then we have $\mathrm{CP}_\textit{in}(\mathcal{P},\mathcal{S})
= \emptyset$, 
\[
\mathrm{CP}_\textit{out}(\mathcal{S},\mathcal{P})
= \mathrm{CP}_\textit{out}(\mathcal{P},\mathcal{S})^{-1}
=
\left\{ 
\begin{array}{l@{\quad}ll@{\,}l}
\langle y, \mathsf{+}(y,\mathsf{0}) \rangle\\
\langle \mathsf{s}(\mathsf{+}(x,y)), \mathsf{+}(y,\mathsf{s}(x)) \rangle\\
\end{array}
\right\}
\]
and
\[
\mathrm{CP}_\textit{in}(\mathcal{S},\mathcal{P})
= 
\left\{ 
\begin{array}{l@{\quad}ll@{\,}l}
\langle \mathsf{+}(y,z), \mathsf{+}(\mathsf{0},\mathsf{+}(y,z)) \rangle\\
\langle \mathsf{+}(\mathsf{s}(\mathsf{+}(x,y)),z), 
        \mathsf{+}(\mathsf{s}(x),\mathsf{+}(y,z)) \rangle\\
\end{array}
\right\}.
\]\medskip
\end{exa}

\noindent The \textit{parallel extension}
$\parto_\mathcal{R}$ of the rewrite relation $\to_\mathcal{R}$ of 
a TRS $\mathcal{R}$ is defined like this:
$s \parto_{\{ p_1,\ldots,p_n \},\mathcal{R}} t$
if  
$p_1,\ldots,p_n$ are parallel positions in the term $s$
and 
there exist rewrite rules $l_1 \to r_1,\ldots,l_n \to r_n \in \mathcal{R}$ 
and substitution $\theta_1,\ldots,\theta_n$ such that 
$s/p_i = l_i\theta_i$ for each $i$
and $t = s[r_1\theta_1,\ldots,r_n\theta_n]_{p_1,\ldots,p_n}$.
If the missing information is clear from the context,
$s \parto_{\{ p_1,\ldots,p_n \},\mathcal{R}} t$
is written as 
$s \parto_\mathcal{R} t$ or $s \parto t$.
We call $s \parto_{\mathcal{R}} t$ a \textit{parallel rewrite step}.
For substitutions $\rho,\rho'$,
we write $\rho \parto_{\mathcal{R}} \rho'$ 
if $\rho(x) \parto_{\mathcal{R}} \rho(x)$ for any $x \in \mathcal{V}$.
We note that
$\parto_\mathcal{R}$ includes
the identity relation,
i.e.\ $t \parto_\mathcal{R} t$ for any term $t$.

In the rest of this subsection,
we present several lemmas that will be
used several times in later subsections.
The first lemma is used to 
connect our abstract criterion for Church-Rosser modulo
with concrete criteria for confluence.
For this, we introduce the notion of reversibility.

\begin{defi}[reversible relation]
A relation $\to$ is said to be
\textit{reversible} if
${\to} \subseteq  {\stackrel{*}{\gets}}$.
A TRS $\mathcal{R}$ is reversible
if $\to_{\mathcal{R}}$ is reversible.
\end{defi}

\begin{lem}[confluence by CRM and reversibility]
\label{lem:reversible + CR modulo}
Let $\mathcal{P},\mathcal{S}$ be TRSs
such that $\mathcal{P}$ is reversible.
If $\to_\mathcal{S}$ is Church-Rosser modulo
$\stackrel{*}{\leftrightarrow}_\mathcal{P}$
then $\mathcal{S} \cup \mathcal{P}$ is confluent.
\end{lem}

\begin{proof}
Suppose $s \stackrel{*}{\to}_{\mathcal{S} \cup \mathcal{P}} t_0$
and 
$s \stackrel{*}{\to}_{\mathcal{S} \cup \mathcal{P}} t_1$.
Then 
$t_0 \stackrel{*}{\leftrightarrow}_{\mathcal{S} \cup \mathcal{P}} t_1$.
Since $\to_\mathcal{S}$ is Church-Rosser modulo
$\stackrel{*}{\leftrightarrow}_\mathcal{P}$,
we have
$t_0 \stackrel{*}{\to}_{\mathcal{S}}
u \stackrel{*}{\leftrightarrow}_{\mathcal{P}}
v \stackrel{*}{\gets}_{\mathcal{S}} t_1$.
By the reversibility of $\mathcal{P}$,
we have 
$u \stackrel{*}{\to}_{\mathcal{P}} v$.
Hence 
$t_0 \stackrel{*}{\to}_{\mathcal{S} \cup \mathcal{P}}
v \stackrel{*}{\gets}_{\mathcal{S}\cup \mathcal{P}} t_1$.
\end{proof}

It is well-known that 
$\mathcal{S}$ is locally confluent,
i.e.\ 
${\gets}_\mathcal{S} \circ {\to}_\mathcal{S}
\subseteq 
{\stackrel{*}{\to}_\mathcal{S}} \circ {\stackrel{*}{\gets}}_\mathcal{S}$,
if $\mathrm{CP}(\mathcal{S},\mathcal{S})
\subseteq 
{\stackrel{*}{\to}_\mathcal{S}} \circ {\stackrel{*}{\gets}}_\mathcal{S}$.
The next lemma parametrized this fact by a rewrite relation $\blacktriangleright$.

\begin{lem}
\label{lem:main-I}
Let $\mathcal{S}$ be a TRS
and ${\blacktriangleright}$
be a rewrite relation.
Suppose that 
$\mathrm{CP}(\mathcal{S},\mathcal{S})
\subseteq {\blacktriangleright} \cap {\blacktriangleleft}$
and 
${\stackrel{*}{\to}}_\mathcal{S} \circ {\stackrel{*}{\gets}}_\mathcal{S}
\subseteq {\blacktriangleright}$.
Then 
${\gets}_\mathcal{S} \circ {\to}_\mathcal{S}
\subseteq 
{\blacktriangleright}$.
\end{lem}

\begin{proof}
Suppose
$t_0 \gets_{p,\mathcal{S}} s \to_{q,\mathcal{S}} t_1$.
We distinguish the cases by relative positions of $p$ and $q$.
If $p \parallel q$ then
$s = s[l\sigma,l'\rho]_{p,q}$
and 
$t_0 = s[r\sigma,l'\rho]_{p,q} \gets_p s  \to_q s[l\sigma,r'\rho]_{p,q} = t_1$
for some rewrite rules $l \to r, l' \to r' \in \mathcal{S}$
and substitutions $\sigma,\rho$.
Then we have 
$t_0 = s[r\sigma,l'\rho]_{p,q} \to_q s[r\sigma,r'\rho]_{p,q}
\gets_p s[l\sigma,r'\rho]_{p,q} = t_1$.
Thus $t_0 \blacktriangleright t_1$ follows from our assumption
${\stackrel{*}{\to}}_\mathcal{S} \circ {\stackrel{*}{\gets}}_\mathcal{S}
\subseteq {\blacktriangleright}$.
Suppose $q \le p$.
Let $s/q = l\sigma$ and $l \to r \in \mathcal{S}$.
Then either 
(1) $p/q \in \mathrm{Pos}_\mathcal{F}(l)$
or 
(2) there exists $q_x \in \mathrm{Pos}_\mathcal{V}(l)$
such that $l/q_x = x \in \mathcal{V}$ and $q.q_x \le p$.
\begin{enumerate}[(1)]
\item
Then $t_0 = s[u\rho]_q$ 
and $t_1 = s[v\rho]_q$
for some $\langle u, v \rangle \in \mathrm{CP}(\mathcal{S},\mathcal{S})$
and substitution $\rho$.
Thus by assumption $u \blacktriangleright v$.
Then, since $\blacktriangleright$ is a rewrite relation,
we have 
$t_0 = s[u\rho]_q \blacktriangleright s[v\rho]_q = t_1$.

\item
Then $t_1 = s[r\sigma]_q$
and $s = s[l\sigma]_q  \to_{p,\mathcal{S}} t_0 \stackrel{*}{\to}_\mathcal{S} s[l\sigma']_q$ 
for some substitution $\sigma'$
such that $\sigma(x) \to_{p/(q.q_x),\mathcal{S}} \sigma'(x)$
and $\sigma'(y) = \sigma(y)$ for any $y \neq x$.
Thus 
$t_0\stackrel{*}{\to}_\mathcal{S} 
s[l\sigma']_q
\to_\mathcal{S}
s[r\sigma']_q
\stackrel{*}{\gets}_\mathcal{S}
s[r\sigma]_q = t_1$.
The claim follows from our assumption
${\stackrel{*}{\to}}_\mathcal{S} \circ {\stackrel{*}{\gets}}_\mathcal{S}
\subseteq {\blacktriangleright}$.
\end{enumerate}
The case of $p < q$ follows
similarly to the case of $q \le p$
using the condition
$\mathrm{CP}(\mathcal{S},\mathcal{S})
\subseteq {\blacktriangleleft}$.
\end{proof}

%

\begin{lem}
\label{lem:main-II}
Let $\mathcal{Q},\mathcal{R}$ be TRSs
such that $\mathcal{Q}$ is bidirectional,
and let $\blacktriangleright$ be a rewrite relation.
Suppose that 
$\mathrm{CP}(\mathcal{Q},\mathcal{R})
\subseteq {\blacktriangleright}$,
$\mathrm{CP}(\mathcal{R},\mathcal{Q})
\subseteq {\blacktriangleleft}$.
(1) If $\mathcal{R}$ is linear 
and ${\stackrel{+}{\to}_\mathcal{R}}
\circ {\stackrel{=}{\gets}_\mathcal{Q}}
\circ {\stackrel{*}{\gets}_\mathcal{R}}
\subseteq {\blacktriangleright}$
then 
${\gets}_\mathcal{Q} \circ {\to}_\mathcal{R}
\subseteq {\blacktriangleright}$.
(2) If $\mathcal{R}$ is left-linear 
and ${\stackrel{+}{\to}_\mathcal{R}}
\circ {\pargets_\mathcal{Q}}
\circ {\stackrel{*}{\gets}_\mathcal{R}}
\subseteq {\blacktriangleright}$
then 
${\gets}_\mathcal{Q} \circ {\to}_\mathcal{R}
\subseteq {\blacktriangleright}$.
\end{lem}

\proof
Below we present a proof for (1).
Any difference to the proof of (2)
will be mentioned in the proof.
Suppose
$t_0 \gets_{p,\mathcal{Q}} s \to_{q,\mathcal{R}} t_1$.
We distinguish the cases by relative positions of $p$ and $q$.
\begin{enumerate}[(i)]
\item $p \parallel q$.
Then 
$s = s[l\sigma,l'\rho]_{p,q}$
and 
$t_0 = s[r\sigma,l'\rho] \gets_{p,\mathcal{Q}} s  
\to_{q,\mathcal{R}} C[l\sigma,r'\rho]_{p,q} = t_1$
for some $l \to r \in \mathcal{Q}, l' \to r' \in \mathcal{R}$.
Then we have 
$t_0 = s[r\sigma,l'\rho]_{p,q} \to_{q,\mathcal{R}} s[r\sigma,r'\rho]_{p,q}
\gets_{p,\mathcal{Q}} s[l\sigma,r'\rho]_{p,q} = t_1$.
Since 
${\stackrel{+}{\to}_\mathcal{R}}
\circ {\stackrel{=}{\gets}_\mathcal{Q}}
\circ {\stackrel{*}{\gets}_\mathcal{R}}
\subseteq {\blacktriangleright}$ 
(${\stackrel{+}{\to}_\mathcal{R}}
\circ {\pargets}_\mathcal{Q}
\circ {\stackrel{*}{\gets}_\mathcal{R}}
\subseteq {\blacktriangleright}$ in the proof of (2))
by our assumption,
$t_0 \blacktriangleright t_1$.

\item $p \le q$.
Let $s/p = l\sigma$ and $l \to r \in \mathcal{Q}$.
Then either 
(a) $ p/q  \in \mathrm{Pos}_\mathcal{F}(l)$
or 
(b) there exists $p_x \in \mathrm{Pos}_\mathcal{V}(l)$
such that $l/p_x = x \in \mathcal{V}$ and $p.p_x \le q$.

\begin{enumerate}[(a)]
\item
Then $t_0 = s[v\rho]_p$ 
and $t_1 = s[u\rho]_p$
for some $\langle u, v \rangle \in \mathrm{CP}(\mathcal{R},\mathcal{Q})$
and substitution $\rho$.
Thus by assumption 
$v \blacktriangleright u$.
Since $\blacktriangleright$ is a rewrite relation,
we have $t_0 = s[v\rho]_p \blacktriangleright s[u\rho]_p = t_1$.

\item
Then $t_0 = s[r\sigma]_p$
and we have $t_1 \stackrel{*}{\to}_\mathcal{R} s[l\sigma']_p$ 
for some substitution $\sigma'$
such that $\sigma(x) \to_{ q/(p.p_x),\mathcal{R}} \sigma'(x)$
and $\sigma(y) = \sigma'(y)$ for any $y \neq x$.
Since $\mathcal{Q}$ is bidirectional,
we have $\mathcal{V}(l) = \mathcal{V}(r)$.
Thus 
$t_0 = s[r\sigma]_p
\stackrel{+}{\to}_\mathcal{R} s[r\sigma']_p
\gets_\mathcal{Q}
s[l\sigma']_p 
\stackrel{*}{\gets}_\mathcal{R} t_1$.
Since ${\stackrel{+}{\to}_\mathcal{R}}
\circ {\stackrel{=}{\gets}_\mathcal{Q}}
\circ {\stackrel{*}{\gets}_\mathcal{R}}
\subseteq {\blacktriangleright}$
(${\stackrel{+}{\to}_\mathcal{R}}
\circ {\pargets}_\mathcal{Q}
\circ {\stackrel{*}{\gets}_\mathcal{R}}
\subseteq {\blacktriangleright}$ in the proof of (2))
by our assumption,
$t_0 \blacktriangleright t_1$.

\end{enumerate}

\item $p > q$.
Let $s/q = l'\rho$ and $l' \to r' \in \mathcal{R}$.
Then either 
(a) $p/ q \in \mathrm{Pos}_\mathcal{F}(l')$
or 
(b) there exists $q_x \in \mathrm{Pos}_\mathcal{V}(l')$
such that $l'/q_x = x \in \mathcal{V}$ and $q.q_x \le p$.

\begin{enumerate}[(a)]
\item
Then $t_0 = s[u\sigma]_q$ 
and $t_1 = s[v\sigma]_q$
for some $\langle u, v \rangle \in \mathrm{CP}(\mathcal{Q},\mathcal{R})$
and substitution $\sigma$.
By assumption 
$u \blacktriangleright v$.
Since $\blacktriangleright$ is a rewrite relation,
$t_0 = s[u\sigma]_q \blacktriangleright s[v\sigma]_q = t_1$.

\item
Then $t_1 = s[r'\rho]_q$, 
and by the left-linearity of $\mathcal{R}$,
$t_0 = s[l'\rho']_q$ for some substitution $\rho'$
such that $\rho(x) \to_{p /(q.q_x),\mathcal{Q}} \rho'(x)$
and $\rho(y) = \rho'(y)$ for any $y \neq x$.
Furthermore, by the right-linearity of $\mathcal{R}$,
$s[r'\rho]_q \stackrel{=}{\to}_\mathcal{Q} s[r'\rho']_q$.
(In the proof of (2),
we have $s[r'\rho]_q \parto_\mathcal{Q} s[r'\rho']_q$.)
Thus, 
$t_0 = s[l'\rho']_q
\to_{q,\mathcal{R}} s[r'\rho']_q
\stackrel{=}{\gets}_\mathcal{Q}
s[r'\rho]_q = t_1$
($t_0 = s[l'\rho']_q
\to_{q,\mathcal{R}} s[r'\rho']_q
\pargets_\mathcal{Q}
s[r'\rho]_q = t_1$ in the proof of (2)).
Since ${\stackrel{+}{\to}_\mathcal{R}}
\circ {\stackrel{=}{\gets}_\mathcal{Q}}
\circ {\stackrel{*}{\gets}_\mathcal{R}}
\subseteq {\blacktriangleright}$ 
(${\stackrel{+}{\to}_\mathcal{R}}
\circ {\pargets}_\mathcal{Q}
\circ {\stackrel{*}{\gets}_\mathcal{R}}
\subseteq {\blacktriangleright}$ in the proof of (2))
by our assumption,
$t_0 \blacktriangleright t_1$.\qed
\end{enumerate}
\end{enumerate}

\subsection{Confluence criterion for linear TRSs}

In this subsection, we give a 
confluence criterion for TRSs
that can be partitioned into 
linear terminating TRS $\mathcal{S}$
and reversible TRS $\mathcal{P}$.
We then discuss the possibility of
relaxing linearity condition
to left-linearity.

The next two lemmas are corollaries of 
Lemmas \ref{lem:main-I} and \ref{lem:main-II},
respectively.

\begin{lem}
\label{lem:linear-I}
Let $\mathcal{P},\mathcal{S},\mathcal{P}'$ be TRSs
such that $\mathcal{P}$ is bidirectional.
Suppose $\mathrm{CP}(\mathcal{S},\mathcal{S})
\subseteq 
{\stackrel{*}{\to}}_{\mathcal{S}\cup\mathcal{P}'}
\circ {\stackrel{=}{\gets}}_{\mathcal{P} \cup \mathcal{P}^{-1}} 
\circ {\stackrel{*}{\gets}}_{\mathcal{S}\cup\mathcal{P}'}$.
Then 
${\gets}_\mathcal{S} \circ {\to}_\mathcal{S}
\subseteq 
{\stackrel{*}{\to}}_{\mathcal{S}\cup\mathcal{P}'}
\circ {\stackrel{=}{\gets}}_{\mathcal{P} \cup \mathcal{P}^{-1}} 
\circ {\stackrel{*}{\gets}}_{\mathcal{S}\cup\mathcal{P}'}$.
\end{lem}

\begin{proof}
Take ${\blacktriangleright} := 
{\stackrel{*}{\to}}_{\mathcal{S}\cup\mathcal{P}'}
\circ {\stackrel{=}{\gets}}_{\mathcal{P} \cup \mathcal{P}^{-1}} 
\circ {\stackrel{*}{\gets}}_{\mathcal{S}\cup\mathcal{P}'}$.
Then $\mathrm{CP}(\mathcal{S},\mathcal{S})
\subseteq {\blacktriangleright} \cap {\blacktriangleleft}$
as ${\blacktriangleright} \cap {\blacktriangleleft}
= {\blacktriangleright}$.
Furthermore, we have ${\stackrel{*}{\to}}_\mathcal{S}
\circ {\stackrel{*}{\gets}}_\mathcal{S}
\subseteq {\blacktriangleright}$.
Hence the claim follows from Lemma~\ref{lem:main-I}.
\end{proof}

\begin{lem}
\label{lem:linear-III}
Let $\mathcal{P},\mathcal{S},\mathcal{P}'$ be TRSs
such that $\mathcal{P}$ is bidirectional
and $\mathcal{S}$ is linear.
Let ${\blacktriangleright}
=
(
{\stackrel{=}{\gets}}_{\mathcal{P} \cup \mathcal{P}^{-1}} 
\circ 
{\stackrel{*}{\gets}}_{\mathcal{S} \cup \mathcal{P}'})
\cup
({\to}_{\mathcal{S}}
\circ 
{\stackrel{*}{\to}}_{\mathcal{S} \cup \mathcal{P}'}
\circ 
{\stackrel{=}{\gets}}_{\mathcal{P} \cup \mathcal{P}^{-1}} 
\circ 
{\stackrel{*}{\gets}}_{\mathcal{S} \cup \mathcal{P}'})
$.
Suppose 
$\mathrm{CP}(\mathcal{P} \cup \mathcal{P}^{-1},\mathcal{S})
\subseteq {\blacktriangleright}$ and
$\mathrm{CP}(\mathcal{S},\mathcal{P} \cup \mathcal{P}^{-1})
\subseteq {\blacktriangleleft}$.
Then 
${\gets}_{\mathcal{P}\cup \mathcal{P}^{-1}}
 \circ {\to}_\mathcal{S}
\subseteq {\blacktriangleright}$.
\end{lem}

\begin{proof}
Take $\mathcal{Q} := \mathcal{P} \cup  \mathcal{P}^{-1}$,
which is a bidirectional TRS by bidirectionality of $\mathcal{P}$,
and $\mathcal{R} := \mathcal{S}$ in Lemma~\ref{lem:main-II}.
Then by the condition,
$\mathrm{CP}(\mathcal{P} \cup \mathcal{P}^{-1},\mathcal{S})
\subseteq {\blacktriangleright}$ and
$\mathrm{CP}(\mathcal{S},\mathcal{P} \cup \mathcal{P}^{-1})
\subseteq {\blacktriangleleft}$.
Furthermore,
${\stackrel{+}{\to}_\mathcal{R}}
\circ {\stackrel{=}{\gets}_\mathcal{Q}}
\circ {\stackrel{*}{\gets}_\mathcal{R}}
=
{\stackrel{+}{\to}}_{\mathcal{S}}
\circ 
{\stackrel{=}{\gets}}_{\mathcal{P} \cup \mathcal{P}^{-1}}
\circ 
{\stackrel{*}{\gets}}_{\mathcal{S}}
\subseteq {\blacktriangleright}$.
Hence the claim follows from Lemma~\ref{lem:main-II}~(1).
\end{proof}

We arrive at our first criterion for confluence.

\begin{thm}[confluence criterion for linear $\mathcal{S}$]
\label{thm:linear}
Let $\mathcal{P},\mathcal{S},\mathcal{P}'$ be TRSs
such that $\mathcal{S}$ is linear,
$\mathcal{P}$ is reversible,
$\mathcal{P}' \subseteq \mathcal{P} \cup \mathcal{P}^{-1}$
and 
$\mathcal{S}$ is terminating relative to $\mathcal{P}'$.
Suppose
(i) $\mathrm{CP}(\mathcal{S},\mathcal{S})
\subseteq 
{\stackrel{*}{\to}}_{\mathcal{S}\cup \mathcal{P}'} 
\circ 
{\stackrel{=}{\gets}}_{\mathcal{P} \cup \mathcal{P}^{-1}}
\circ
{\stackrel{*}{\gets}}_{\mathcal{S}\cup\mathcal{P'}}$,
(ii) 
$\mathrm{CP}(\mathcal{P} \cup \mathcal{P}^{-1},\mathcal{S})
\subseteq 
(
{\stackrel{=}{\gets}}_{\mathcal{P} \cup \mathcal{P}^{-1}}
\circ 
{\stackrel{*}{\gets}}_{\mathcal{S} \cup \mathcal{P}'})
\cup
({\to}_{\mathcal{S}}
\circ 
{\stackrel{*}{\to}}_{\mathcal{S} \cup \mathcal{P}'}
\circ 
{\stackrel{=}{\gets}}_{\mathcal{P} \cup \mathcal{P}^{-1}}
\circ 
{\stackrel{*}{\gets}}_{\mathcal{S} \cup \mathcal{P}'})
$ 
and
(iii) $\mathrm{CP}(\mathcal{S},\mathcal{P} \cup \mathcal{P}^{-1})
\subseteq 
({\stackrel{*}{\to}}_{\mathcal{S} \cup \mathcal{P}'}
\circ 
{\stackrel{=}{\to}}_{\mathcal{P} \cup \mathcal{P}^{-1}}
)
\cup
({\stackrel{*}{\to}}_{\mathcal{S} \cup \mathcal{P}'}
\circ 
{\stackrel{=}{\to}}_{\mathcal{P} \cup \mathcal{P}^{-1}}
\circ 
{\stackrel{*}{\gets}}_{\mathcal{S} \cup \mathcal{P}'}
\circ 
{\gets}_{\mathcal{S}})
$.
Then 
$\mathcal{S}\cup \mathcal{P}$ is confluent.
\end{thm}

\begin{proof}
From Lemma~\ref{lem:linear-I} and our assumption (i),
we have 
(a) ${\gets}_\mathcal{S} \circ {\to}_\mathcal{S}
\subseteq 
{\stackrel{*}{\to}}_{\mathcal{S}\cup \mathcal{P}'} 
\circ 
{\stackrel{=}{\gets}}_{\mathcal{P} \cup \mathcal{P}^{-1}}
\circ
{\stackrel{*}{\gets}}_{\mathcal{S}\cup\mathcal{P'}}$,
From Lemma~\ref{lem:linear-III} and our assumptions (ii), (iii),
we have
(b) ${\gets}_{\mathcal{P}\cup \mathcal{P}^{-1}} \circ {\to}_\mathcal{S}
\subseteq 
(
{\stackrel{=}{\gets}}_{\mathcal{P} \cup \mathcal{P}^{-1}}
\circ 
{\stackrel{*}{\gets}}_{\mathcal{S} \cup \mathcal{P}'})
\cup
({\to}_{\mathcal{S}}
\circ 
{\stackrel{*}{\to}}_{\mathcal{S} \cup \mathcal{P}'}
\circ 
{\stackrel{=}{\gets}}_{\mathcal{P} \cup \mathcal{P}^{-1}}
\circ 
{\stackrel{*}{\gets}}_{\mathcal{S} \cup \mathcal{P}'})
$.
Note that, by the definition of rewrite rules,
reversible TRSs are bidirectional.
Take ${\vdashv} := {\leftrightarrow}_\mathcal{P}
= {\gets}_{\mathcal{P}\cup \mathcal{P}^{-1}}$, 
${\leadsto} := {\to}_\mathcal{P'}$ and ${\to} := {\to}_\mathcal{S}$.
Then since $\mathcal{S}$ is terminating relative to $\mathcal{P'}$,
the relation ${\to} \circ \stackrel{*}{\leadsto}$ is well-founded.
Thus one can apply Theorem \ref{thm:ARS}
so as to prove ${\to}_\mathcal{S}$ is Church-Rosser modulo
$\stackrel{*}{\leftrightarrow}_\mathcal{P}$.
Since $\mathcal{P}$ is reversible,
${\to}_{\mathcal{S} \cup \mathcal{P}}$ is confluent
by Lemma \ref{lem:reversible + CR modulo}.
\end{proof}

\begin{rem}
Conditions like (i)--(iii) are referred to as 
\textit{critical pair conditions} in the sequel.
\end{rem}

By taking $\mathcal{P}' = \emptyset$ in 
Theorem~\ref{thm:linear},
we obtain the next corollary.

\begin{cor}[Theorem 1 of \cite{JKR83}]
\label{cor:linear}
Let $\mathcal{P},\mathcal{S}$ be TRSs
such that $\mathcal{S}$ is linear,
$\mathcal{P}$ is reversible
and 
$\mathcal{S}$ is terminating.
Suppose
(i) $\mathrm{CP}(\mathcal{S},\mathcal{S})
\subseteq 
{\stackrel{*}{\to}}_{\mathcal{S}}
\circ 
{\stackrel{=}{\gets}}_{\mathcal{P} \cup \mathcal{P}^{-1}}
\circ
{\stackrel{*}{\gets}}_{\mathcal{S}}$,
(ii) 
$\mathrm{CP}(\mathcal{P} \cup \mathcal{P}^{-1},\mathcal{S})
\subseteq 
{\stackrel{*}{\to}}_{\mathcal{S}}
\circ 
{\stackrel{=}{\gets}}_{\mathcal{P} \cup \mathcal{P}^{-1}}
\circ 
{\stackrel{*}{\gets}}_{\mathcal{S}}
$ 
and
(iii) $\mathrm{CP}(\mathcal{S},\mathcal{P} \cup \mathcal{P}^{-1})
\subseteq 
{\stackrel{*}{\to}}_{\mathcal{S}}
\circ 
{\stackrel{=}{\to}}_{\mathcal{P} \cup \mathcal{P}^{-1}}
\circ 
{\stackrel{*}{\gets}}_{\mathcal{S}}
$.
Then 
$\mathcal{S}\cup \mathcal{P}$ is confluent.
\end{cor}

To generalize Theorem \ref{thm:linear}
for (possibly not right-linear) 
left-linear $\mathcal{S}$,
we have to use Lemma \ref{lem:main-II} (2)
instead of Lemma \ref{lem:main-II} (1).
For this, we needs condition 
${\stackrel{+}{\to}_\mathcal{R}}
\circ {\pargets_\mathcal{Q}}
\circ {\stackrel{*}{\gets}_\mathcal{R}}
\subseteq {\blacktriangleright}$.
However, this fact reduces the application of Theorem \ref{thm:ARS}
to just an application of Corollary \ref{cor:ARS-II'}
(and hence that of Corollary \ref{cor:ARS-III}),
because the $\eqvdashv$-step in the condition does not contribute
and hence we have to take ${\leadsto} := {\vdashv}$.
In the rest of this subsection,
we sketch how a proof 
analogous to the linear case
can be applied to obtain a confluence criterion
based on the Corollary \ref{cor:ARS-III}.

\begin{lem}
\label{lem:relative-I}
Let $\mathcal{P},\mathcal{S}$ be TRSs
such that $\mathcal{P}$ is bidirectional.
Suppose $\mathrm{CP}(\mathcal{S},\mathcal{S})
\subseteq 
{\stackrel{*}{\to}}_{\mathcal{S}}
\circ 
{\stackrel{*}{\gets}}_{\mathcal{P} \cup \mathcal{P}^{-1}}
\circ 
{\stackrel{*}{\gets}}_{\mathcal{S}}$.
Then 
${\gets}_\mathcal{S} \circ {\to}_\mathcal{S}
\subseteq 
{\stackrel{*}{\to}}_{\mathcal{S}}
\circ 
{\stackrel{*}{\gets}}_{\mathcal{P} \cup \mathcal{P}^{-1}}
\circ 
{\stackrel{*}{\gets}}_{\mathcal{S}}$.
\end{lem}

\begin{proof}
Take ${\blacktriangleright} := 
{\stackrel{*}{\to}}_{\mathcal{S}}
\circ 
{\stackrel{*}{\gets}}_{\mathcal{P} \cup \mathcal{P}^{-1}}
\circ 
{\stackrel{*}{\gets}}_{\mathcal{S}}$
in Lemma~\ref{lem:main-I}.
\end{proof}

\begin{lem}
\label{lem:relative-II}
Let $\mathcal{P},\mathcal{S}$ be TRSs
such that $\mathcal{P}$ is bidirectional
and $\mathcal{S}$ is left-linear.
Suppose 
(i) $\mathrm{CP}(\mathcal{P}\cup\mathcal{P}^{-1},\mathcal{S})
\subseteq 
{\stackrel{+}{\to}}_{\mathcal{S}}
\circ 
{\stackrel{*}{\gets}}_{\mathcal{P} \cup \mathcal{P}^{-1}}
\circ 
{\stackrel{*}{\gets}}_{\mathcal{S}}$
and
(ii) $\mathrm{CP}(\mathcal{S},\mathcal{P}\cup\mathcal{P}^{-1})
\subseteq 
{\stackrel{*}{\to}}_{\mathcal{S}}
\circ 
{\stackrel{*}{\to}}_{\mathcal{P} \cup \mathcal{P}^{-1}}
\circ 
{\stackrel{+}{\gets}}_{\mathcal{S}}$.
Then 
${\gets}_{\mathcal{P} \cup \mathcal{P}^{-1}}  \circ {\to}_\mathcal{S}
\subseteq 
{\stackrel{+}{\to}}_{\mathcal{S}}
\circ 
{\stackrel{*}{\gets}}_{\mathcal{P} \cup \mathcal{P}^{-1}}
\circ 
{\stackrel{*}{\gets}}_{\mathcal{S}}$.
\end{lem}

\begin{proof}
Take ${\blacktriangleright} := 
{\stackrel{+}{\to}}_{\mathcal{S}}
\circ 
{\stackrel{*}{\gets}}_{\mathcal{P} \cup \mathcal{P}^{-1}}
\circ 
{\stackrel{*}{\gets}}_{\mathcal{S}}$
in Lemma~\ref{lem:main-II} (2).
\end{proof}

\begin{prop}[Theorem 3.3 of \cite{Hue80}]
\label{prop:relative}
Let $\mathcal{P},\mathcal{S}$ be TRSs
such that $\mathcal{S}$ is left-linear,
$\mathcal{P}$ is reversible and
$\mathcal{S}$ is terminating relative to $\mathcal{P}$.
Suppose
(i) $\mathrm{CP}(\mathcal{S},\mathcal{S})
\subseteq 
{\stackrel{*}{\to}}_{\mathcal{S}}
\circ 
{\stackrel{*}{\gets}}_{\mathcal{P} \cup \mathcal{P}^{-1}}
\circ 
{\stackrel{*}{\gets}}_{\mathcal{S}}$
(ii) $\mathrm{CP}(\mathcal{P} \cup \mathcal{P}^{-1},\mathcal{S})
\subseteq 
{\stackrel{+}{\to}}_{\mathcal{S}}
\circ 
{\stackrel{*}{\gets}}_{\mathcal{P} \cup \mathcal{P}^{-1}}
\circ 
{\stackrel{*}{\gets}}_{\mathcal{S}}$
and
(iii) $\mathrm{CP}(\mathcal{S},\mathcal{P}\cup\mathcal{P}^{-1})
\subseteq 
{\stackrel{*}{\to}}_{\mathcal{S}}
\circ 
{\stackrel{*}{\to}}_{\mathcal{P} \cup \mathcal{P}^{-1}}
\circ 
{\stackrel{+}{\gets}}_{\mathcal{S}}$.
Then 
$\mathcal{S}\cup \mathcal{P}$ is confluent.
\end{prop}

\begin{proof}
Since $\mathcal{P}$ is reversible,
$\mathcal{S}$ is terminating modulo
$\mathcal{P} \cup \mathcal{P}^{-1}$.
By Lemmas~\ref{lem:relative-I} and \ref{lem:relative-II}
and Corollary~\ref{cor:ARS-III},
${\to}_\mathcal{S}$ is Church-Rosser modulo
$\stackrel{*}{\leftrightarrow}_\mathcal{P}$.
Since $\mathcal{P}$ is reversible,
${\to}_{\mathcal{S} \cup \mathcal{P}}$ is confluent
by Lemma \ref{lem:reversible + CR modulo}.
\end{proof}

\begin{rem}
It is straightforward to modify
Lemmas~\ref{lem:relative-I} and \ref{lem:relative-II}
and use either Corollary~\ref{cor:ARS-II'} (or Corollary~\ref{cor:Huet80})
to replace conditions (i)--(iii) 
of Proposition~\ref{prop:relative}
with 
(i) $\mathrm{CP}(\mathcal{S},\mathcal{S})
\subseteq 
{\stackrel{*}{\to}}_{\mathcal{S}\cup \mathcal{P} \cup \mathcal{P}^{-1}}
\circ 
{\stackrel{*}{\gets}}_{\mathcal{S}\cup \mathcal{P} \cup \mathcal{P}^{-1}}$
(ii) $\mathrm{CP}(\mathcal{P} \cup \mathcal{P}^{-1},\mathcal{S})
\subseteq 
{\to}_\mathcal{S}
\circ 
{\stackrel{*}{\to}}_{\mathcal{S}\cup \mathcal{P} \cup \mathcal{P}^{-1}}
\circ 
{\stackrel{*}{\gets}}_{\mathcal{S}\cup \mathcal{P} \cup \mathcal{P}^{-1}}$,
and
(iii) $\mathrm{CP}(\mathcal{S},\mathcal{P}\cup\mathcal{P}^{-1})
\subseteq 
{\stackrel{*}{\to}}_{\mathcal{S}\cup \mathcal{P} \cup \mathcal{P}^{-1}}
\circ 
{\stackrel{*}{\gets}}_{\mathcal{S}\cup \mathcal{P} \cup \mathcal{P}^{-1}}
\circ 
{\gets}_\mathcal{S}$, respectively
(or with
(i) $\mathrm{CP}(\mathcal{S},\mathcal{S})
\subseteq 
{\stackrel{*}{\to}}_{\mathcal{S}}
\circ 
{\stackrel{*}{\gets}}_{\mathcal{P} \cup \mathcal{P}^{-1}}
\circ 
{\stackrel{*}{\gets}}_{\mathcal{S}}$,
(ii) $\mathrm{CP}(\mathcal{P} \cup \mathcal{P}^{-1},\mathcal{S})
\subseteq 
{\stackrel{*}{\to}}_{\mathcal{S}}
\circ 
{\stackrel{*}{\gets}}_{\mathcal{P} \cup \mathcal{P}^{-1}}
\circ 
{\stackrel{*}{\gets}}_{\mathcal{S}}$,
and
(iii) $\mathrm{CP}(\mathcal{S},\mathcal{P}\cup\mathcal{P}^{-1})
\subseteq 
{\stackrel{*}{\to}}_{\mathcal{S}}
\circ 
{\stackrel{*}{\to}}_{\mathcal{P} \cup \mathcal{P}^{-1}}
\circ 
{\stackrel{*}{\gets}}_{\mathcal{S}}$, respectively).
Similar to the abstract case, 
such replacements do not strengthen or weaken
the applicability of the proposition
(c.f.\ Remark~\ref{rem:Differnce with Huet}).
\end{rem}

\subsection{Confluence criterion based on parallel rewrite steps}

As discussed in the previous subsection,
if we put ${\vdashv} := {\leftrightarrow_\mathcal{P}}
= {\gets_{\mathcal{P}\cup\mathcal{P}^{-1}}}$,
the application of our abstract criterion (Theorem \ref{thm:ARS}) 
to the left-linear case reduces to the application of Corollary \ref{cor:ARS-III}.
In this subsection, we relax the linear condition
of the $\mathcal{S}$-part to left-linear by
considering ${\vdashv} := {\pargets_{\mathcal{P}\cup \mathcal{P}^{-1}}}$
instead of ${\vdashv} := {\gets_{\mathcal{P}\cup\mathcal{P}^{-1}}}$.
This allows us to partially recover the applicability of Theorem \ref{thm:ARS}.

The next lemma is analogous to Lemma \ref{lem:linear-I},
which is obtained from Lemma \ref{lem:main-I} again.

\begin{lem}
\label{lem:parallel-I}
Let $\mathcal{P},\mathcal{S},\mathcal{P}'$ be TRSs
such that $\mathcal{P}$ is bidirectional.
Suppose that 
$\mathrm{CP}(\mathcal{S},\mathcal{S})
\subseteq 
{\stackrel{*}{\to}}_{\mathcal{S} \cup \mathcal{P}'} 
\circ {\pargets}_{\mathcal{P} \cup \mathcal{P}^{-1}}
\circ {\stackrel{*}{\gets}}_{\mathcal{S} \cup \mathcal{P}'}$.
Then 
${\gets}_\mathcal{S} \circ {\to}_\mathcal{S}
\subseteq 
{\stackrel{*}{\to}}_{\mathcal{S} \cup \mathcal{P}'}
\circ
{\pargets}_{\mathcal{P} \cup \mathcal{P}^{-1}}
\circ 
{\stackrel{*}{\gets}}_{\mathcal{S} \cup \mathcal{P}'}$.
\end{lem}

\begin{proof}
Take ${\blacktriangleright} := 
{\stackrel{*}{\to}}_{\mathcal{S} \cup \mathcal{P}'} 
\circ {\pargets}_{\mathcal{P} \cup \mathcal{P}^{-1}}
\circ {\stackrel{*}{\gets}}_{\mathcal{S} \cup \mathcal{P}'}$.
Then $\mathrm{CP}(\mathcal{S},\mathcal{S})
\subseteq {\blacktriangleright} \cap {\blacktriangleleft}$ and
${\stackrel{*}{\to}}_\mathcal{S}
\circ {\stackrel{*}{\gets}}_\mathcal{S}
\subseteq {\blacktriangleright}$.
Hence the claim follows from Lemma~\ref{lem:main-I}.
\end{proof}

To present an analogy of Lemma \ref{lem:linear-III},
we first present the parametrized version of the lemma
in the same spirit as Lemma \ref{lem:main-II}.

\begin{lem}
\label{lem:parallel-II}
Let $\mathcal{Q},\mathcal{R}$ be TRSs
such that $\mathcal{Q}$ is bidirectional
and $\mathcal{R}$ is left-linear.
Let $\blacktriangleright$ be a rewrite relation
such that 
$s_i \blacktriangleright t_i$ and
$s_j \gets_\mathcal{Q} t_j$ for any $j \in \{ 1,\ldots,n \} \setminus \{ i \}$
implies
$C[s_1,\ldots,s_n] \blacktriangleright C[t_1,\ldots,t_n]$.
Suppose
(i) $\mathrm{CP}_\mathit{in}(\mathcal{Q},\mathcal{R})
= \emptyset$ and
(ii) 
$\mathrm{CP}(\mathcal{R},\mathcal{Q})
\subseteq {\blacktriangleleft}$.
If 
${\stackrel{+}{\to}}_\mathcal{R}
\circ {\pargets}_\mathcal{Q} \circ
{\stackrel{*}{\gets}}_\mathcal{R}
\subseteq {\blacktriangleright}$
then 
${\pargets}_\mathcal{Q} \circ {\to}_\mathcal{R}
\subseteq 
{\blacktriangleright}$.
\end{lem}

\proof
Suppose
$t_0 \pargets_{U,\mathcal{Q}} s \to_{q,\mathcal{R}} t_1$.
Let $U = \{ p_1,\ldots, p_n \}$ where $p_1,\ldots, p_n$ are positions
from left to right,
$s/p_i = l_i\sigma_i$ for $l_i \to r_i \in \mathcal{Q}$ 
and substitutions $\sigma_i$
($1 \le i \le n$)
and $s/q = l'\rho$ for $l' \to r' \in \mathcal{R}$ and a substitution $\rho$.
We distinguish two cases:
(1) the case that $\exists p \in U.~p \le q$ and
(2) the case that $\forall p \in U.~p \not\le q$.
\begin{enumerate}[(1)]
\item Suppose $p_i \in U$ and $p_i \le q$.
Then either 
(a) $q/p_i  \in \mathrm{Pos}_\mathcal{F}(l_i)$
or 
(b) there exists $p_x \in \mathrm{Pos}_\mathcal{V}(l_i)$
such that $l_i/p_x = x \in \mathcal{V}$ and $p_i.p_x \le q$.
\begin{enumerate}[(a)]
\item
Then $t_0/p_i = v\rho$ and $t_1/p_i = u\rho$
for some $\langle u, v \rangle \in \mathrm{CP}(\mathcal{R},\mathcal{Q})$
and substitution $\rho$.
Then, from our assumption (ii),
we have 
$v \blacktriangleright u$.
Since $\blacktriangleright$ is a rewrite relation,
$t_0/p_i = v\rho \blacktriangleright u\rho = t_1/p_i$.
By the assumption on $\blacktriangleright$,
since for any $j \neq i$,
$s/p_j = r_j\sigma_j \gets_\mathcal{Q} 
l_j\sigma_j = t_1/p_j$, we have
$t_0 = s[r_1\sigma_1,\ldots,t_0/p_i,\ldots,r_n\sigma_n]_{p_1,\ldots,p_i,\ldots,p_n}
\blacktriangleright
s[l_1\sigma_1,\ldots,t_1/p_i,\ldots,l_n\sigma_n]_{p_1,\ldots,p_i,\ldots,p_n}
= t_1$.

\item
Then we have $t_0/p_i = r_i\sigma_i$ and 
$t_1/p_i \stackrel{*}{\to}_\mathcal{R} l_i\sigma_i'$ 
for some substitution $\sigma_i'$ 
such that $\sigma_i(x) \to_{q/(p_i.p_x),\mathcal{R}} \sigma_i'(x)$
and $\sigma_i'(y) = \sigma_i(y)$ for any $y \neq x$.
Since $\mathcal{Q}$ is bidirectional,
we have $\mathcal{V}(l_i) = \mathcal{V}(r_i)$.
Thus we have
\[
\begin{array}{lllll}
t_0 
& =& C[r_1\sigma_1,\ldots,r_i\sigma_i,\ldots,r_n\sigma_n]_{p_1,\ldots,p_i,\ldots,p_n}\\
&\stackrel{+}{\to}_\mathcal{R}&
C[r_1\sigma_1,\ldots,r_i\sigma_i',\ldots,r_n\sigma_n]_{p_1,\ldots,p_i,\ldots,p_n}\\
&\pargets_\mathcal{Q}&
C[l_1\sigma_1,\ldots,l_i\sigma_i',\ldots,l_n\sigma_n]_{p_1,\ldots,p_i,\ldots,p_n}\\
& \stackrel{*}{\gets}_\mathcal{R}&
C[l_1\sigma_1,\ldots,t_1/p_i,\ldots,l_n\sigma_n]_{p_1,\ldots,p_i,\ldots,p_n}
= t_1.
\end{array}
\]
From our assumption that
${\stackrel{+}{\to}}_\mathcal{R}
\circ {\pargets}_\mathcal{Q} \circ 
{\stackrel{*}{\gets}}_\mathcal{R}
\subseteq {\blacktriangleright}$
it follows $t_0 \blacktriangleright t_1$.

\end{enumerate}

\item Suppose $\forall p \in U.~p \not\le q$.
Let $U' = \{ p_i \in U \mid q < p_i \} = \{ p_l, \ldots, p_k \}$,
$q_i =  p_i/q$ for $l \le i \le k$,
and thus $l'\rho = l'\rho[l_l\sigma_l,\ldots,l_k\sigma_k]_{q_l,\ldots,q_k}$.
By our assumption (i),
for each $p_i \in U'$
there exists $q_x \in \mathrm{Pos}_\mathcal{V}(l')$
such that $l'/q_x = x \in \mathcal{V}$ and $q.q_x \le p_i$.
Thus, $s/q = l'\rho = l'\rho[l_l\sigma_l,\ldots,l_k\sigma_k]_{q_l,\ldots,q_k}
\to_\mathcal{R}
r'\rho = r'\rho[l_{j_1}\sigma_{j_1},\ldots,l_{j_m}\sigma_{j_m}]_{o_1,\ldots,o_m}
= t_1/q$
for some positions $o_1,\cdots,o_m$ and $j_1,\ldots,j_m \in \{l,\ldots, k\}$.
Furthermore, by the left-linearity of $\mathcal{R}$, we have
$l'\rho[r_l\sigma_l,\ldots,r_k\sigma_k]_{q_l,\ldots,q_k}
\to_\mathcal{R}
r'\rho[r_{j_1}\sigma_{j_1},\ldots,r_{j_m}\sigma_{j_m}]_{o_1,\ldots,o_m}$.
Thus, 
\[
\begin{array}{lllll}
t_0 & = &
s[r_1\sigma_1,\ldots,
              l'\rho[r_l\sigma_l,\ldots,r_k\sigma_k]_{q_l,\ldots,q_k},
              \ldots,r_n\sigma_n]_{p_1,\ldots,q,\ldots,p_n}\\
&\to_\mathcal{R}&
s[r_1\sigma_1,\ldots,
              r'\rho[r_{j_1}\sigma_{j_1},\ldots,r_{j_m}\sigma_{j_m}]_{o_1,\ldots,o_m},
              \ldots,r_n\sigma_n]_{p_1,\ldots,q,\ldots,p_n}\\
&\pargets_\mathcal{Q}&
s[l_1\sigma_1,\ldots,
              r'\rho[l_{j_1}\sigma_{j_1},\ldots,l_{j_m}\sigma_{j_m}]_{o_1,\ldots,o_m},
              \ldots,l_n\sigma_n]_{p_1,\ldots,q,\ldots,p_n} = t_1.
\end{array}
\]
From our assumption that
${\stackrel{+}{\to}}_\mathcal{R}
\circ {\pargets}_\mathcal{Q} \circ
{\stackrel{*}{\gets}}_\mathcal{R}
\subseteq {\blacktriangleright}$
follows $t_0 \blacktriangleright t_1$.\qed\smallskip
\end{enumerate}

\noindent The next lemma is an analogy of Lemma \ref{lem:linear-III}
based on parallel steps.

\begin{lem}
\label{lem:parallel-III}
Let $\mathcal{P},\mathcal{S},\mathcal{P}'$ be TRSs
such that $\mathcal{P}$ is bidirectional
and $\mathcal{S}$ is left-linear.
Let ${\blacktriangleright}
=
(
{\pargets}_{\mathcal{P} \cup \mathcal{P}^{-1}}
\circ 
{\stackrel{*}{\gets}}_{\mathcal{S} \cup \mathcal{P}'})
\cup
({\to}_{\mathcal{S}}
\circ 
{\stackrel{*}{\to}}_{\mathcal{S} \cup \mathcal{P}'}
\circ 
{\pargets}_{\mathcal{P} \cup \mathcal{P}^{-1}}
\circ 
{\stackrel{*}{\gets}}_{\mathcal{S} \cup \mathcal{P}'})
$.
Suppose 
$\mathrm{CP}_\mathit{in}(\mathcal{P} \cup \mathcal{P}^{-1},\mathcal{S})
= \emptyset$ and
$\mathrm{CP}(\mathcal{S},\mathcal{P} \cup \mathcal{P}^{-1})
\subseteq {\blacktriangleleft}$.
Then 
${\pargets}_{\mathcal{P}\cup \mathcal{P}^{-1}}
\circ {\to}_\mathcal{S}
\subseteq 
{\blacktriangleright}$.
\end{lem}

\begin{proof}
Take $\mathcal{Q} := \mathcal{P} \cup  \mathcal{P}^{-1}$,
which is a bidirectional TRS by bidirectionality of $\mathcal{P}$,
and $\mathcal{R} := \mathcal{S}$ in Lemma~\ref{lem:parallel-II}.
Then by the condition,
$\mathrm{CP}_\mathit{in}(\mathcal{Q},\mathcal{R}) = \emptyset$ and
$\mathrm{CP}(\mathcal{R},\mathcal{Q})
\subseteq {\blacktriangleleft}$.
Furthermore,
${\stackrel{+}{\to}}_\mathcal{R}
\circ {\pargets}_\mathcal{Q}
\circ {\stackrel{*}{\gets}}_\mathcal{R}
\subseteq 
{\to}_{\mathcal{S}}
\circ 
{\stackrel{*}{\to}}_{\mathcal{S} \cup \mathcal{P}'}
\circ 
{\pargets}_{\mathcal{P} \cup \mathcal{P}^{-1}}
\circ 
{\stackrel{*}{\gets}}_{\mathcal{S} \cup \mathcal{P}'}
\subseteq {\blacktriangleright}$
and 
$s_i \blacktriangleright t_i$ and
$s_j \gets_\mathcal{Q} t_j$ for any $j \in \{ 1,\ldots,n \} \setminus \{ i \}$
implies
$C[s_1,\ldots,s_n] \blacktriangleright C[t_1,\ldots,t_n]$.
Hence the claim follows from Lemma~\ref{lem:parallel-II}.
\end{proof}

We arrive at our second criterion for confluence.

\begin{thm}[confluence criterion based on parallel rewrite steps]
\label{thm:parallel}
Let $\mathcal{P},\mathcal{S},\mathcal{P}'$ be TRSs
such that $\mathcal{S}$ is left-linear,
$\mathcal{P}$ is reversible,
$\mathcal{P}' \subseteq \mathcal{P} \cup \mathcal{P}^{-1}$,
and 
$\mathcal{S}$ is terminating relative to $\mathcal{P}'$.
Suppose
(i) $\mathrm{CP}(\mathcal{S},\mathcal{S})
\subseteq 
{\stackrel{*}{\to}}_{\mathcal{S}\cup \mathcal{P}'} 
\circ 
{\pargets}_{\mathcal{P} \cup \mathcal{P}^{-1}}
\circ
{\stackrel{*}{\gets}}_{\mathcal{S}\cup\mathcal{P'}}$,
(ii) 
$\mathrm{CP}_\mathit{in}(\mathcal{P} \cup \mathcal{P}^{-1},\mathcal{S})
= \emptyset$
and
(iii) $\mathrm{CP}(\mathcal{S},\mathcal{P} \cup \mathcal{P}^{-1})
\subseteq 
(
{\stackrel{*}{\to}}_{\mathcal{S} \cup \mathcal{P}'}
\circ 
{\parto}_{\mathcal{P} \cup \mathcal{P}^{-1}}
)
\cup
({\stackrel{*}{\to}}_{\mathcal{S} \cup \mathcal{P}'}
\circ 
{\parto}_{\mathcal{P} \cup \mathcal{P}^{-1}}
\circ 
{\stackrel{*}{\gets}}_{\mathcal{S} \cup \mathcal{P}'}
\circ 
{\gets}_{\mathcal{S}}
)
$.
Then 
$\mathcal{S}\cup \mathcal{P}$ is confluent.
\end{thm}

\begin{proof}
By our assumption (i) and Lemma~\ref{lem:parallel-I},
we have (a) ${\gets}_\mathcal{S} \circ {\to}_\mathcal{S}
\subseteq 
{\stackrel{*}{\to}}_{\mathcal{S}\cup \mathcal{P}'} 
\circ 
{\pargets}_{\mathcal{P} \cup \mathcal{P}^{-1}}
\circ
{\stackrel{*}{\gets}}_{\mathcal{S}\cup\mathcal{P'}}$.
From our assumptions (ii) and (iii), 
it follows that
(b) 
${\pargets}_{\mathcal{P}\cup \mathcal{P}^{-1}} \circ {\to}_\mathcal{S}
\subseteq 
(
{\pargets}_{\mathcal{P} \cup \mathcal{P}^{-1}}
\circ 
{\stackrel{*}{\gets}}_{\mathcal{S} \cup \mathcal{P}'}
)
\cup
({\to}_{\mathcal{S}}
\circ 
{\stackrel{*}{\to}}_{\mathcal{S} \cup \mathcal{P}'}
\circ 
{\pargets}_{\mathcal{P} \cup \mathcal{P}^{-1}}
\circ 
{\stackrel{*}{\gets}}_{\mathcal{S} \cup \mathcal{P}'}
)
$ by Lemma~\ref{lem:parallel-III}.
Take ${\vdashv} := {\pargets}_{\mathcal{P}\cup\mathcal{P}^{-1}}$,
${\to} := {\to}_\mathcal{S}$
and 
${\leadsto} := {\to}_{\mathcal{P}'}$.
Then, by the termination of $\mathcal{S}$ relative to 
$\mathcal{P}'$, 
${\to} \circ {\stackrel{*}{\leadsto}}$ is well-founded.
Thus one can apply Theorem \ref{thm:ARS}
so as to prove ${\to}_\mathcal{S}$ is Church-Rosser modulo
$\stackrel{*}{\pargets}_{\mathcal{P}\cup\mathcal{P}^{-1}}$.
Since 
${\stackrel{*}{\pargets}}_{\mathcal{P}\cup\mathcal{P}^{-1}}
= {\stackrel{*}{\leftrightarrow}}_\mathcal{P}$,
it follows that ${\to}_\mathcal{S}$ is Church-Rosser modulo
$\stackrel{*}{\leftrightarrow}_\mathcal{P}$.
Hence, since $\to_\mathcal{P}$ is reversible,
${\to}_{\mathcal{S} \cup \mathcal{P}}$ is confluent
by Lemma \ref{lem:reversible + CR modulo}.
\end{proof}

Comparing to our first criterion for confluence (Theorem \ref{thm:linear}),
we impose the condition
$\mathrm{CP}_\mathit{in}(\mathcal{P} \cup \mathcal{P}^{-1},\mathcal{S})
= \emptyset$
while relaxing the linearity condition of $\mathcal{S}$ to left-linearity.
Hence Theorems \ref{thm:linear} and \ref{thm:parallel}
are incomparable.

By taking $\mathcal{P}' = \emptyset$ in 
Theorem~\ref{thm:parallel},
we obtain the next corollary.

\begin{cor}
\label{cor:parallel}
Let $\mathcal{P},\mathcal{S}$ be TRSs
such that $\mathcal{S}$ is left-linear,
$\mathcal{P}$ is reversible,
and 
$\mathcal{S}$ is terminating.
Suppose
(i) $\mathrm{CP}(\mathcal{S},\mathcal{S})
\subseteq 
{\stackrel{*}{\to}}_{\mathcal{S}} 
\circ 
{\pargets}_{\mathcal{P} \cup \mathcal{P}^{-1}}
\circ
{\stackrel{*}{\gets}}_{\mathcal{S}}$,
(ii) 
$\mathrm{CP}_\mathit{in}(\mathcal{P} \cup \mathcal{P}^{-1},\mathcal{S})
= \emptyset$
and
(iii) $\mathrm{CP}(\mathcal{S},\mathcal{P} \cup \mathcal{P}^{-1})
\subseteq 
{\stackrel{*}{\to}}_{\mathcal{S}}
\circ 
{\parto}_{\mathcal{P} \cup \mathcal{P}^{-1}}
\circ 
{\stackrel{*}{\gets}}_{\mathcal{S}}
$.
Then 
$\mathcal{S}\cup \mathcal{P}$ is confluent.
\end{cor}

\subsection{Confluence criterion based on parallel critical pairs}

In this subsection, we relax the 
condition (ii) 
$\mathrm{CP}_\mathit{in}(\mathcal{P}\cup \mathcal{P}^{-1},\mathcal{S})
= \emptyset$
of the previous theorem using the
notion of parallel critical pairs \cite{Gramlich:PCP,Toy81}.

\begin{defi}[parallel critical pairs \cite{Gramlich:PCP,Toy81}]
Let $s_1,\ldots,s_n,t$ ($n \ge 1$) be terms whose variables are disjoint.
The terms $s_1,\ldots,s_n$ \textit{parallel-overlap} on $t$ 
(at parallel positions $p_1,\ldots,p_n$)
if $t/p_i \notin \mathcal{V}$ for any $1 \le i \le n$
and 
$\{ s_1 \approx t/p_1,\ldots,
s_n \approx t/p_n \}$ is unifiable.
Let $l_1 \to r_1,\ldots,l_n \to r_n$ and $l' \to r'$ be rewrite rules.
W.l.o.g.\ let their variables be mutually disjoint.
Suppose that $l_1,\ldots,l_n$ parallel-overlap on $l'$ at parallel positions
$p_1,\ldots,p_n$
and $\sigma$ is the mgu of 
$\{ l_1 \approx l'/p_1,\ldots, l_n \approx l'/p_n \}$.
Then the term $l'[l_1,\ldots,l_n]_{p_1,\ldots,p_n}\sigma$ yields a 
\textit{parallel critical pair}
$\langle l'[r_1,\ldots,r_n]_{p_1,\ldots,p_n}\sigma, r'\sigma \rangle$
obtained by the parallel-overlap of $l_1 \to r_1,\ldots,l_n \to r_n$ 
on $l' \to r'$ at positions $p_1,\ldots,p_n$.
In the case of self-overlap 
(i.e.\ when $n = 1$ and $l_1 \to r_1$ and $l' \to r'$ 
are identical modulo renaming),
we do not consider the case $p_1 = \epsilon$.
We write $\langle l'[r_1,\ldots,r_n]_{p_1,\ldots,p_n}\sigma, r'\sigma \rangle_X$
if $X = \mathcal{V}_{\{ p_1,\ldots, p_n \}}(l'\sigma)$.
We call the parallel critical pair \textit{outer} if $n = 1$ and $p_1 = \epsilon$,
and \textit{inner} if $p_i > \epsilon$ for all $i$.
The set of outer (inner) parallel critical pairs obtained by the parallel-overlaps of 
rewrite rules from $\mathcal{R}$ on a rewrite rule from $\mathcal{Q}$
is denoted by 
$\mathrm{PCP}_\textit{out}(\mathcal{R}, \mathcal{Q})$
($\mathrm{PCP}_\textit{in}(\mathcal{R}, \mathcal{Q})$, respectively).
(Note, however, 
that $\mathrm{PCP}_\textit{out}(\mathcal{R}, \mathcal{Q})
= \mathrm{CP}_\textit{out}(\mathcal{R}, \mathcal{Q})$.)
We put
$\mathrm{PCP}(\mathcal{R},\mathcal{Q}) = 
\mathrm{PCP}_\textit{out}(\mathcal{R}, \mathcal{Q})
\cup
\mathrm{PCP}_\textit{in}(\mathcal{R}, \mathcal{Q})$.
\end{defi}

\begin{exa}[parallel critical pairs]
Let $\mathcal{R} = \{ 
\mathsf{f}(\mathsf{g}(x), \mathsf{g}(y)) \to 
\mathsf{h}(\mathsf{g}(x)) \}$
and $\mathcal{Q} = \{ 
\mathsf{g}(x) \to \mathsf{h}(x) \}$.
Then we have
$\mathrm{PCP}(\mathcal{R}, \mathcal{Q}) = 
\mathrm{PCP}_\mathit{out}(\mathcal{Q}, \mathcal{R}) = 
\emptyset$
and 
$\mathrm{PCP}_\mathit{in}(\mathcal{Q}, \mathcal{R}) = 
\{ 
\langle \mathsf{f}(\mathsf{h}(x), \mathsf{h}(y)),\;  
       \mathsf{h}(\mathsf{g}(x)) \rangle_{\{ x,y \}},
\langle \mathsf{f}(\mathsf{g}(x), \mathsf{h}(y)),\;  
       \mathsf{h}(\mathsf{g}(x)) \rangle_{\{ y \}},
\langle \mathsf{f}(\mathsf{h}(x), \mathsf{g}(y)),\;  
       \mathsf{h}(\mathsf{g}(x)) \rangle_{\{ x \}}
\}$.
\end{exa}

We first present a property of parallel rewrite steps
that will be used to prove a key lemma below.

\begin{lem}
\label{lem:paralle step}
Let $s,t$ terms such that $s \parto_{V} t$
and $\rho,\rho'$ substitutions such that $\rho \parto \rho'$.
If $\mathrm{dom}(\rho) \cap  \mathcal{V}_V(t) = \emptyset$
then $s\rho \parto t\rho'$.
\end{lem}

\begin{proof}
We have $s\rho \parto_{V} t\rho
\parto_{W} t\rho'$ for some $W$
such that for any $q \in W$ there exists 
$q' \in \mathrm{Pos}_{\mathrm{dom}(\rho)}(t)$
such that $q' \le q$.
Since $\mathrm{dom}(\rho) \cap  \mathcal{V}_V(t) = \emptyset$,
$p \parallel q'$ holds for any $p \in V$
and $q' \in \mathrm{Pos}_{\mathrm{dom}(\rho)}(t)$.
Thus $p \parallel q$ for any $p \in V$ and $q \in W$.
Hence $s\rho \parto t\rho'$.
\end{proof}

The following lemma is a key lemma which shows
the preservation of parallel rewrite steps
via substitutions.
For this lemma, a variable condition
on parallel critical pairs 
and parallel rewrite steps is essential.

\begin{lem}
\label{lem:PCP-I}
Let $\mathcal{S},\mathcal{P}$ be TRSs,
$\langle u,v \rangle_X 
\in \mathrm{PCP}_\mathit{in}(\mathcal{P},\mathcal{S})$ and
$\rho,\rho'$ substitutions such that $\rho \parto_\mathcal{P} \rho'$
and $\mathrm{dom}(\rho) \cap X = \emptyset$.
If $u \stackrel{*}{\to} u' \pargets_{V,\mathcal{P}} v' 
\stackrel{*}{\gets} v$ and
$\mathcal{V}_V(u') \subseteq X$
then
$u\rho' \stackrel{*}{\to} u' \rho' \pargets_{\mathcal{P}}
v'\rho \stackrel{*}{\gets} v\rho$,
where 
$u\rho' \stackrel{*}{\to} u' \rho'$
and $v'\rho \stackrel{*}{\gets} v\rho$
are the obvious instances of $u \stackrel{*}{\to} u'$
and $v' \stackrel{*}{\gets} v$, respectively.
\end{lem}

\begin{proof}
It is clear from our assumption that 
$u\rho' \stackrel{*}{\to} u' \rho' 
\pargets_{\mathcal{P}}
u'\rho
\pargets_{V,\mathcal{P}}
v'\rho 
\stackrel{*}{\gets} v\rho$.
Thus it remains to show 
$u'\rho' \pargets_{\mathcal{P}} v'\rho$.
Since $\mathcal{V}_V(u') \subseteq X$
and $\mathrm{dom}(\rho) \cap X = \emptyset$,
$\mathrm{dom}(\rho) \cap \mathcal{V}_V(u') = \emptyset$.
Thus, since $u' \pargets_{V,\mathcal{P}} v'$, 
it follows that $u'\rho' \pargets_{\mathcal{P}} v'\rho$
by Lemma \ref{lem:paralle step}.
\end{proof}

\begin{exa}
Suppose we have TRSs
$\mathcal{S} = 
\{ +(\mathsf{0},y) \to y,
+(\mathsf{s}(x),y)\to \mathsf{s}(+(x,y))
\}$
and 
$\mathcal{P} = 
\{ \mathsf{s}(\mathsf{s}(x)) \to \mathsf{s}(x) \}$.
Then we have 
$\langle +(\mathsf{s}(x),y),\;  \mathsf{s}(+(\mathsf{s}(x),y) \rangle_{\{ x \}}
\in \mathrm{PCP}_{\mathit{in}}(\mathcal{P},\mathcal{S})$.
This critical pair is joinable as
$+(\mathsf{s}(x),y)
\to_{\mathcal{S}}
\mathsf{s}(+(x,y)) = u'
\pargets_{\{ \epsilon \},\mathcal{P}}
v' = \mathsf{s}(\mathsf{s}(+(x,y)))
\gets_{\mathcal{S}}
\mathsf{s}(+(\mathsf{s}(x),y))$.
However, since 
$\mathcal{V}_V(u') = \mathcal{V}_{\{ \epsilon \}} (\mathsf{s}(+(x,y)))
= \{ x,y \} \nsubseteq X = \{ x  \}$,
the variable condition 
of Lemma \ref{lem:PCP-I} is not satisfied.
Take $\rho,\rho'$ such that 
$\mathrm{dom}(\rho) = \mathrm{dom}(\rho') =  \{ y \}$,
$\rho(y) = \mathsf{s}(\mathsf{s}(z))$ and
$\rho'(y) = \mathsf{s}(z)$.
Then we have 
$\rho \parto_{\mathcal{P}} \rho'$.
Now, we can construct a rewrite sequence
$+(\mathsf{s}(x),y)\rho'
= +(\mathsf{s}(x),\mathsf{s}(z))
\to_{\mathcal{S}}
\mathsf{s}(+(x,\mathsf{s}(z)))
\stackrel{*}{\gets}_{\mathcal{P}}
\mathsf{s}(\mathsf{s}(+(x,\mathsf{s}(\mathsf{s}(z)))))
\gets_{\mathcal{S}}
\mathsf{s}(+(\mathsf{s}(x),\mathsf{s}(\mathsf{s}(z))))
=
\mathsf{s}(+(\mathsf{s}(x),y))\rho$.
However, $\mathsf{s}(+(x,\mathsf{s}(z))
\pargets_{\mathcal{P}}
\mathsf{s}(\mathsf{s}(+(x,\mathsf{s}(\mathsf{s}(z)))))$
does not hold, as the two redex occurrences
are not parallel.
\end{exa}

\begin{rem}
The variable condition 
$\mathcal{V}_V(u') \subseteq X$
in Lemma \ref{lem:PCP-I}
is not equivalent to $\mathcal{V}_V(v') \subseteq X$
if there exists $l \to r \in \mathcal{P}$
such that $\mathcal{V}(r) \subsetneq \mathcal{V}(l)$.
In the sequel, however, 
$\mathcal{V}(r) = \mathcal{V}(l)$ for all 
$l \to r \in \mathcal{P}$ holds
whenever we apply the lemma.
\end{rem}

We extend Lemma~\ref{lem:parallel-II} as follows.

\begin{lem}
\label{lem:PCP-II}
Let $\mathcal{Q},\mathcal{R}$ be TRSs
such that $\mathcal{Q}$ is bidirectional
and $\mathcal{R}$ is left-linear.
Let $\blacktriangleright$ be rewrite relations
such that 
$s_i \blacktriangleright t_i$ and
$s_j \gets_\mathcal{Q} t_j$ for any $j \in \{ 1,\ldots,n \} \setminus \{ i \}$
implies
$C[s_1,\ldots,s_n] \blacktriangleright C[t_1,\ldots,t_n]$.
Suppose 
(i) for any $\langle u,v \rangle_X 
\in \mathrm{PCP}_\mathit{in}(\mathcal{Q},\mathcal{R})$
and substitutions $\rho,\rho'$ 
such that $\rho \parto_\mathcal{Q} \rho'$
and $\mathrm{dom}(\rho) \cap X = \emptyset$,
we have 
$u\rho' \blacktriangleright v\rho$ and 
(ii) 
$\mathrm{CP}(\mathcal{R},\mathcal{Q})
\subseteq {\blacktriangleleft}$.
If 
${\stackrel{+}{\to}}_\mathcal{R}
\circ {\pargets}_\mathcal{Q} \circ
{\stackrel{*}{\gets}}_\mathcal{R}
\subseteq {\blacktriangleright}$
then 
${\pargets}_\mathcal{Q} \circ {\to}_\mathcal{R}
\subseteq 
{\blacktriangleright}$.
\end{lem}

\begin{proof}
Suppose
$t_0 \pargets_{U,\mathcal{Q}} s \to_{q,\mathcal{R}} t_1$.
Let $U = \{ p_1,\ldots, p_n \}$ where $p_1,\ldots, p_n$ are positions
from left to right,
$s/p_i = l_i\sigma_i$ for $l_i \to r_i \in \mathcal{Q}$ 
and substitutions $\sigma_i$
($1 \le i \le n$)
and $s/q = l'\rho$ for $l' \to r' \in \mathcal{R}$ and substitution $\rho$.
The same proof as in Lemma~\ref{lem:parallel-II} applies 
other than the case of
$\forall p \in U.~p \not\le q$.
Let $\{ p_k,\ldots,p_m \} = \{ p_i \in U \mid q \le p_i \} $.
For each $p_i$ ($k \le i \le m$)
either $p_i/q \in \mathrm{Pos}_\mathcal{F}(l')$
or 
there exists $q_x \in \mathrm{Pos}_\mathcal{V}(l')$
such that $q.q_x \le p_i$.
W.l.o.g.\ let
$\{ p_k,\ldots,p_l \} = \{ p_i \mid 
p_i/q \in \mathrm{Pos}_\mathcal{F}(l') \}$
and 
$\{ p_{l+1},\ldots,p_m \} = \{ p_i \mid 
\exists q_x \in \mathrm{Pos}_\mathcal{V}(l').~q.q_x \le p_i
\}$.
Then there exists 
a parallel critical pair $\langle u,v \rangle_X$
obtained from overlaps of 
$l_{k} \to r_{k},\ldots,l_{l} \to r_{l}$
on $l' \to r'$ at 
$p_{k}/q,\ldots,p_{l}/q$.
Since $l'$ is linear
(and $\mathcal{V}(l'),\mathcal{V}(l_1),\ldots,\mathcal{V}(l_m)$
are mutually disjoint),
$t_0/q = u\rho'$
and 
$t_1/q = v\rho$
for some substitutions $\rho,\rho'$
such that $\rho \parto_\mathcal{Q} \rho'$
and $\mathrm{dom}(\rho) \cap X = \emptyset$.
Hence by our assumption
$u\rho' \blacktriangleright v\rho$.
Thus by our assumption on $\blacktriangleright$,
it follows that
$t_0 = s[r_1\sigma_1,\ldots,r_{k-1}\sigma_{k-1},u\rho',
r_{m+1}\sigma_{m+1},\ldots,r_{n}\sigma_{n}]_{p_1,\ldots,p_{k-1},q,p_{l+1},\ldots,p_{n}}
\blacktriangleright
s[l_1\sigma_1,\ldots,l_{k-1}\sigma_{k-1},v\rho,
l_{m+1}\sigma_{m+1},\ldots,l_{n}\sigma_{n}]_{p_1,\ldots,p_{k-1},q,p_{l+1},\ldots,p_{n}}
= t_1 $.
\end{proof}

The following lemma
is analogous to 
Lemmas \ref{lem:linear-III}, \ref{lem:parallel-III}.

\begin{lem}
\label{lem:PCP-III}
Let $\mathcal{P},\mathcal{S},\mathcal{P}'$ be TRSs
such that $\mathcal{P}$ is bidirectional
and $\mathcal{S}$ is left-linear.
Let ${\blacktriangleright}
=
(
{\pargets}_{\mathcal{P} \cup \mathcal{P}^{-1}}
\circ 
{\stackrel{*}{\gets}}_{\mathcal{S} \cup \mathcal{P}'})
\cup
({\to}_{\mathcal{S}}
\circ 
{\stackrel{*}{\to}}_{\mathcal{S} \cup \mathcal{P}'}
\circ 
{\pargets}_{\mathcal{P} \cup \mathcal{P}^{-1}}
\circ 
{\stackrel{*}{\gets}}_{\mathcal{S} \cup \mathcal{P}'})
$.
Suppose 
(i) for all $\langle u,v \rangle_X
\in \mathrm{PCP}_\mathit{in}(\mathcal{P} \cup \mathcal{P}^{-1},\mathcal{S})$,
either
$u 
=
u'
\pargets_{V,\mathcal{P}\cup \mathcal{P}^{-1}}
\circ
\stackrel{*}{\gets}_{\mathcal{S} \cup \mathcal{P}'}
v$
or
$u 
\to_{\mathcal{S}}
\circ
\stackrel{*}{\to}_{\mathcal{S} \cup \mathcal{P}'}
u'
\pargets_{V,\mathcal{P}\cup \mathcal{P}^{-1}}
\circ
\stackrel{*}{\gets}_{\mathcal{S} \cup \mathcal{P}'}
v$
for some $u'$ and $V$ satisfying 
$\mathcal{V}_V(u') \subseteq X$
and
(ii)
$\mathrm{CP}(\mathcal{S},\mathcal{P} \cup \mathcal{P}^{-1})
\subseteq {\blacktriangleleft}$.
Then 
${\pargets}_{\mathcal{P}\cup \mathcal{P}^{-1}}
\circ {\to}_\mathcal{S}
\subseteq 
{\blacktriangleright}$.
\end{lem}

\begin{proof}
Take $\mathcal{Q} := \mathcal{P} \cup  \mathcal{P}^{-1}$,
which is a bidirectional TRS by bidirectionality of $\mathcal{P}$,
and $\mathcal{R} := \mathcal{S}$ in Lemma~\ref{lem:PCP-II}.
Then by the condition (ii)
we have $\mathrm{CP}(\mathcal{R},\mathcal{Q})
\subseteq {\blacktriangleleft}$.
Furthermore,
by the condition (i) and Lemma~\ref{lem:PCP-I},
for any $\langle u,v \rangle_X 
\in \mathrm{PCP}_\mathit{in}(\mathcal{Q},\mathcal{R})$
and substitutions $\rho,\rho'$ 
such that  $\rho \parto_\mathcal{Q} \rho'$
and $\mathrm{dom}(\rho) \cap X = \emptyset$,
we have 
either
$u \rho'
= u' \rho'
\pargets_{\mathcal{P}\cup \mathcal{P}^{-1}}
\circ
\stackrel{*}{\gets}_{\mathcal{S} \cup \mathcal{P}'}
v\rho$
or
$u \rho'
\to_{\mathcal{S}}
\circ
\stackrel{*}{\to}_{\mathcal{S} \cup \mathcal{P}'}
u' \rho'
\pargets_{\mathcal{P}\cup \mathcal{P}^{-1}}
\circ
\stackrel{*}{\gets}_{\mathcal{S} \cup \mathcal{P}'}
v\rho$
and hence 
$u \rho' \blacktriangleright v\rho$.
We also have 
${\stackrel{+}{\to}}_\mathcal{R}
\circ {\pargets}_\mathcal{Q}
\circ {\stackrel{*}{\gets}}_\mathcal{R}
\subseteq 
{\to}_{\mathcal{S}}
\circ 
{\stackrel{*}{\to}}_{\mathcal{S} \cup \mathcal{P}'}
\circ 
{\pargets}_{\mathcal{P} \cup \mathcal{P}^{-1}}
\circ 
{\stackrel{*}{\gets}}_{\mathcal{S} \cup \mathcal{P}'}
\subseteq {\blacktriangleright}$.
By the definition of ${\blacktriangleright}$,
$s_i \blacktriangleright t_i$ and
$s_j \gets_{\mathcal{P}\cup\mathcal{P}^{-1}} t_j$ for any $j \in \{ 1,\ldots,n \} \setminus \{ i \}$
implies
$C[s_1,\ldots,s_n] \blacktriangleright C[t_1,\ldots,t_n]$.
Hence the claim follows from Lemma~\ref{lem:PCP-II}.
\end{proof}

The next theorem strengthens Theorem \ref{thm:parallel}.

\begin{thm}
\label{thm:PCP}
Let $\mathcal{P},\mathcal{S},\mathcal{P}'$ be TRSs
such that $\mathcal{S}$ is left-linear,
$\mathcal{P}$ is reversible,
$\mathcal{P}' \subseteq \mathcal{P} \cup \mathcal{P}^{-1}$
and 
$\mathcal{S}$ is terminating relative to $\mathcal{P}'$.
Suppose
(i) $\mathrm{CP}(\mathcal{S},\mathcal{S})
\subseteq 
{\stackrel{*}{\to}}_{\mathcal{S}\cup \mathcal{P}'} 
\circ 
{\pargets}_{\mathcal{P} \cup \mathcal{P}^{-1}}
\circ
{\stackrel{*}{\gets}}_{\mathcal{S}\cup\mathcal{P'}}$,
(ii) for all $\langle u,v \rangle_X
\in \mathrm{PCP}_\mathit{in}(\mathcal{P} \cup \mathcal{P}^{-1},\mathcal{S})$,
either
$u 
=
u'
\pargets_{V,\mathcal{P}\cup \mathcal{P}^{-1}}
\circ
\stackrel{*}{\gets}_{\mathcal{S} \cup \mathcal{P}'}
v$
or
$u 
\to_{\mathcal{S}}
\circ
\stackrel{*}{\to}_{\mathcal{S} \cup \mathcal{P}'}
u'
\pargets_{V,\mathcal{P}\cup \mathcal{P}^{-1}}
\circ
\stackrel{*}{\gets}_{\mathcal{S} \cup \mathcal{P}'}
v$
for some $u'$ and $V$ satisfying 
$\mathcal{V}_V(u') \subseteq X$
and
(iii) $\mathrm{CP}(\mathcal{S},\mathcal{P} \cup \mathcal{P}^{-1})
\subseteq 
({\stackrel{*}{\to}}_{\mathcal{S} \cup \mathcal{P}'}
\circ 
{\parto}_{\mathcal{P} \cup \mathcal{P}^{-1}}
)
\cup
({\stackrel{*}{\to}}_{\mathcal{S} \cup \mathcal{P}'}
\circ 
{\parto}_{\mathcal{P} \cup \mathcal{P}^{-1}}
\circ 
{\stackrel{*}{\gets}}_{\mathcal{S} \cup \mathcal{P}'}
\circ 
{\gets}_{\mathcal{S}})
$.
Then 
$\mathcal{S}\cup \mathcal{P}$ is confluent.
\end{thm}

\begin{proof}
By our assumption (i) and Lemma~\ref{lem:PCP-I},
we have (a) ${\gets}_\mathcal{S} \circ {\to}_\mathcal{S}
\subseteq 
{\stackrel{*}{\to}}_{\mathcal{S}\cup \mathcal{P}'} 
\circ 
{\pargets}_{\mathcal{P} \cup \mathcal{P}^{-1}}
\circ
{\stackrel{*}{\gets}}_{\mathcal{S}\cup\mathcal{P'}}$.
From our assumptions (ii) and (iii), 
it follows that
(b) 
${\pargets}_{\mathcal{P}\cup \mathcal{P}^{-1}} \circ {\to}_\mathcal{S}
\subseteq 
(
{\pargets}_{\mathcal{P} \cup \mathcal{P}^{-1}}
\circ 
{\stackrel{*}{\gets}}_{\mathcal{S} \cup \mathcal{P}'}
)
\cup
({\to}_{\mathcal{S}}
\circ 
{\stackrel{*}{\to}}_{\mathcal{S} \cup \mathcal{P}'}
\circ 
{\pargets}_{\mathcal{P} \cup \mathcal{P}^{-1}}
\circ 
{\stackrel{*}{\gets}}_{\mathcal{S} \cup \mathcal{P}'}
)
$ by Lemma~\ref{lem:PCP-III}.
Take ${\vdashv} := {\pargets}_{\mathcal{P}\cup\mathcal{P}^{-1}}$,
${\to} := {\to}_\mathcal{S}$
and 
${\leadsto} := {\to}_{\mathcal{P}'}$.
Then, by the termination of $\mathcal{S}$ relative to 
$\mathcal{P}'$, 
${\to}\circ \stackrel{*}{\leadsto}$ is well-founded.
Thus one can apply Theorem \ref{thm:ARS}
so as to prove ${\to}_\mathcal{S}$ is Church-Rosser modulo
$\stackrel{*}{\pargets}_{\mathcal{P}\cup\mathcal{P}^{-1}}$.
Since 
${\stackrel{*}{\pargets}}_{\mathcal{P}\cup\mathcal{P}^{-1}}
= {\stackrel{*}{\leftrightarrow}}_\mathcal{P}$,
it follows that ${\to}_\mathcal{S}$ is Church-Rosser modulo
$\stackrel{*}{\leftrightarrow}_\mathcal{P}$.
Hence, since $\to_\mathcal{P}$ is reversible,
${\to}_{\mathcal{S} \cup \mathcal{P}}$ is confluent
by Lemma \ref{lem:reversible + CR modulo}.
\end{proof}

By the definition of parallel critical pairs,
$\mathrm{CP}_\mathit{in}(\mathcal{P} \cup \mathcal{P}^{-1},\mathcal{S})
= \emptyset$
implies 
$\mathrm{PCP}_\mathit{in}(\mathcal{P} \cup \mathcal{P}^{-1},\mathcal{S})
= \emptyset$.
Thus the condition (ii) of Theorem~\ref{thm:parallel} is
a particular case of condition (ii) of Theorem~\ref{thm:PCP}.
Hence 
Theorem~\ref{thm:parallel} is subsumed by Theorem~\ref{thm:PCP}.

By taking $\mathcal{P}' = \emptyset$ in 
Theorem~\ref{thm:PCP},
we obtain the next corollary.

\begin{cor}
\label{cor:PCP}
Let $\mathcal{P},\mathcal{S}$ be TRSs
such that $\mathcal{S}$ is left-linear,
$\mathcal{P}$ is reversible
and 
$\mathcal{S}$ is terminating.
Suppose
(i) $\mathrm{CP}(\mathcal{S},\mathcal{S})
\subseteq 
{\stackrel{*}{\to}}_{\mathcal{S}}
\circ 
{\pargets}_{\mathcal{P} \cup \mathcal{P}^{-1}}
\circ
{\stackrel{*}{\gets}}_{\mathcal{S}}$,
(ii) for all $\langle u,v \rangle_X
\in \mathrm{PCP}_\mathit{in}(\mathcal{P} \cup \mathcal{P}^{-1},\mathcal{S})$,
$u 
\stackrel{*}{\to}_{\mathcal{S}}
u'
\pargets_{V,\mathcal{P}\cup \mathcal{P}^{-1}}
\circ 
\stackrel{*}{\gets}_{\mathcal{S}}
v$
for some $u'$ and $V$ satisfying 
$\mathcal{V}_V(u') \subseteq X$
and
(iii) $\mathrm{CP}(\mathcal{S},\mathcal{P} \cup \mathcal{P}^{-1})
\subseteq 
{\stackrel{*}{\to}}_{\mathcal{S}}
\circ 
{\parto}_{\mathcal{P} \cup \mathcal{P}^{-1}}
\circ 
{\stackrel{*}{\gets}}_{\mathcal{S}}
$.
Then 
$\mathcal{S}\cup \mathcal{P}$ is confluent.
\end{cor}

Considering a particular case of the condition (ii),
we obtain the following corollary.

\begin{cor}[Theorem 2 of \cite{JKR83}]
\label{cor:PCP-JKR83}
Let $\mathcal{P},\mathcal{S}$ be TRSs
such that $\mathcal{S}$ is left-linear,
$\mathcal{P}$ is reversible
and 
$\mathcal{S}$ is terminating.
Suppose
(i) $\mathrm{CP}(\mathcal{S},\mathcal{S})
\subseteq 
{\stackrel{*}{\to}}_{\mathcal{S}}
\circ 
{\pargets}_{\mathcal{P} \cup \mathcal{P}^{-1}}
\circ
{\stackrel{*}{\gets}}_{\mathcal{S}}$,
(ii) 
$\mathrm{PCP}_\mathit{in}(\mathcal{P} \cup \mathcal{P}^{-1},\mathcal{S})
\subseteq
{\stackrel{*}{\to}}_{\mathcal{S}}
\circ 
{\stackrel{*}{\gets}}_{\mathcal{S}}$
and
(iii) $\mathrm{CP}(\mathcal{S},\mathcal{P} \cup \mathcal{P}^{-1})
\subseteq 
{\stackrel{*}{\to}}_{\mathcal{S}}
\circ 
{\parto}_{\mathcal{P} \cup \mathcal{P}^{-1}}
\circ 
{\stackrel{*}{\gets}}_{\mathcal{S}}
$.
Then 
$\mathcal{S}\cup \mathcal{P}$ is confluent.
\end{cor}

\subsection{Examples and comparison}

In this subsection,
some examples to illustrate applicability
of our confluence criteria in previous subsections
are presented.
Relations among our confluence criteria
and locations of given examples are summarized in 
Figure \ref{fig:relation}.

\begin{exa}
\label{exp:add2pp-assoc-com}
Let 
$\mathcal{R}_3 = \{ 
(\mathsf{add}_1),
(\mathsf{add}_2),
(C),
(A) \} \cup {}$
\[
\left\{
\begin{array}{clclclcl}
(\mathsf{add}_3) & \mathsf{+}(x,\mathsf{0})    &\to&  x &\quad\\
(\mathsf{add}_4) & \mathsf{+}(x,\mathsf{s}(y)) &\to&  \mathsf{s}(\mathsf{+}(x,y)) \\
\end{array}
\right\}
\]
where 
$(\mathsf{add}_1)$,
$(\mathsf{add}_2)$,
$(C)$,
$(A)$ are rewrite rules given in Example~\ref{exp:plus-com-assoc}.
We now prove confluence of $\mathcal{R}_3$
using Theorem~\ref{thm:parallel}.
For this, 
put $\mathcal{S} = \{ 
(\mathsf{add}_1),(\mathsf{add}_2),(\mathsf{add}_3),(\mathsf{add}_4)
\}$
and $\mathcal{P} = \{ (C),(A) \}$.
Then $\mathcal{S}$ is linear and terminating.
As demonstrated in Example~\ref{exp:plus-com-assoc},
$\mathcal{P}$ is reversible.
We have
$\mathrm{CP}(\mathcal{S},\mathcal{S}) = $
\[
\left\{ 
\begin{array}{l@{\,}ll@{\,}l}
\langle \mathsf{0}, \mathsf{0} \rangle 
& \in  {\stackrel{*}{\gets}}_\mathcal{S} &
\langle \mathsf{s}(y), \mathsf{s}(\mathsf{+}(\mathsf{0},y)) \rangle & \in  {\gets}_\mathcal{S}\\
\langle \mathsf{s}(\mathsf{+}(x,\mathsf{0})), \mathsf{s}(x)  \rangle 
& \in  {\to}_\mathcal{S} &
\langle \mathsf{s}(x), \mathsf{s}(\mathsf{+}(x,\mathsf{0})) \rangle & \in  {\gets}_\mathcal{S} \\
\langle \mathsf{s}(\mathsf{+}(\mathsf{0},y)), \mathsf{s}(y)  \rangle 
& \in  {\to}_\mathcal{S} &
\langle \mathsf{s}(\mathsf{+}(x,\mathsf{s}(y))), \mathsf{s}(\mathsf{+}(\mathsf{s}(x),y)) \rangle
& \in  {\to}_\mathcal{S} \circ {\gets}_\mathcal{S}\\
\langle \mathsf{s}(\mathsf{+}(\mathsf{s}(x),y)), \mathsf{s}(\mathsf{+}(x,\mathsf{s}(y))) \rangle
& \in  {\to}_\mathcal{S} \circ {\gets}_\mathcal{S}\\
\end{array}
\right\},
\]
$\mathrm{CP}_\textit{in}(\mathcal{P}\cup \mathcal{P}^{-1},\mathcal{S})
= \emptyset$ and
$\mathrm{CP}(\mathcal{S},\mathcal{P}\cup \mathcal{P}^{-1}) =$
\[
\left\{ 
\begin{array}{l@{\,}ll@{\,}l}
\langle y, \mathsf{+}(y,\mathsf{0}) \rangle
& \in  {\gets}_\mathcal{S} &
\langle \mathsf{+}(y,z), \mathsf{+}(\mathsf{0},\mathsf{+}(y,z)) \rangle
& \in  {\gets}_\mathcal{S}\\
\langle \mathsf{+}(y,z), \mathsf{+}(\mathsf{+}(\mathsf{0},y),z) \rangle
& \in  {\gets}_\mathcal{S} &
\langle \mathsf{+}(x,z), \mathsf{+}(\mathsf{+}(x,\mathsf{0}),z) \rangle
& \in  {\gets}_\mathcal{S}\\
\langle \mathsf{s}(\mathsf{+}(x,y)), \mathsf{+}(y,\mathsf{s}(x)) \rangle
& \in  {\leftrightarrow}_\mathcal{P} \circ {\gets}_\mathcal{S}\\
\langle \mathsf{+}(\mathsf{s}(\mathsf{+}(x,y)),z), 
        \mathsf{+}(\mathsf{s}(x),\mathsf{+}(y,z)) \rangle
& \in  {\to}_\mathcal{S} \circ {\leftrightarrow}_\mathcal{P} \circ {\gets}_\mathcal{S}\\
\langle \mathsf{s}(\mathsf{+}(x,\mathsf{+}(y,z))), 
        \mathsf{+}(\mathsf{+}(\mathsf{s}(x),y),z) \rangle
& \in  {\leftrightarrow}_\mathcal{P} \circ {\stackrel{*}{\gets}}_\mathcal{S}\\
\langle \mathsf{+}(x,\mathsf{s}(\mathsf{+}(y,z))), 
        \mathsf{+}(\mathsf{+}(x,\mathsf{s}(y)),z) \rangle
& \in  {\to}_\mathcal{S} \circ {\leftrightarrow}_\mathcal{P} \circ {\stackrel{*}{\gets}}_\mathcal{S}\\
\langle x, \mathsf{+}(\mathsf{0},x) \rangle
& \in  {\gets}_\mathcal{S} &
\langle \mathsf{+}(x,y), \mathsf{+}(x,\mathsf{+}(y,\mathsf{0})) \rangle
& \in  {\gets}_\mathcal{S}\\
\langle \mathsf{+}(y,z),
        \mathsf{+}(y,\mathsf{+}(\mathsf{0},z)) \rangle
& \in  {\gets}_\mathcal{S} &
\langle \mathsf{+}(x,y), \mathsf{+}(\mathsf{+}(x,y),\mathsf{0}) \rangle
& \in  {\gets}_\mathcal{S}\\
\langle \mathsf{s}(\mathsf{+}(x,y)), \mathsf{+}(\mathsf{s}(y),x) \rangle
& \in  {\leftrightarrow}_\mathcal{P} \circ {\gets}_\mathcal{S}\\
\langle \mathsf{s}(\mathsf{+}(\mathsf{+}(x,y),z)), \mathsf{+}(x,\mathsf{+}(y,\mathsf{s}(z))) \rangle
& \in  {\leftrightarrow}_\mathcal{P} \circ {\stackrel{*}{\gets}}_\mathcal{S}\\
\langle \mathsf{+}(\mathsf{s}(\mathsf{+}(x,y)),z), 
        \mathsf{+}(x,\mathsf{+}(\mathsf{s}(y),z)) \rangle
& \in  {\to}_\mathcal{S} \circ {\leftrightarrow}_\mathcal{P} \circ {\stackrel{*}{\gets}}_\mathcal{S}\\
\langle \mathsf{+}(x,\mathsf{s}(\mathsf{+}(y,z))), 
        \mathsf{+}(\mathsf{+}(x,y),\mathsf{s}(z)) \rangle
& \in  {\to}_\mathcal{S} \circ {\leftrightarrow}_\mathcal{P} \circ {\gets}_\mathcal{S}\\
\end{array}
\right\}.
\]
Thus one can apply Theorem~\ref{thm:parallel} 
to obtain the confluence of 
$\mathcal{R}_3 = \mathcal{S}  \cup \mathcal{P}$.
To this example, one can apply 
Theorem~\ref{thm:linear},
Corollaries~\ref{cor:linear}, \ref{cor:parallel}
and Proposition~\ref{prop:relative} as well.
\end{exa}

\begin{exa}
\label{exp:add-assoc-com}
Let us consider
$\mathcal{R}_2$ of Example~\ref{exp:plus-com-assoc}.
For this, we consider
a partition $\mathcal{S} = \{ (\mathsf{add}_1), (\mathsf{add}_2) \}$ and
$\mathcal{P} = \{ (C), (A) \}$.
Then $\mathcal{S}$ is linear and terminating
and $\mathcal{P}$ is reversible.
But there exists a critical pair
$\langle y, \mathsf{+}(y,\mathsf{0}) \rangle
\in \mathrm{CP}(\mathcal{S},\mathcal{P} \cup \mathcal{P}^{-1})$
obtained by the overlap of $(\mathsf{add}_1)$ on $(C)$
which is not included
in 
$({\stackrel{*}{\to}}_{\mathcal{S} \cup \mathcal{P}'}
\circ 
{\stackrel{=}{\to}_{\mathcal{P}  \cup\mathcal{P}^{-1}}})
\cup
({\stackrel{*}{\to}}_{\mathcal{S} \cup \mathcal{P}'}
\circ 
{\stackrel{=}{\to}_{\mathcal{P}  \cup\mathcal{P}^{-1}}}
\circ 
{\stackrel{*}{\gets}}_{\mathcal{S} \cup \mathcal{P}'}
\circ 
{\stackrel{*}{\gets}}_\mathcal{S})$ for any 
$\mathcal{P}' \subseteq \mathcal{P}  \cup \mathcal{P}^{-1}$.
Thus the critical pair conditions of 
Theorem~\ref{thm:linear} are not satisfied.
Since $\langle y, \mathsf{+}(y,\mathsf{0}) \rangle$
is obtained by the overlap of $(\mathsf{add}_1)$ on $(C)$,
another possible
partition $\mathcal{S} = \{ (\mathsf{add}_1), (\mathsf{add}_2), (A) \}$ and
$\mathcal{P} = \{ (C) \}$ is not effective either.
Similarly, Theorem \ref{thm:PCP} and 
Proposition \ref{prop:relative} are also not applicable.
We will revisit this example in the next section.
\end{exa}

\begin{figure}[t]
\begin{center}
\PSforPDF{
\begin{pspicture}(-3,-1)(9,4.3) 
\psframe[framearc=.25,linestyle=dashed](-.5,-.8)(6.5,1.9) 
\psellipse[](1,2)(2.8,1.2)
\psellipse[](1,2)(4,2)
\psellipse[](5,2)(2.8,1.2)
\psellipse[](5,2)(4,2)
\rput(3,1.6){$\bullet$ $\mathcal{R}_3$}
\rput(0.5,1.3){$\bullet$ $\mathcal{R}_4$}
\rput(0,2.3){$\bullet$ $\mathcal{R}_6$}
\rput(6.5,2.3){$\bullet$ $\mathcal{R}_8$}
\rput(-1,3.3){$\bullet$ $\mathcal{R}_7$}
\rput(4.2,2.5){$\bullet$ $\mathcal{R}_5$}
\rput(8,0){$\bullet$ $\mathcal{R}_2$}
\rput(0.3,2.8){Cor.~\ref{cor:PCP}}
\rput(.7,3.6){Thm.~\ref{thm:PCP}}
\rput(5.8,2.8){Cor.~\ref{cor:linear}}
\rput(5.2,3.6){Thm.~\ref{thm:linear}}
\rput(4,-.5){Prop.~\ref{prop:relative}}
\end{pspicture}}
\end{center}
\caption{Relations among confluence criteria}
\label{fig:relation}
\end{figure}

\begin{exa}
\label{exp:add2pp-assoc-com-dbl}
Let 
\[
\mathcal{R}_4 = \mathcal{R}_3 \cup 
\left\{
\begin{array}{clcl}
 (\mathsf{dbl})  & \mathsf{dbl}(x) &\to& \mathsf{+}(x,x)
\end{array}
\right\}.
\]
One can easily confirm that the
confluence of $\mathcal{R}_1$ is shown
in the same way as $\mathcal{R}_3$ 
using Theorem~\ref{thm:parallel}
by putting
$\mathcal{S} = \{ 
(\mathsf{add}_1),(\mathsf{add}_2),(\mathsf{add}_3),(\mathsf{add}_4),
(\mathsf{dbl}) \}$
and 
$\mathcal{P} = \{ (C), (A) \}$.
On the other hand, since $\mathcal{S}$ is not linear, 
Theorem \ref{thm:linear} does not apply.
To this example, one can apply 
Corollary~\ref{cor:parallel}
and Proposition~\ref{prop:relative} as well.
\end{exa}

\begin{exa}
\label{exp:add2pp-assoc-com-ss}
Consider a TRS
\[
\mathcal{R}_5 = \mathcal{R}_3 \cup 
\left\{
\begin{array}{clcl}
(\mathsf{ss}_1) & \mathsf{s}(x) &\to& \mathsf{s}(\mathsf{s}(x))\\
(\mathsf{ss}_2) & \mathsf{s}(\mathsf{s}(x)) &\to& \mathsf{s}(x) 
\end{array}
\right\}.
\]
By putting
$\mathcal{S}= \{ 
(\mathsf{add}_1),(\mathsf{add}_2),(\mathsf{add}_3),(\mathsf{add}_4)
\}$
and 
$\mathcal{P} = \{ (C), (A), (\mathsf{ss}_1), (\mathsf{ss}_2)  \}$,
one can show the confluence of $\mathcal{R}_5$
using Corollary~\ref{cor:linear}.
On the other hand,
the condition of Corollary \ref{cor:PCP}
is not satisfied because
$\mathrm{CP}_\mathit{in}(\mathcal{P} \cup \mathcal{P}^{-1},\mathcal{S})
\neq \emptyset$.
Proposition~\ref{prop:relative} does not apply either,
since $\mathcal{S}$ is not terminating relative to $\mathcal{P}$
as e.g.\
$\mathsf{+}(\mathsf{s}(\mathsf{s}(x)),y)
\to_\mathcal{S} \mathsf{s}(\mathsf{+}(\mathsf{s}(x),y))
\to_\mathcal{P} \mathsf{s}(\mathsf{+}(\mathsf{s}(\mathsf{s}(x)),y))
\to_\mathcal{S} \cdots$.
Take $\mathcal{P}' = \{ (\mathsf{ss}_2)  \}$.
Then 
the conditions of Theorem \ref{thm:PCP} are satisfied---for example,
for 
$\langle \mathsf{+}(\mathsf{s}(\mathsf{s}(x)),y), \mathsf{s}(\mathsf{+}(x,y)) \rangle_{\{ x \}}
\in \mathrm{PCP}_\mathit{in}(\mathcal{P} \cup \mathcal{P}^{-1},\mathcal{S})$,
we have 
a rewrite sequence
$\mathsf{+}(\mathsf{s}(\mathsf{s}(x)),y)
\stackrel{*}{\to}_{\mathcal{S} \cup \mathcal{P}'}
\mathsf{s}(\mathsf{+}(x,y))
\pargets_{\emptyset,\mathcal{P} \cup\mathcal{P}^{-1}}
\stackrel{*}{\gets}_{\mathcal{S}}
\mathsf{s}(\mathsf{+}(x,y))$.
Hence confluence of $\mathcal{R}_5$
is shown by Theorem \ref{thm:PCP}.
\end{exa}

\begin{exa}
\label{exp:add2pe-dbl-assoc-com}
Let 
\[
\mathcal{R}_6 = \{ 
(\mathsf{add}_1),
(\mathsf{add}_2),
(\mathsf{add}_3),
(C),
(A),
(\mathsf{dbl})
 \} \cup
\left\{
\begin{array}{clcl}
(\mathsf{add}_5) &  \mathsf{+}(x,\mathsf{s}(y)) &\to&
\mathsf{+}(\mathsf{s}(x),y)
\end{array}
\right\}.
\]
Let 
$\mathcal{S} = \{ 
(\mathsf{add}_1),
(\mathsf{add}_2),
(\mathsf{add}_3),
(\mathsf{add}_5),
(\mathsf{dbl}) \}$
and 
$\mathcal{P} = \{ (C), (A) \}$.
Then $\mathcal{S}$ is not 
terminating relative to
$\mathcal{P}$,
because of 
$\mathsf{+}(x,\mathsf{s}(x))
\to_{\mathcal{S}}
\mathsf{+}(\mathsf{s}(x),x)
\to_{\mathcal{P}}
\mathsf{+}(x,\mathsf{s}(x))
\to_{\mathcal{S}}
\cdots$.
Thus Proposition $\ref{prop:relative}$ 
is not applicable.
Since $\mathcal{S}$ is not linear,
Theorem \ref{thm:linear} does not apply either.
On the other hand, 
one easily checks the confluence of $\mathcal{R}_6$
using Theorems \ref{thm:parallel} or \ref{thm:PCP}.
To this example,
Corollaries \ref{cor:parallel} and \ref{cor:PCP} also apply.
\end{exa}

\begin{exa}
\label{exp:add2pp-dbl-assoc-com-ss}
Let 
\[
\mathcal{R}_7 = \mathcal{R}_4 \cup \mathcal{R}_5 = \{ 
(\mathsf{add}_1),
(\mathsf{add}_2),
(\mathsf{add}_3),
(\mathsf{add}_4),
(C),
(A),
(\mathsf{dbl}),
(\mathsf{ss}_1),
(\mathsf{ss}_2)
 \}.
\]
By the same reason as for $\mathcal{R}_6$,
Proposition $\ref{prop:relative}$ and Theorem \ref{thm:linear} do not apply.
By the same reason as for $\mathcal{R}_5$,
Corollary \ref{cor:PCP} does not apply.
Confluence of $\mathcal{R}_7$ can be 
shown as in Example~\ref{exp:add2pp-assoc-com-ss}
by taking 
$\mathcal{S} = 
(\mathsf{add}_1),
(\mathsf{add}_2),
(\mathsf{add}_3),
(\mathsf{add}_4),
(\mathsf{dbl}) \}$,
$\mathcal{P} = 
\{ (C),
(A),
(\mathsf{ss}_1),
(\mathsf{ss}_2) \}$ and
$\mathcal{P}' = \{ (\mathsf{ss}_2)  \}$.
\end{exa}

\begin{exa}
\label{exp:loop2}
Let 
\[
\mathcal{R}_8 =
\left\{
\begin{array}{clcl@{\quad}clcl}
(a) & \mathsf{f}(\mathsf{g}(x),\mathsf{g}(y)) &\to& \mathsf{f}(\mathsf{g}(x),\mathsf{h}(y))&
(b) & \mathsf{f}(\mathsf{h}(x),\mathsf{g}(y)) &\to& \mathsf{f}(\mathsf{g}(x),\mathsf{g}(y))\\
(c) & \mathsf{f}(\mathsf{g}(x),\mathsf{h}(y)) &\to& \mathsf{f}(x,y)&
(d) & \mathsf{f}(\mathsf{h}(x),\mathsf{h}(y)) &\to& \mathsf{f}(y,x)\\
(e) & \mathsf{f}(x,y) &\to& \mathsf{f}(y,x)&
(f) & \mathsf{g}(x) &\to& \mathsf{h}(x)\\
(g) & \mathsf{h}(x) &\to& \mathsf{g}(x)\\
\end{array}
\right\}.
\]
Let $\mathcal{S} = \{ (a),(b),(c),(d) \}$
and 
$\mathcal{P} = \{ (e),(f),(g) \}$.
We have
$\mathrm{CP}(\mathcal{S},\mathcal{S}) = \emptyset$,
$\mathrm{CP}_\textit{in}(\mathcal{P}\cup \mathcal{P}^{-1},\mathcal{S})
= $ 
\[
\left\{ 
\begin{array}{l@{\,}l@{\,}l@{\,}l}
\langle \mathsf{f}(\mathsf{h}(x),\mathsf{g}(y)),   \mathsf{f}(\mathsf{g}(y),\mathsf{h}(x)) \rangle 
& \in  \leftrightarrow_\mathcal{P} &
\langle \mathsf{f}(\mathsf{g}(x),\mathsf{h}(y)),   \mathsf{f}(\mathsf{g}(x),\mathsf{h}(y)) \rangle 
& \in  \mathit{id}\\
\langle \mathsf{f}(\mathsf{g}(x),\mathsf{g}(y)),   \mathsf{f}(\mathsf{g}(x),\mathsf{g}(y)) \rangle 
& \in  \mathit{id} \\
\langle \mathsf{f}(\mathsf{h}(x),\mathsf{h}(y)),   \mathsf{f}(\mathsf{g}(x),\mathsf{g}(y)) \rangle & \in  \to_\mathcal{S} \circ \leftrightarrow_\mathcal{P} \circ \stackrel{*}{\gets}_\mathcal{S} \\
\langle \mathsf{f}(\mathsf{h}(x),\mathsf{h}(y)),   \mathsf{f}(x,y) \rangle       
& \in  \to_\mathcal{S} \circ \leftrightarrow_\mathcal{P} &
\langle \mathsf{f}(\mathsf{g}(x),\mathsf{g}(y)),   \mathsf{f}(x,y) \rangle       & \in  \stackrel{*}{\to}_\mathcal{S} \\
\langle \mathsf{f}(\mathsf{g}(x),\mathsf{h}(y)),   \mathsf{f}(y,x) \rangle       
& \in \to_\mathcal{S} \circ \leftrightarrow_\mathcal{P} &
\langle \mathsf{f}(\mathsf{h}(x),\mathsf{g}(y)),   \mathsf{f}(y,x) \rangle       & \in  \stackrel{*}{\to}_\mathcal{S} \circ \leftrightarrow_\mathcal{P} \\
\end{array}
\right\}
\]
and
$\mathrm{CP}(\mathcal{S},\mathcal{P}\cup \mathcal{P}^{-1}) =$
\[
\left\{ 
\begin{array}{l@{\,}l@{\quad}l@{\,}l}
\langle \mathsf{f}(\mathsf{g}(x),\mathsf{g}(y)),   \mathsf{f}(\mathsf{g}(y),\mathsf{h}(x)) \rangle & \in  \leftrightarrow_\mathcal{P} &
\langle \mathsf{f}(\mathsf{g}(x),\mathsf{h}(y)),   \mathsf{f}(\mathsf{g}(y),\mathsf{g}(x)) \rangle & \in  \leftrightarrow_\mathcal{P}  \\
\langle \mathsf{f}(x,y),   \mathsf{f}(\mathsf{h}(y),\mathsf{g}(x)) \rangle 
& \in \stackrel{*}{\gets}_\mathcal{S} &
\langle \mathsf{f}(x,y),   \mathsf{f}(\mathsf{h}(x),\mathsf{h}(y)) \rangle 
& \in  \leftrightarrow_\mathcal{P} \circ \gets_\mathcal{S} \\
\end{array}
\right\}.
\]
Thus $\mathcal{R}_8$ is confluent by Corollary \ref{cor:linear}.
Proposition~\ref{prop:relative} does not apply,
since $\mathcal{S}$ is not terminating relative to $\mathcal{P}$.
Furthermore, the conditions of Theorem \ref{thm:PCP} are not satisfied.
For 
$\langle \mathsf{f}(\mathsf{h}(x),\mathsf{h}(y)), 
\mathsf{f}(x,y) \rangle_{\{ x \}}
\in \mathrm{PCP}_\mathit{in}(\mathcal{P} \cup \mathcal{P}^{-1},\mathcal{S})$,
the critical pair conditions can not be satisfied.
For, any rewrite sequence
$\mathsf{f}(\mathsf{h}(x),\mathsf{h}(y))
\stackrel{*}{\to}_{\mathcal{S} \cup \mathcal{P}'}
\circ
\pargets_{V, \mathcal{P} \cup \mathcal{P}^{-1}}
\circ
\stackrel{*}{\gets}_{\mathcal{S} \cup \mathcal{P}'}
\mathsf{f}(x,y)$
satisfying the critical pair conditions,
we have $\mathsf{f}(x,y) \to \mathsf{f}(y,x) \in \mathcal{P}'$.
(Note that if we take $\mathcal{P}' := \emptyset$
then the rewrite sequence
$\mathsf{f}(\mathsf{h}(x),\mathsf{h}(y))
\stackrel{*}{\to}_{\mathcal{S}}
\mathsf{f}(y,x)
\pargets_{\{ \epsilon \}, \mathcal{P} \cup \mathcal{P}^{-1}}
\mathsf{f}(x,y)$
does not satisfy the critical pair conditions
because of the variable condition
as $\mathcal{V}_{\{ \epsilon \}}(\mathsf{f}(y,x)) = 
\{ x,y \} \not\subseteq \{ x \}$.)
But then $\mathcal{S}$ is not terminating relative to $\mathcal{P}'$.
\end{exa}

\begin{thm}\hfill
\begin{enumerate}[\em(1)]
\item
Corollary \ref{cor:linear}
and Theorem \ref{thm:PCP} are incomparable.
\item
Corollary \ref{cor:PCP}
and Theorem \ref{thm:linear} are incomparable.
\item
Theorems \ref{thm:linear}
and \ref{thm:PCP} are incomparable.
\item
Corollaries \ref{cor:linear}
and \ref{cor:PCP} are incomparable.
\end{enumerate}
\end{thm}

\begin{proof}
Each claim is witnessed by the preceding examples. See Figure \ref{fig:relation}.
\end{proof}

\section{Reduction-preserving completion}

In this section,
we give a reduction-preserving completion procedure 
so that the applicability of
the criteria presented in the previous section is enhanced. 
We first present an abstract procedure formulated
in terms of inference rules (subsection 1)
and then give a concrete one 
which forms a basis of our implementation (subsection 2).

\subsection{Abstract completion procedure}

As witnessed in Example \ref{exp:add-assoc-com},
there are cases where our confluence criteria 
are not applicable directly.
Our idea is to construct a TRS suitable for applying our criteria
by exchanging or adding rewrite rules 
without changing the reduction relation
so that the confluence of the transformed TRS implies
that of the original TRS.
Using the reversibility of $\mathcal{P}$
allows several flexibilities on such transformations.

The notion of reduction equivalence
and the following properties of reduction equivalence
are well-known in literature and the latter are easily proved.

\begin{defi}[reduction equivalence]
Two relations $\to_0$ and $\to_1$ are
said to be \textit{reduction equivalent}
if ${\stackrel{*}{\to}}_0  = {\stackrel{*}{\to}}_1$.
Two TRSs $\mathcal{R}$ and $\mathcal{Q}$
are reduction equivalent if
so are $\to_\mathcal{R}$ and $\to_{\mathcal{Q}}$.
\end{defi}

\begin{prop}[properties of reduction equivalence]
\label{prop:reduction equivalence}
(i) If ${\to}_\mathcal{R} \subseteq {\stackrel{*}{\to}}_{\mathcal{Q}}$
and ${\to}_{\mathcal{Q}} \subseteq {\stackrel{*}{\to}}_{\mathcal{R}}$
then $\mathcal{R}$ and $\mathcal{Q}$ are reduction equivalent.
(ii) If $\mathcal{R}$ and $\mathcal{Q}$ are reduction equivalent
then the confluence of $\mathcal{R}$ and $\mathcal{Q}$ 
coincide.
\end{prop}

We now demonstrate how 
the confluence criteria in the previous section
can be applied indirectly
using the notion of reduction equivalence.

\begin{exa}[confluence by reduction equivalence]
\label{exp:confluence by reduction equivalence}
We show the confluence of $\mathcal{R}_2 
= \{ (\mathsf{add}_1), (\mathsf{add}_2), (C), (A) \}$
of Example~\ref{exp:plus-com-assoc}.
Theorems \ref{thm:PCP} and \ref{thm:linear} can not be 
applied directly to prove this
as illustrated in 
Example \ref{exp:add-assoc-com}.
Consider $\mathcal{R}_3 = \mathcal{R}_2 \cup \{ (\mathsf{add}_3), (\mathsf{add}_4) \}$
of Example~\ref{exp:add2pp-assoc-com}.
Then since
we have $\mathsf{+}(x,\mathsf{0}) \to_{\mathcal{R}_2}
\mathsf{+}(\mathsf{0},x) \to_{\mathcal{R}_2} x$ and
$\mathsf{+}(x,\mathsf{s}(y)) \to_{\mathcal{R}_2}
\mathsf{+}(\mathsf{s}(y),x) \to_{\mathcal{R}_2}
\mathsf{s}(\mathsf{+}(y,x)) \to_{\mathcal{R}_2}
\mathsf{s}(\mathsf{+}(x,y))$,
the inclusions ${\to}_{\mathcal{R}_2}
\subseteq {\to}_{\mathcal{R}_3}
\subseteq {\stackrel{*}{\to}}_{\mathcal{R}_2}$ hold.
Hence $\mathcal{R}_3$
and $\mathcal{R}_2$ are reduction equivalent
by Proposition~\ref{prop:reduction equivalence} (i).
As we have shown in Example~\ref{exp:add2pp-assoc-com},
$\mathcal{R}_3$ is confluent.
Thus by Proposition~\ref{prop:reduction equivalence} (ii),
$\mathcal{R}_2$ is confluent too.
\end{exa}

\begin{figure}
\renewcommand{\arraystretch}{3}    
\[
\begin{array}{l}
\deduce{~\strut}{\textit{Partition}}\quad
\infer[\small~\mbox{$\mathcal{S} \cup \mathcal{P} = \mathcal{S}' \cup \mathcal{P}'$,
                    $\mathcal{P}'$: reversible}]
  {\langle \mathcal{S}', \mathcal{P}' \rangle}
  {\langle \mathcal{S}, \mathcal{P} \rangle}\\
\deduce{~\strut}{\textit{Replacement}}\quad
\infer[\small~r \stackrel{*}{\leftrightarrow}_\mathcal{P} r']
  {\langle \mathcal{S} \cup \{ l \to r' \}, \mathcal{P} \rangle}
  {\langle \mathcal{S} \cup \{ l \to r \}, \mathcal{P} \rangle}\\
\deduce{~\strut}{\textit{Addition}}\quad
\infer[\small~l \stackrel{*}{\leftrightarrow}_\mathcal{P} \circ \stackrel{*}{\to}_\mathcal{S} r]
  {\langle \mathcal{S} \cup \{ l \to r \},  \mathcal{P} \rangle}
  {\langle \mathcal{S}, \mathcal{P} \rangle}\\
\end{array}
\]
\caption{Inference rules of reduction-preserving completion}
\label{fig:reduction-preserving completion}
\end{figure}

In this example, two additional rewrite rules $(\mathsf{add}_2)$ 
and $(\mathsf{add}_3)$ are given
by hand.
But in automated confluence proving procedures,
one needs to find such new rewrite rules automatically.
We next present a completion-like procedure
to automate such additions (or more generally transformations)
of rewrite rules.
We first present an abstract version of the procedure
in the form of inference rules and prove its soundness
w.r.t.\ the confluence proof.

\begin{defi}[abstract reduction-preserving completion procedure]
\label{def:abstract reduction-preserving completion procedure}
Inference rules of an \textit{abstract reduction-preserving completion procedure}
are listed in Figure~\ref{fig:reduction-preserving completion}.
The inference rules act on a pair of TRSs
$\mathcal{S}$ and $\mathcal{P}$ where
$\mathcal{P}$ is reversible.
One step derivation using any of inference rules (from upper to lower)
is denoted by $\leadsto$.
We also write $\leadsto^p$ ($\leadsto^r, \leadsto^a$)
for an inference step by the rule
\textit{Partition} (\textit{Replacement}, \textit{Addition}, respectively).
\end{defi}

\begin{thm}[soundness of the abstract reduction-preserving completion procedure]
\label{thm:soundness of the abstract reduction-preserving completion procedure}
Let 
$\langle \mathcal{R}, \emptyset \rangle
= \langle \mathcal{S}_0, \mathcal{P}_0 \rangle
\stackrel{*}{\leadsto}
\langle \mathcal{S}_n,\mathcal{P}_n \rangle$
be a derivation of abstract reduction-preserving completion procedure.
If $\mathcal{S}_n \cup \mathcal{P}_n$
is confluent then $\mathcal{R}$ is confluent.
\end{thm} 

\begin{proof}
We show,  for any inference step
$\langle \mathcal{S}_i,\mathcal{P}_i \rangle
\leadsto 
\langle \mathcal{S}_{i+1},\mathcal{P}_{i+1} \rangle$,
that $\mathcal{S}_i \cup \mathcal{P}_i$
and $\mathcal{S}_{i+1} \cup \mathcal{P}_{i+1}$
are reduction equivalent
and that $\mathcal{P}_{i+1}$ is reversible
whenever so is $\mathcal{P}_i$.
\begin{iteMize}{$\bullet$}
\item
$\langle \mathcal{S}_i,\mathcal{P}_i \rangle
\leadsto 
\langle \mathcal{S}_{i+1},\mathcal{P}_{i+1} \rangle$
by \textit{Partition}.
Then 
since $\mathcal{S}_i \cup \mathcal{P}_i = 
\mathcal{S}_{i+1} \cup \mathcal{P}_{i+1}$
and $\mathcal{P}_{i+1}$ is reversible by the side condition,
the claim follows immediately.

\item
$\langle \mathcal{S}_i,\mathcal{P}_i \rangle
\leadsto 
\langle \mathcal{S}_{i+1},\mathcal{P}_{i+1} \rangle$
by \textit{Replacement}.
Then 
$\mathcal{S}_i = \mathcal{S}_i' \cup \{ l \to r \}$,
$r \stackrel{*}{\leftrightarrow}_{\mathcal{P}_i} r'$ and 
$\mathcal{S}_{i+1} = \mathcal{S}_i' \cup \{ l \to r' \}$
for some $\mathcal{S}_i'$, $l,r,r'$
and $\mathcal{P}_{i+1} = \mathcal{P}_i$.
By the reversibility of $\mathcal{P}_i$,
we have $l \to_{\mathcal{S}_{i}} r \stackrel{*}{\to}_{\mathcal{P}_i} r'$
hence 
${\to}_{\mathcal{S}_{i+1} \cup \mathcal{P}_{i+1}}
\subseteq 
{\stackrel{*}{\to}}_{\mathcal{S}_i \cup \mathcal{P}_i}$.
By the reversibility of $\mathcal{P}_i$,
we also have $l \to_{\mathcal{S}_{i+1}} r' \stackrel{*}{\to}_{\mathcal{P}_i} r$,
hence 
${\to}_{\mathcal{S}_i \cup \mathcal{P}_i} 
\subseteq 
{\stackrel{*}{\to}}_{\mathcal{S}_{i+1} \cup \mathcal{P}_{i+1}}$.
Thus by Proposition~\ref{prop:reduction equivalence} (i),
$\mathcal{S}_i \cup \mathcal{P}_i$
and $\mathcal{S}_{i+1} \cup \mathcal{P}_{i+1}$
are reduction equivalent.
Hence, by $\mathcal{P}_{i+1} = \mathcal{P}_i$, 
the claim follows.

\item
$\langle \mathcal{S}_i,\mathcal{P}_i \rangle
\leadsto 
\langle \mathcal{S}_{i+1},\mathcal{P}_{i+1} \rangle$
by \textit{Addition}.
Then 
$l \stackrel{*}{\leftrightarrow}_{\mathcal{P}_i} \circ \stackrel{*}{\to}_{\mathcal{S}_i} r$
and 
$\mathcal{S}_{i+1} = \mathcal{S}_i \cup \{ l \to r \}$
for some $l,r$
and and $\mathcal{P}_{i+1} = \mathcal{P}_i$.
Since $\mathcal{S}_i \cup \mathcal{P}_i \subseteq 
\mathcal{S}_{i+1} \cup \mathcal{P}_{i+1}$,
we have 
${\to}_{\mathcal{S}_i \cup \mathcal{P}_i} \subseteq 
{\stackrel{*}{\to}}_{\mathcal{S}_{i+1} \cup \mathcal{P}_{i+1}}$.
By the reversibility of $\mathcal{P}_i$,
$l \stackrel{*}{\to}_{\mathcal{P}_i}
{\circ} \stackrel{*}{\to}_{\mathcal{S}_i} r'$.
Hence
${\to}_{\mathcal{S}_{i+1} \cup \mathcal{P}_{i+1}}
\subseteq 
{\stackrel{*}{\to}}_{\mathcal{S}_i \cup \mathcal{P}_i}$.
Thus by Proposition~\ref{prop:reduction equivalence} (i),
$\mathcal{S}_i \cup \mathcal{P}_i$
and $\mathcal{S}_{i+1} \cup \mathcal{P}_{i+1}$ are reduction equivalent.
Hence, by $\mathcal{P}_{i+1} = \mathcal{P}_i$, 
the claim follows.

\end{iteMize}
Thus by induction on $n$, it follows that
$\mathcal{R}$ and $\mathcal{S}_n \cup \mathcal{P}_n$
are reduction equivalent. 
Hence if $\mathcal{S}_n \cup \mathcal{P}_n$ is confluent,
$\mathcal{R}$ is confluent
by Proposition~\ref{prop:reduction equivalence} (ii).
\end{proof}

\begin{exa}[derivations in abstract reduction-preserving completion procedure]
\label{exp:successful derivation}
The confluence proof of Example~\ref{exp:confluence by reduction equivalence}
is derived by the abstract reduction-preserving completion procedure.
Give $\mathcal{R}_2 = \{ (\mathsf{add}_1), (\mathsf{add}_2), (C), (A) \}$ 
as the input to the procedure.
Let us consider the following derivation.
\[
\begin{array}{lcll}
\langle \{ (\mathsf{add}_1), (\mathsf{add}_2), (C), (A) 
\}, \emptyset \rangle
&\leadsto^p & \langle \{ (\mathsf{add}_1), (\mathsf{add}_2) \}, \{ (C), (A) \} \rangle 
\\
&\leadsto^a & \langle \{ (\mathsf{add}_1), (\mathsf{add}_2),(\mathsf{add}_3) \}, \{ (C), (A) \} \rangle 
\\
&\leadsto^a & \langle \{ (\mathsf{add}_1), (\mathsf{add}_2),
(\mathsf{add}_3), (\mathsf{add}_4') \}, \{ (C), (A) \} \rangle 
\\
&\leadsto^r & \langle \{ (\mathsf{add}_1), (\mathsf{add}_2),
(\mathsf{add}_3),(\mathsf{add}_4)  \}, \{ (C), (A) \} \rangle 
\\
\end{array}
\]
where $(\mathsf{add}_4'): \mathsf{+}(x,\mathsf{s}(y)) \to \mathsf{s}(\mathsf{+}(y,x))$.
As we have already demonstrated in Example \ref{exp:add2pp-assoc-com},
$\{ (\mathsf{add}_1),(\mathsf{add}_2),
(\mathsf{add}_3),(\mathsf{add}_4), (A), (C) \} = \mathcal{R}_3 
$ is confluent.
Thus, by Theorem~\ref{thm:soundness of the abstract reduction-preserving completion procedure},
$\mathcal{R}_2$ is confluent.
\end{exa}

\subsection{Concrete completion procedure}

In this subsection,
we present a concrete reduction-preserving completion procedure
that can be used as the basis of an automated completion procedure.
The procedure presented below is designed so as to 
apply Theorem~\ref{thm:parallel},
but it is straightforward to modify the procedure
suitable for Theorems~\ref{thm:linear}, \ref{thm:PCP}
or Corollaries \ref{cor:linear}, \ref{cor:parallel}, \ref{cor:PCP}
or any combinations of them.

\begin{defi}[concrete reduction-preserving completion procedure]\hfill\\
\label{def:concrete reduction-preserving completion procedure}
Input: a TRS $\mathcal{R}$\\
Output: \texttt{Success} or \texttt{Failure} (or may diverge)
\begin{enumerate}[Step 1]
\item
Put $\mathcal{Q}_0 := \mathcal{R}$ and $i:= 0$.
Proceed to Step 2.
\item
Take a partition 
$\mathcal{S}_i \cup \mathcal{P}_i = \mathcal{Q}_i $
such that 
$\mathcal{S}_i$ is left-linear and terminating,
$\mathcal{P}_i$ is reversible and
$\mathrm{CP}_\textit{in}(\mathcal{P}_i\cup{\mathcal{P}_i}^{-1},\mathcal{S}_i) = \emptyset$.
Proceed to Step 3.
If there is no such a partition then return \texttt{Failure}.
\item
Set $\mathcal{U}_1 := \emptyset$.\\
For each $\langle u, v \rangle  \in
\mathrm{CP}(\mathcal{S}_i,\mathcal{P}_i\cup {\mathcal{P}_i}^{-1})$, do:
\begin{iteMize}{$\bullet$}
\item Take $\mathcal{S}_i$-normal forms $\hat u, \hat v$ of $u,v$, respectively.
\item Check whether $\hat u \parto_{\mathcal{P}_i\cup {\mathcal{P}_i}^{-1}}  \hat v$.
If not then put $\mathcal{U}_1 := \mathcal{U}_1 \cup \{ \langle b, v \approx \hat u \rangle \}$,
where $b : = \mathtt{true}$ or $:= \mathtt{false}$
according to whether $v$ is $\mathcal{S}_i$-normal or not.
\end{iteMize}
\item
Set $\mathcal{U}_2 := \emptyset$.\\
For each $\langle u, v \rangle 
\in \mathrm{CP}({\mathcal{S}_i},{\mathcal{S}_i})$, do:
\begin{iteMize}{$\bullet$}
\item Take $\mathcal{S}_i$-normal forms $\hat u, \hat v$ of $u,v$, respectively.
\item Check whether $\hat u \pargets_{\mathcal{P}_i\cup {\mathcal{P}_i}^{-1}}  \hat v$. If 
not then put $\mathcal{U}_2 := \mathcal{U}_2 \cup \{ \langle \mathsf{false},
\hat u  \approx \hat v \rangle \}$.
\end{iteMize}
If $\mathcal{U}_1 = \mathcal{U}_2 =  \emptyset$
then return \texttt{Success}.

\item
Let $\mathcal{P}' := \emptyset$.\\
For each $\langle b, u \approx v \rangle  \in 
\mathcal{U}_1 \cup \mathcal{U}_2$, do:
\begin{iteMize}{$\bullet$}
\item
Check whether there exist $u',v'$
such that
$u 
= u'
\pargets_{\mathcal{P}_i\cup {\mathcal{P}_i}^{-1}} 
v'
\stackrel{*}{\gets}_{\mathcal{S}_i \cup \mathcal{P}_i\cup {\mathcal{P}_i}^{-1}}  
v$
if $b = \mathtt{true}$, 
and 
$u 
\to_{\mathcal{S}_i}
\circ 
\stackrel{*}{\to}_{\mathcal{S}_i \cup \mathcal{P}_i\cup {\mathcal{P}_i}^{-1}}  
u'
\pargets_{\mathcal{P}_i\cup {\mathcal{P}_i}^{-1}} 
v'
\stackrel{*}{\gets}_{\mathcal{S}_i \cup \mathcal{P}_i\cup {\mathcal{P}_i}^{-1}}  
v$ if $b = \mathsf{false}$.
\item Put $\mathcal{P}' := \mathcal{P}' \cup \{ l \to r 
\in \mathcal{P}_i\cup {\mathcal{P}_i}^{-1}
\mid$
$l \to r$
is used in the rewrite sequences
$u 
\stackrel{*}{\to}_{\mathcal{S}_i \cup \mathcal{P}_i\cup {\mathcal{P}_i}^{-1}}  
u'$ or
$v
\stackrel{*}{\to}_{\mathcal{S}_i \cup \mathcal{P}_i\cup {\mathcal{P}_i}^{-1}}  
v' \}$.
\end{iteMize}
Finally, check whether $\mathcal{S}_i$ is terminating relative to $\mathcal{P}'$.
If it is, then return \texttt{Success}.
Otherwise take some set $\mathcal{U}' 
\subseteq 
\{ v \to \hat u \mid \langle b, v \approx \hat u \rangle
\in \mathcal{U}_1 \}
\cup 
\{ l \to r, r \to l  \mid 
\langle b, l \approx r \rangle \in \mathcal{U}_2,
l \stackrel{*}{\leftrightarrow}_{\mathcal{P}_i}
r \}$
of rewrite rules 
and put $\mathcal{Q}_{i+1} := \mathcal{Q}_{i} \cup \mathcal{U}'$, $i:= i+1$
and go to Step 2.
\end{enumerate}

\medskip

\noindent
During the step 2, one may perform the following additional steps.
\begin{enumerate}[Step 1a.]
\item[Step 2a.]
If there exist $l \to r \in \mathcal{S}_i$ and $r'$ such that 
$r \leftrightarrow_{\mathcal{P}_i} r'$
and $\mathrm{CP}_\textit{in}({\mathcal{P}_i}\cup{\mathcal{P}_i}^{-1},\{ l \to r \}) \neq \emptyset$,
then put $\mathcal{Q}_{i+1} := (\mathcal{Q}_{i} \setminus \{ l \to r \}) \cup \{ l \to r' \}$,
$i:= i+1$.
\item[Step 2b.]
Let $\langle u, v \rangle
\in \mathrm{CP}_\textit{in}({\mathcal{P}_i}\cup{\mathcal{P}_i}^{-1},{\mathcal{S}_i})$
and let $\hat v$ be an $\mathcal{S}_i$-normal form of $v$.
Then put $\mathcal{Q}_{i+1} := \mathcal{Q}_{i} \cup \{ u \to \hat v \}$
and $i:= i+1$.
\end{enumerate}


\medskip

\noindent
Before moving from step 4 to step 2,
one may perform the following additional step.

\begin{enumerate}[Step 1a.]
\item[Step 4a.]
Set $\mathcal{S}_i := \mathcal{S}_{i-1}, \mathcal{P}_i := \mathcal{P}_{i-1}$.
If there exist $l \to r \in \mathcal{S}_i$ and $r'$
such that $r \leftrightarrow_{\mathcal{P}_i} r'$ and
there exists $\langle u, v \rangle \in
 \mathrm{CP}(\{ l \to r \},{\mathcal{P}_i}\cup{\mathcal{P}_i}^{-1})
\cup \mathrm{CP}(\{ l \to r \},{\mathcal{S}_i})
\cup \mathrm{CP}(\mathcal{S}_i, \{ l \to r \})$
such that 
the critical pair conditions
are not satisfied,
then put $\mathcal{Q}_{i+1} := (\mathcal{Q}_{i} \setminus \{ l \to r \}) \cup \{ l \to r' \}$,
$i:= i+1$.
\end{enumerate}
\end{defi}

\begin{rem}
Steps 2--5 try to show the condition 
of Theorem~\ref{thm:parallel} directly,
and if the relative termination check 
of $\mathcal{S}_i$ relative to $\mathcal{P}'$
fails in Step 5,
some of critical pairs which lead to the equations 
in $\mathcal{U}_1 \cup\mathcal{U}_2$ are problematic.
Thus, taking some non-empty $\mathcal{U}'$ is
indispensable for the success of the completion procedure.
On the other hand, it is not at all guaranteed
whether extra steps 2a, 2b, 4a are helpful
for the success of the completion procedure---they
just add some flexibilities to modify equations.
Adding such flexibilities may be helpful but 
there is a trade-off between adding such 
extra steps and efficiency of the completion procedure.
We here present these extra steps because they 
perform well for some examples, including
Example~\ref{exp:extra steps} below.
\end{rem}

\begin{cor}[soundness of the concrete reduction-preserving completion procedure]
\label{cor:soundness of the concrete reduction-preserving completion procedure}
If the procedure of Definition~\ref{def:concrete reduction-preserving completion procedure}
succeeds for the input $\mathcal{R}$,
then $\mathcal{R}$ is confluent.
\end{cor}

\begin{proof}
It suffices to show
if the procedure succeeds
then there exists a successful derivation
of the abstract reduction-preserving completion procedure
ending with $\mathcal{S}_n, \mathcal{P}_n$
satisfying the conditions of Theorem \ref{thm:parallel}.
Step 1 corresponds to the empty derivation.
Step 2 corresponds to an inference step by \textit{Partition}.
If the procedure succeeds at Step 4,
then $\mathcal{S}_i, \mathcal{P}_i$
satisfy the conditions of Corollary \ref{cor:parallel}
and hence that of Theorem \ref{thm:parallel}.
If the procedure succeeds at Step 5,
then $\mathcal{S}_i, \mathcal{P}_i$
satisfy the conditions of Theorem \ref{thm:parallel}.
Suppose that the procedure does not return $\mathtt{Success}$
at Steps 4 or 5.
Then it is readily checked that
for $l \to r \in \mathcal{U}'$,
$l \gets_{\mathcal{S}_i} \circ \to_{\mathcal{P}_i \cup \mathcal{P}_i^{-1}} r$ 
or 
$l \stackrel{*}{\leftrightarrow}_{\mathcal{P}_i} r$ holds.
Thus, in this case, Step 5 is simulated by multiple inference steps by \textit{Addition}.
Similarly, Step 2b is simulated by multiple inference steps by \textit{Addition}
and Steps 2a and 4a are simulated by inference steps by \textit{Replace}.
\end{proof}

\begin{exa}
\label{exp:plus1zero-plus2succ-complus-assocplus}
Give
$\mathcal{R}_2 = \{ (\mathsf{add}_1), (\mathsf{add}_2), (C), (A)
\}$
as the input.
\begin{enumerate}[\phantom0(1)]
\item (Step 1) 
We put $\mathcal{Q}_0 := \mathcal{R}_2$.
\item (Step 2)
We take $\mathcal{S}_0 = \{ (\mathsf{add}_1),(\mathsf{add}_2) \}$
and $\mathcal{P}_0 = \{ (C),(A) \}$.
Then $\mathcal{S}_0$ is left-linear and terminating,
$\mathcal{P}_0$ is reversible
and
$\mathrm{CP}_\textit{in}(\mathcal{P}_0\cup{\mathcal{P}_0}^{-1},\mathcal{S}_0) = \emptyset$.
\item (Step 3)
We have
$\mathrm{CP}(\mathcal{S}_0,\mathcal{P}_0\cup {\mathcal{P}_0}^{-1})=$
\[
\left\{
\begin{array}{cl@{\qquad}cl}
\mbox{(1)} &\langle 
\mathsf{+}(y,z),
\mathsf{+}(0,\mathsf{+}(y,z))
\rangle&
\mbox{(5)} &\langle 
\mathsf{s}(\mathsf{+}(x,\mathsf{+}(y,z))),
\mathsf{+}(\mathsf{+}(\mathsf{s}(x),y),z)
\rangle\\
\mbox{(2)} &\langle 
\mathsf{+}(y,z),
\mathsf{+}(\mathsf{+}(0,y),z)
\rangle&
\mbox{(6)} &\langle 
\mathsf{+}(\mathsf{s}(\mathsf{+}(x,y)),z),
\mathsf{+}(\mathsf{s}(x),\mathsf{+}(y,z))
\rangle\\
\mbox{(3)} &\langle 
\mathsf{+}(x,z),
\mathsf{+}(\mathsf{+}(x,0),z)
\rangle&
\mbox{(7)} &\langle 
\mathsf{+}(x,\mathsf{s}(\mathsf{+}(y,z))),
\mathsf{+}(\mathsf{+}(x,\mathsf{s}(y)),z)
 \rangle\\
\mbox{(4)} &\langle 
y,
\mathsf{+}(y,0)
\rangle&
\mbox{(8)} &\langle 
\mathsf{s}(\mathsf{+}(x,y)),
\mathsf{+}(y,\mathsf{s}(x))
\rangle\\  
\end{array}
\right\}.
\]
Then for $\langle u, v \rangle \in
\{ \mbox{(3)}, \mbox{(4)}, \mbox{(7)}, \mbox{(8)} \}$,
$u \stackrel{*}{\to}_{\mathcal{S}_0}
\circ 
\pargets_{\mathcal{P}_0 \cup \mathcal{P}_0^{-1}}
\circ 
 \stackrel{*}{\gets}_{\mathcal{S}_0}
v$ does not hold and $v$ is $\mathcal{S}_0$-normal.
Thus we put
\[
\mathcal{U}_1 :=
\left\{ 
\begin{array}{l}
\langle \mathtt{true}, 
\mathsf{+}(\mathsf{+}(x,0),z) \approx \mathsf{+}(x,z)
\rangle\\
\langle \mathtt{true}, 
\mathsf{+}(y,0) \approx y
\rangle\\
\langle \mathtt{true}, 
\mathsf{+}(\mathsf{+}(x,\mathsf{s}(y)),z) \approx  \mathsf{+}(x,\mathsf{s}(\mathsf{+}(y,z)))
\rangle\\
\langle \mathtt{true}, 
\mathsf{+}(y,\mathsf{s}(x)) \approx \mathsf{s}(\mathsf{+}(x,y))
\rangle 
\end{array}
\right\}
\]
and proceed to Step 4.

\item (Step 4)
We have
$\mathrm{CP}(\mathcal{S}_0,\mathcal{S}_0)= \emptyset$
and thus $\mathcal{U}_2 = \emptyset$.
Since $b = \mathtt{true}$ for any $\langle b, u \approx v \rangle \in \mathcal{U}_1$,
we check 
$u \pargets_{\mathcal{P}_0\cup {\mathcal{P}_0}^{-1}} 
\circ 
\stackrel{*}{\gets}_{\mathcal{S}_0 \cup \mathcal{P}_0\cup {\mathcal{P}_0}^{-1}}  
v$.
But this fails since, for example,
$\mathsf{+}(y,0)
 \pargets_{\mathcal{P}_0\cup {\mathcal{P}_0}^{-1}} 
\circ 
\stackrel{*}{\gets}_{\mathcal{S}_0 \cup \mathcal{P}_0\cup {\mathcal{P}_0}^{-1}}  
y
$ does not hold.
Now we put 
\[
\mathcal{U}' :=
\left\{
\begin{array}{clcl}
(\mathsf{add}_3)&\mathsf{+}(y,0)\to y \qquad
(\mathsf{add}_4')& \mathsf{+}(y,\mathsf{s}(x)) \to \mathsf{s}(\mathsf{+}(x,y))\\
\end{array}
\right\}.
\]
and go to the step 2.
\item (Step 2)
We take $\mathcal{S}_1 = \{ (\mathsf{add}_1),(\mathsf{add}_2),(\mathsf{add}_3),(\mathsf{add}_4')
 \}$ and $\mathcal{P}_1 = \{ (C),(A) \}$.
Then $\mathcal{S}_1$ is left-linear and terminating,
$\mathcal{P}_1$ is reversible,
and
$\mathrm{CP}_\textit{in}(\mathcal{P}_1\cup{\mathcal{P}_1}^{-1},\mathcal{S}_1) = \emptyset$.

\item (Step 3)
There are four elements
in $\mathrm{CP}^{-1}(\mathcal{S}_1,\mathcal{P}_1\cup {\mathcal{P}_1}^{-1})$
which are not joinable as 
$u \stackrel{*}{\to}_{\mathcal{S}_1}
\circ 
\pargets_{\mathcal{P}_1 \cup \mathcal{P}_1^{-1}}
\circ 
 \stackrel{*}{\gets}_{\mathcal{S}_1}
v$, namely
\[
\left\{
\begin{array}{clcl}
(9) & \langle \mathsf{+}(x,\mathsf{s}(\mathsf{+}(y,z)),  \mathsf{+}(\mathsf{+}(x,\mathsf{s}(y)),z) \rangle\\
(10) & \langle \mathsf{s}(\mathsf{+}(x,\mathsf{+}(y,z)),  \mathsf{+}(y,\mathsf{+}(z,\mathsf{s}(x))) \rangle\\
(11) &  \langle \mathsf{+}(\mathsf{s}(\mathsf{+}(x,y)),z),  \mathsf{+}(y,\mathsf{+}(\mathsf{s}(x),z)) \rangle\\
(12) &\langle \mathsf{+}(x,\mathsf{s}(\mathsf{+}(y,z))),  \mathsf{+}(\mathsf{+}(x,z),\mathsf{s}(y)) \rangle
\end{array}
\right\}.
\]
We put
\[
\mathcal{U}_1 
:=
\left\{ 
\begin{array}{l}
\langle \mathtt{false}, 
\mathsf{+}(\mathsf{+}(x,\mathsf{s}(y)),z) \approx  \mathsf{s}(\mathsf{+}(x,\mathsf{+}(y,z)))
\rangle\\
\langle \mathtt{false}, 
\mathsf{+}(y,\mathsf{+}(z,\mathsf{s}(x))) \approx \mathsf{s}(\mathsf{+}(x,\mathsf{+}(y,z))
\rangle\\
\langle \mathtt{false}, 
\mathsf{+}(y,\mathsf{+}(\mathsf{s}(x),z)) \approx \mathsf{s}(\mathsf{+}(\mathsf{+}(x,y),z)
\rangle\\
\langle \mathtt{false}, 
\mathsf{+}(\mathsf{+}(x,z),\mathsf{s}(y)) \approx \mathsf{s}(\mathsf{+}(x,\mathsf{+}(y,z)))
\rangle 
\end{array}
\right\}
\]
and proceed to Step 4.

\item (Step 4)
We have
$\mathrm{CP}(\mathcal{S}_1,\mathcal{S}_1)= \emptyset$
and thus $\mathcal{U}_2 = \emptyset$.
Since $b = \mathtt{false}$ for any $\langle b, u \approx v \rangle \in \mathcal{U}_1$,
we check 
$u 
\to_{\mathcal{S}_1}
\circ 
\stackrel{*}{\to}_{\mathcal{S}_1 \cup \mathcal{P}_1\cup {\mathcal{P}_1}^{-1}}  
u'
\pargets_{\mathcal{P}_1\cup {\mathcal{P}_1}^{-1}} 
v'
\stackrel{*}{\gets}_{\mathcal{S}_1 \cup \mathcal{P}_1\cup {\mathcal{P}_1}^{-1}}  
v$.
Then we have 
$u 
\to_{\mathcal{S}_1}
\circ 
\stackrel{*}{\to}_{\mathcal{S}_1 \cup \mathcal{P}'}
u'
\pargets_{\mathcal{P}_1\cup {\mathcal{P}_1}^{-1}} 
v'
\stackrel{*}{\gets}_{\mathcal{S}_1 \cup \mathcal{P}'}
v$ for all 
$\langle b, u \approx v \rangle \in \mathcal{U}_1$
by taking $\mathcal{P}' = \{ \mathsf{+}(\mathsf{+}(x,y),z) \to \mathsf{+}(x,\mathsf{+}(y,z)) \}$.
Then 
$\mathcal{S}_1$ is terminating relative to
$\mathcal{P}'$
and $\mathtt{Success}$ is returned.
\end{enumerate}
\end{exa}

\begin{exa}
\label{exp:extra steps}
Suppose that 
the relative termination check of
$\mathcal{S}_1$ relative to $\mathcal{P}'$ fails
in the last step of
Example \ref{exp:plus1zero-plus2succ-complus-assocplus}.
Then the procedure continues as follows.
\begin{enumerate}[\phantom0(1)]
\addtocounter{enumi}{7}
\item Here, we put $\mathcal{U}' := \emptyset$,
$\mathcal{Q}_2 := \mathcal{Q}_1$, $i := 2$ and proceed to Step 4a.
\item (Step 4a)
Put $\mathcal{S}_2 := \mathcal{S}_1$
and $\mathcal{P}_2 := \mathcal{P}_1$.
Since 
$\mbox{(11)} \in 
\mathrm{CP}(\{ (\mathsf{add}_4') \},\mathcal{P}_2\cup {\mathcal{P}_2}^{-1})$
and 
$\mathsf{s}(\mathsf{+}(x,y)) \to_{\mathcal{P}_2} \mathsf{s}(\mathsf{+}(y,x))$,
one can put $\mathcal{Q}_3 := 
(\mathcal{Q}_2 \setminus \{ (\mathsf{add}_4') \})
\cup 
(\mathsf{add}_4) 
\}$
and $i:= 3$ and go to Step 2.
\item (Step 2)
We take $\mathcal{S}_3 = \{ (\mathsf{add}_1),(\mathsf{add}_2),
(\mathsf{add}_3),(\mathsf{add}_4) \}$
and $\mathcal{P}_3 = \{ (C), (A) \}$.
Then $\mathcal{S}_3$ is left-linear and terminating,
$\mathcal{P}_3$ is reversible
and
$\mathrm{CP}_\textit{in}(\mathcal{P}_3\cup{\mathcal{P}_3}^{-1},\mathcal{S}_3) = \emptyset$.
Thus proceed to Step 3.

\item (Step 3)
For any $\langle u, v \rangle \in
\mathrm{CP}(\mathcal{S}_3,\mathcal{P}_3\cup {\mathcal{P}_3}^{-1})$,
we have $u \stackrel{*}{\to}_{\mathcal{S}_3}
\circ 
\pargets_{\mathcal{P}_3 \cup \mathcal{P}_3^{-1}}
\circ 
 \stackrel{*}{\gets}_{\mathcal{S}_3}
v$ (Example~\ref{exp:add2pp-assoc-com}).
Thus $\mathcal{U}_1 := \emptyset$
and proceed to Step 4.

\item (Step 4)
For any $\langle u, v \rangle \in
\mathrm{CP}(\mathcal{S}_3,\mathcal{S}_3)$,
we have $u \stackrel{*}{\to}_{\mathcal{S}_3}
\circ 
\pargets_{\mathcal{P}_3 \cup \mathcal{P}_3^{-1}}
\circ 
 \stackrel{*}{\gets}_{\mathcal{S}_3}
v$ (Example~\ref{exp:add2pp-assoc-com}).
Thus $\mathcal{U}_2 := \emptyset$.
Since $\mathcal{U}_1 = \mathcal{U}_2 = \emptyset$,
\texttt{Success} is returned.
\end{enumerate}
\end{exa}

\section{Implementation and experiments}


All the results of this paper have been implemented.
The program is written in 
\textsf{SML/NJ}\footnote{\texttt{http://www.smlnj.org/}}
and is built upon our confluence prover \textsf{ACP}\footnote{
\texttt{http://www.nue.riec.tohoku.ac.jp/tools/acp/}}
\cite{DD,ACP,YoshidaJSSST}.

In Figure~\ref{tab:pseude code}, we present pseudo-code 
of main functions of our implementation of reduction-preserving completion procedure
enough for describing some heuristics employed in the implementation.
Short descriptions of functions 
involved in our pseudo-code and heuristics 
employed follow.

\begin{iteMize}{$\bullet$}
\item (\texttt{checkConfluence}~$\mathcal{R}$)
is the main function of the procedure.
It simulates multiple runs in the breadth-first strategy.



\item (\texttt{decompose}~$\mathcal{R}$)
decomposes $\mathcal{R}$ into $\mathcal{S} \cup \mathcal{P}$
and duplicates $\mathcal{S} \cup \mathcal{P}$.
Hence a list of triples $(\mathcal{S},\mathcal{P},i)$ where
$\mathcal{S} \cup \mathcal{P} = \mathcal{R}$ and $i \in \{ 0, 1, 2\}$ are returned.
Here, however, 
instead of returning all partitions,
we select partitions based on a heuristic, namely that  
$\mathcal{P}$ is a set of the rules $l \to r$ satisfying
either 
(1) $r \to l \in \mathcal{R}$ or
(2) $\mathcal{F}(l) = \mathcal{F}(r)$ and $l(\epsilon) = r(\epsilon)$.

\item (\texttt{check}~$(\mathcal{S},\mathcal{P},i$)) checks
whether conditions of Theorem~\ref{thm:linear} when $i = 0$
(or Theorem~\ref{thm:PCP} when $i = 1$, 
Proposition~\ref{prop:relative} when $i = 2$)
are satisfied
based on the algorithm given 
in our concrete completion procedure
(Definition \ref{def:concrete reduction-preserving completion procedure}).
Reversibility is tested by checking
$r \stackrel{\le k}{\to} l$ 
(i.e.\ there is a rewrite sequence 
from $r$ to $l$ of length less than or equal to $k$ steps)
for some 
constant $k$.
In our implementation, we set $k = 10$.

\item (\texttt{trans}~($\mathcal{S},\mathcal{P}$)~$\texttt{nj}$)
returns a collection of transformed TRSs 
obtained by addition and replacement of rewrite rules
constructed from non-joinable critical pairs 
(Steps 5 and 2b
of the Definition~\ref{def:concrete reduction-preserving completion procedure})
and 
rewrite rules generating such critical pairs
(Steps 2a and 4a
of of the Definition~\ref{def:concrete reduction-preserving completion procedure}).
Here, the addition of rewrite rules is restricted
based on some heuristic.\medskip
\end{iteMize}

\begin{figure*}
\begin{center}
\[
\begin{array}{l}
\texttt{fun}~\texttt{check~(S,P,i)} = \texttt{if~i} = 0~
\texttt{then}~\mbox{\it (apply Theorem~\ref{thm:PCP})}\\
\qquad\qquad\qquad\qquad\qquad
\texttt{else}~~\texttt{if~i} = 1~~\mathtt{then}~\mbox{\it (apply Theorem~\ref{thm:linear})}\\
\qquad\qquad\qquad\qquad\qquad
\texttt{else}~~\mbox{\it (apply Proposition~\ref{prop:relative})}\\
\texttt{fun}~\texttt{checkConfluence~R} = \\
\qquad\texttt{let}~\texttt{fun}~\texttt{step}~\texttt{[\,]} = \texttt{Failure} \\
\qquad\qquad\texttt{|}\quad\,\,\texttt{step~((S,P,i)::rest)} =
\texttt{case~check~(S,P,i)~of}\\
\qquad\qquad\qquad\qquad\texttt{NONE} \Rightarrow \texttt{step~rest}\\
\qquad\qquad\qquad\quad\texttt{|~SOME~[]}\Rightarrow \texttt{Success}\\
\qquad\qquad\qquad\quad\texttt{|~SOME~nj}\Rightarrow \texttt{step~(rest~@}\\
\qquad\qquad\qquad\qquad\qquad\qquad\qquad\qquad
\texttt{(mapAppend~decompose~(trans~(S,P)~nj)))}\\
\qquad\texttt{in~step~(decompose~R)~end}
\end{array}
\]
\caption{Pseudo-code of main functions}
\label{tab:pseude code}
\end{center}
\end{figure*}

\noindent Table~\ref{table:summary} shows the summary of our experiments.
We have tested various combinations of our results:
Rows (1)--(8) are proofs by confluence criteria
of Theorems \ref{thm:linear}, \ref{thm:parallel}, \ref{thm:PCP}
(Corollaries \ref{cor:linear}, \ref{cor:parallel}, \ref{cor:PCP})
and by the combination
of those of Theorem \ref{thm:PCP} (Corollary \ref{cor:PCP})
and Theorem \ref{thm:linear} (Corollary \ref{cor:linear}, respectively).
Those marked with ``w/o RT'' are the ones without
(proper) relative termination checking (i.e.\
Corollaries \ref{cor:linear}, \ref{cor:parallel} and \ref{cor:PCP}
where only termination checking is involved).
Rows (9)--(14) are proofs by the reduction-preserving completion
without the \textit{Replacement} rule,
i.e.\ the Steps $2a,4a$ of
the concrete reduction-preserving completion
(Definition~\ref{def:concrete reduction-preserving completion procedure}).
Rows (15)--(20) are proofs by the reduction-preserving completion
with the \textit{Replacement} rule.
Row (21) are proofs by the reduction-preserving completion
with the \textit{Replacement} rule and 
Huet's criterion (Proposition \ref{prop:relative}).

For the experiments, we used a collection of 85 TRSs 
involving non-terminating rules 
such as commutativity and associativity rules
which have been developed in the course of experiments\footnote{
Four TRSs have been added to the collection of our 
proceedings version \cite{Rcomp} of the paper 
to capture Examples \ref{exp:add2pe-dbl-assoc-com}
and \ref{exp:loop2}.}.
All experiments have been performed 
on a Linux platform of a PC equipped with 1.2GHz CPU and 1GB memory.
The maximal number of steps of the completion procedure
is limited to 20 steps; the columns below the title ``diverge'' 
show the numbers of examples which exceeded this limit,
where these numbers are included in those of ``failure.''
We set the timeout 60 sec.
The columns below the title ``timeout'' 
show the numbers of examples which exceeded this timeout.
Total time (real time) is measured in seconds.

When relative termination checking is disabled 
and only termination checking is used,
the applicability of incomparable confluent criteria
(Corollary \ref{cor:PCP} and Corollary \ref{cor:linear})
does not show much differences in the number of success.
When relative termination checking is activated,
Theorem \ref{thm:PCP} has a clear
advantage over Theorem \ref{thm:linear}
and at the same time total computation time rises much---this
is because of relative termination checking
invoked multiple times for solving each problem.
The applicability of Theorem \ref{thm:parallel} (Corollary \ref{cor:parallel})
which is subsumed by Theorem \ref{thm:PCP} (Corollary \ref{cor:PCP}, respectively)
is limited compared to these two criteria.
Comparing to Huet's criterion (Proposition \ref{prop:relative}) 
which succeeds at 37 examples, 
the direct application of the combination
of Theorems \ref{thm:linear} and \ref{thm:PCP}
without reduction-preserving completion
succeeded at 41 examples.

In each criterion, there is a clear increase
of the number of success by adding the completion procedure.
The increase of total time by 
the introduction of completion procedure is
not much but this depends on largely
our heuristics of choosing partitions
and the limitation on the number of completion steps.
The number of successful examples, however, does not change 
even when we increase that limit to 150 steps.
Activation of relative termination checking is also effective 
even in the presence of completion procedure.
The introduction of the \textit{Replacement} inference rule 
only makes a difference when relative termination checking is
not involved.

We have also tested the confluence prover \textsf{ACP} on our collection.
\textsf{ACP} is an automated confluence prover 
in which divide--and--conquer approach based on
the persistent, layer-preserving, commutative decompositions 
is employed and involving many confluence criteria 
\cite{
GOO98,
Hue80,
Oku98,OO97,
Toy81,
TO94,
TO01,
OO04,
Dev}
as well as the decreasing diagram techniques \cite{CRbyDD,converted}.
As shown in the table, 
most of our examples are not coped with by the confluence prover \textsf{ACP}.

The collection of examples and all details of the experiments 
are available 
on the webpage \texttt{http://www.nue.riec.tohoku.ac.jp/tools/acp/experiments/12lmcs/all.html}.


\begin{table}
\caption{Summary of experiments}
\label{table:summary}
\doublerulesep=.3pt
\begin{tabular}{r@{\,}l|c|c|c|c|r}\hline\hline
   &                               & \tiny success& \tiny  failure& \tiny  diverge & \tiny  timeout 
                                   & \multicolumn{1}{@{\,}c@{\,}}{\tiny time(sec.)}\\\hline
(1)& linear (w/o RT) (Cor.\ \ref{cor:linear})
& 28 & 57 & 0 & 0
& 7.28
\\
(2)& linear    (Thm.\ \ref{thm:linear})
& 32 & 53 & 0 & 0
& 7.41
\\
(3)& parallel (w/o RT) (Cor.\ \ref{cor:parallel})
& 19 & 66 & 0 & 0
& 7.93
\\
(4)& parallel (Thm.\ \ref{thm:parallel})
& 24 & 61 & 0 & 0
& 8.20
\\
(5)& PCP (w/o RT)  (Cor.\ \ref{cor:PCP})
& 28 & 57 & 0 & 0
& 7.97
\\
(6)& PCP (Thm.\ \ref{thm:PCP})
& 39 & 46 & 0 & 0
& 21.44
\\
(7)& linear\&PCP (w/o RT) (Cor.\ \ref{cor:linear}\&\ref{cor:PCP})
& 30 & 55 & 0 & 0
& 8.82
\\
(8)& linear\&PCP (Thm.\ \ref{thm:linear}\&\ref{thm:PCP})
& 41 & 44 & 0 & 0
& 21.28
\\
(9)& completion (linear,w/o RT)
& 47 & 38 & 0 & 0
& 8.17
\\
(10)& completion (linear)
& 61 & 24 & 0 & 0
& 9.00
\\
(11)& completion (PCP,w/o RT)
& 50 & 35 & 0 & 0
& 9.85
\\
(12)& completion (PCP)
& 74 & 11 & 0 & 0
& 30.18
\\
(13)& completion (linear\&PCP,w/o RT)
& 52 & 33 & 0 & 0
& 10.84
\\
(14)& completion (linear\&PCP)
& 77 & 8 & 0 & 0
& 28.26
\\
(15)& completion+repl (linear,w/o RT)
& 61 & 24 & 0 & 0
& 9.51
\\
(16)& completion+repl (linear)
& 61 & 24 & 0 & 0
& 9.09
\\
(17)& completion+repl (PCP,w/o RT)
& 66 & 19 & (3) & 0
& 14.90
\\
(18)& completion+repl (PCP)
& 75 & 10 & 0 & 0
& 29.49
\\
(19)& completion+repl (linear\&PCP,w/o RT)  
& 69 & 16 & (2) & 0
& 14.61
\\
(20)& completion+repl (linear\&PCP) 
& 77 & 8 & 0 & 0
& 29.05
\\
(21)& (20) + Huet (Prop.\ \ref{prop:relative})
& 77 & 8 & 0 & 0
& 40.95
\\\hline
\multicolumn{2}{l|}{Huet (Prop.\ \ref{prop:relative})} 
& 37 & 48 & 0 & 0
& 31.22
\\
\multicolumn{2}{l|}{\textsf{ACP} \cite{DD,ACP,YoshidaJSSST}} 
                                       & 13    & 72   & ---     &  (2)     & 196.40
\\\hline\hline\hline
\end{tabular}
\end{table}

\section{Conclusion}

We have presented a new method for proving 
confluence of TRSs.
The scope of our method is a class of TRSs
that can be partitioned into
a terminating part and a reversible part.
Our method is applicable for TRSs 
containing non-terminating rules
such as commutativity and associativity
which have been difficult to deal with
most of the standard methods.

We have given a new abstract criterion
for Church-Rosser modulo
(Theorem \ref{thm:ARS})
which extends those that appeared in \cite{Hue80} and \cite{JK86}.
Based on this abstract criterion,
we have given two new criteria for confluence of TRSs
formulated in terms of 
its terminating part $\mathcal{S}$
and 
its reversible part $\mathcal{P}$
(Theorems \ref{thm:linear} and \ref{thm:PCP}).
These criteria are effective 
even if $\mathcal{S}$ is not terminating
relative to $\mathcal{P}$;
in case $\mathcal{S}$ is terminating relative to a part of $\mathcal{P}$
then the applicability of criteria is strengthened.
We have also given 
a reduction-preserving completion
by which applicability of our criteria is enhanced.
We have implemented the proposed techniques
and reported experimental results.



\section*{Acknowledgment}

Thanks are due to Junichi Mitimata for 
discussions and experiments on preliminary
results of this paper.
The authors are grateful for
Harald Zankl, Aart Middeldorp,
Dominik Klein
for pointers to related works and/or helpful comments.
Thanks are due to anonymous referees 
for detailed comments.
This work was partially supported by grants from JSPS Nos.~20500002
and 22500002.

\bibliographystyle{plain}
\bibliography{crcomp}

\end{document}